\documentclass[12pt]{article}
\usepackage{setspace}
\usepackage[longnamesfirst,authoryear]{natbib}
\usepackage{amsthm,bbm}
\usepackage[T1]{fontenc}
\usepackage{geometry}
\usepackage{amssymb}
\usepackage{amsfonts,bm}
\usepackage{mathtools}
\usepackage{tabularx}
\usepackage{makeidx}
\usepackage{graphicx,eepic}
\usepackage{amsmath}
\usepackage{enumitem}
\DeclareMathOperator*{\argmax}{arg\,max}

\usepackage{caption}
\usepackage{subcaption}
\usepackage{url}
\usepackage{here}
\usepackage{rotating}
\usepackage{pdfpages}
\usepackage{comment,ulem,verbatim}
\usepackage{enumitem}  
\usepackage{multirow}
\usepackage{array,booktabs}
\usepackage{adjustbox}
\usepackage{siunitx}
\usepackage{lineno}
\usepackage{lipsum}
\usepackage{float}
\usepackage{titlesec}
\usepackage[title]{appendix}

\definecolor{blue(pigment)}{rgb}{0.2, 0.2, 0.6}

\usepackage{hyperref}
\hypersetup
{
	colorlinks=true,
	linkcolor=blue(pigment),
	urlcolor=blue(pigment),
	filecolor=blue(pigment),
	citecolor=blue(pigment)
}

\newtheorem{theorem}{Theorem}
\newtheorem{assumption}{Assumption}

\newtheorem{step}{Step}

\newtheorem{corollary}{Corollary}

\newtheorem{lemma}{Lemma}

\newtheorem{proposition}{Proposition}

\theoremstyle{definition}
\newtheorem{remark}{Remark}
\newtheorem{example}{Example}
\newtheorem*{example*}{Example}

\newtheorem{property}{Property}

\let\counterwithin\relax
\usepackage{chngcntr}
\usepackage{apptools}
\AtAppendix{\counterwithin{lemma}{section}}
\AtAppendix{\counterwithin{corollary}{section}}
\AtAppendix{\counterwithin{assumption}{section}}
\AtAppendix{\counterwithin{example}{section}}
\AtAppendix{\counterwithin{proposition}{section}}
\AtAppendix{\counterwithin{theorem}{section}}

\usepackage{multibib}
\newcites{appendix}{References}

\DeclareRobustCommand{\Id}{\mathbf 1}

\DeclareRobustCommand{\Var}{\mathrm{Var}}

\DeclareRobustCommand{\indep}{\mathop{\perp\!\!\!\!\perp}}
\begin{document}

\title{\vspace{-2em}\Large{Optimal Decision Rules Under Partial Identification}}
\author{Kohei Yata\thanks{Department of Economics, University of Wisconsin--Madison.
email: \href{mailto:yata@wisc.edu}{yata@wisc.edu}.
	I am grateful to my thesis advisors, Yuichi Kitamura, Timothy Armstrong, and Yusuke Narita, for their invaluable advice and encouragement.
	I also thank Donald Andrews, Jack Porter, and seminar and conference participants at the Applied Microeconometrics Conference in Honor of Professor Hidehiko Ichimura, Bravo/JEA/SNSF Workshop, Bristol, BU, CMStatistics 2022, Cornell, CREPE Day 2023, GSE-OSIPP-ISER Joint Conference in Economics 2024, Happy Hour Seminar, IAAE 2023, Michigan State, Notre Dame, Penn State, TSE, UC Irvine, UCL, UW--Madison, Western, and Yale for helpful comments and suggestions.
}
}
\date{\today}
\maketitle
\vspace{-.75cm}
\begin{abstract}
	
        I consider a class of statistical decision problems in which the policymaker must decide between two policies to maximize social welfare (e.g., the population mean of an outcome) based on a finite sample.
        The framework introduced in this paper allows for various types of restrictions on the structural parameter (e.g., the smoothness of a conditional mean potential outcome function) and accommodates settings with partial identification of social welfare.
	As the main theoretical result, I derive a finite-sample optimal decision rule under the minimax regret criterion.
        This rule has a simple form, yet achieves optimality among all decision rules; no ad hoc restrictions are imposed on the class of decision rules.
        I apply my results to the problem of whether to change an eligibility cutoff in a regression discontinuity setup, and illustrate them in an empirical application to a school construction program in Burkina Faso.
    ~

\end{abstract}


\sloppy

\section{Introduction}

    A fundamental goal of empirical research in economics is to inform policy decisions. Evaluation of counterfactual policies often requires extrapolating from observables to unobservables. Without strong model restrictions such as functional form assumptions,
    the performance of each counterfactual policy may be only partially determined by observed data. In such situations, policy decision-making is challenging, since we have no clear understanding of which policy is the best.
    For example, a regression discontinuity (RD) design only credibly estimates the impact of treatment on individuals at the eligibility cutoff. Therefore, without restrictive assumptions such as constant treatment effects, whether to offer the treatment to those away from the cutoff is ambiguous.
    Even randomized controlled trials may provide only partial knowledge of the impact of a new intervention, as can happen if
    the experimental sample is an unrepresentative subset of the target population.

    This paper studies the problem of using data to make policy decisions in settings in which social welfare under each policy may be only partially identified.
    Following the literature on statistical treatment choice \citep{Manski2004hetero}, I formulate the policy decision problem as a statistical decision problem.
    The framework introduced in this paper allows for various types of restrictions on the structural parameter, which potentially leads to partial identification of social welfare.
    It builds on \citeauthor{donoho1994}'s \citeyearpar{donoho1994} framework for optimal estimation in nonparametric regression models, which has recently been applied to estimation and inference on treatment-effect parameters \citep[e.g.,][]{armstrong2018optimal,Imbens2019RDD,kwon2020rd,Armstrong2021ATE,rambachan2023parallel,Chaisemartin2021aet}.
    I extend the framework to study optimal policy choice in a wide range of empirical settings with partial identification.
    Examples of this paper's framework include treatment choice using experiments with imperfect internal or external validity \citep{Stoye2012minimax,ishihara2021meta}; treatment choice using observational data under unconfoundedness with imperfect overlap; and policy adoption choice in difference-in-differences designs without exact parallel trends.
    
    Specifically, in the setup described in Section \ref{section:setup},
    the policymaker must decide between two policies, policy 1 and policy 0, to maximize social welfare.
    The difference in welfare between the two policies is given by $L(\theta)$,
    where $\theta$ is a possibly infinite-dimensional structural parameter that resides in a vector space $\mathbb{V}$, and 
    $L:\mathbb{V}\rightarrow\mathbb{R}$ is a known linear function.
    If $\theta$ is known, it is optimal to choose policy 1 if $L(\theta)\ge 0$ and policy 0 if $L(\theta)<0$.    
    The policymaker does not know $\theta$, but instead has access to a multivariate Gaussian sample $\boldsymbol Y=(Y_1,...,Y_n)'\in\mathbb{R}^n$ of the form
    \begin{align*}
	\boldsymbol{Y} \sim {\cal N}(\boldsymbol{m}(\theta), \boldsymbol\Sigma),
\end{align*}
    where $\boldsymbol{m}:\mathbb{V}\rightarrow\mathbb{R}^n$ is a known linear function and $\boldsymbol\Sigma$ is known.
    After observing $\boldsymbol{Y}$, the policymaker decides between policies 1 and 0.
    The main structural assumption is that $\theta$ belongs to a known set $\Theta\subset\mathbb{V}$ that is convex and centrosymmetric (i.e., $\theta\in\Theta$ implies $-\theta\in\Theta$), which encodes the policymaker's a priori knowledge of parameter restrictions.
    Depending on the restrictions, the welfare contrast $L(\theta)$ may or may not be point identified from the knowledge of the point-identified reduced-form parameter $\boldsymbol m(\theta)$.
    
    As detailed in Section \ref{section:existing}, an example of this setup is the choice between assigning treatment to everyone in the population (policy 1) and assigning treatment to no one (policy 0).
    Suppose that the policymaker has access to data from a regression model $Y_i=f(x_i,d_i)+U_i$, where $f(x,d)$ represents the conditional mean potential outcome under treatment $d\in\{0,1\}$ given covariates $x$, and $\{(x_i,d_i)\}_{i=1}^n$ is treated as fixed.
    Suppose further that treating everyone is preferred to treating no one if the population average treatment effect is positive.
    This problem is a special case when we assume $(U_1,...,U_n)'\sim {\cal N}(\boldsymbol{0}, \boldsymbol\Sigma)$ and set $\boldsymbol Y=(Y_1,...,Y_n)'$, $\theta=f$, $\boldsymbol{m}(f)=(f(x_1,d_1),...,f(x_n,d_n))'$, and $L(f)=\int[f(x,1)-f(x,0)]dP_X$, where $P_X$ is the population distribution of the covariates and is assumed to be known.
     The parameter space $\Theta$ is a class of conditional mean potential outcome functions $f$ that satisfy, for example, some smoothness restrictions (e.g., bounds on derivatives or the linearity of a function), which leads to either point or partial identification of the average treatment effect $L(f)$.
    The main contribution of this paper is to obtain a finite-sample optimal decision rule under the minimax regret criterion, which is a standard criterion used in the literature on statistical treatment choice \citep[e.g.,][]{Manski2004hetero,Manski2007missing,Stoye2009minimax,Stoye2012minimax,Kitagawa2018EWM}.
    A decision rule is a mapping from the sample $\boldsymbol Y$ to the probability of choosing policy 1.
    The minimax regret criterion evaluates the performance of a decision rule based on its worst-case expected welfare loss, or {\it regret}, relative to the oracle welfare-maximizing policy $\mathbf{1}\{L(\theta)\ge 0\}$, where the worst-case scenario is considered over the parameter space $\Theta$.
    I derive a decision rule that minimizes the worst-case regret among all decision rules, with no functional-form restrictions imposed on the class of rules.
    When $\boldsymbol Y$ is non-Gaussian and/or its variance is unknown, a feasible version of this decision rule can be constructed by plugging in an estimated variance.
    Appendix \ref{appendix:asymptotics} provides conditions under which its maximum regret over a class of distributions of $\boldsymbol Y$ converges to that of a minimax regret rule as $n\rightarrow\infty$.

    To solve the minimax regret problem, I use the {\it hardest one-dimensional subfamily} argument, which \cite{donoho1994} used to solve minimax affine estimation problems.
    The key idea is to search for the hardest one-dimensional subproblem---specifically, the one with the highest minimax risk among all subproblems with parameter spaces defined by one-dimensional linear subfamilies of the original parameter space $\Theta$.
    Then, verify that a minimax rule for the hardest one-dimensional subproblem is also minimax optimal for the original problem.
    Applying this strategy to the minimax regret problem is challenging due to
    the difference in the structure of the risk functions: Unlike standard risk functions for estimation, such as mean squared error (MSE), the regret cannot be decomposed into the bias and variance; instead, it can be decomposed into the probability of misidentifying the best policy and the potential welfare loss due to misidentification.
    To derive a minimax regret rule, I first characterize the hardest one-dimensional subproblem by optimizing a certain measure of the strength of the signal with respect to the best policy within subproblems (Lemma \ref{lemma:hardest-1d}).
    The subproblem with the optimal level of the signal strength achieves the best balance between the probability of misidentification and the potential welfare loss.
    I then propose a specific minimax regret rule for the hardest subproblem, and prove its minimax regret optimality for the original problem (Theorem \ref{theorem:main}).

    The results of this paper provide novel insights into how a minimax regret rule uses data to make decisions.
    First, the derived rule depends on the observations $Y_1,...,Y_n$ only through their weighted sum, $\sum_{i=1}^n{w}_iY_i$, although no such restrictions are imposed a priori.
    The weights ${w}_1,...,{w}_n$ can be calculated by solving a sequence of convex optimization problems, which is computationally and analytically tractable in leading examples.\footnote{Thus, this paper addresses a challenge in the application of statistical decision theory raised by \cite{Manski2020decision}, who wrote (p. 2848): ``The primary challenge to use of statistical decision theory is computational \ldots\ Future advances should continue to expand the scope of applications.''}
    
    Second, the minimax regret rule is nonrandomized or randomized, depending on the strength of the parameter restrictions and the variance of $\boldsymbol{Y}$.
    Specifically, if the restrictions are strong or the variance of $\boldsymbol{Y}$ is large in a certain formal sense, the minimax regret rule is a nonrandomized threshold rule of the form $\mathbf{1}\left\{\sum_{i=1}^n{w}_iY_i\ge 0\right\}$.
    On the other hand, if the restrictions are weak or the variance of $\boldsymbol{Y}$ is small, the rule is a randomized threshold rule of the form $\mathbf{1}\left\{\sum_{i=1}^n{w}_iY_i+\xi\ge 0\right\}$, where $\xi$ is generated independent of $\boldsymbol{Y}$ according to a certain distribution.
    In the latter case, randomization plays the role of reducing the probability of misidentification under worst-case parameter values, and thereby leads to a reduction in worst-case regret.
    This result generalizes \citeauthor{Stoye2012minimax}'s \citeyearpar{Stoye2012minimax} from a specific univariate problem to a general class of multivariate problems.
    
    Third, this paper sheds light on the connection between minimax regret treatment choice and minimax estimation.
    The weighted sum $\sum_{i=1}^n{w}_iY_i$ used by the minimax regret rule can be viewed as an estimator of the welfare contrast $L(\theta)$.
    In other words, the minimax regret rule can be viewed as a {\it plug-in} rule, which plugs the estimator $\sum_{i=1}^n{w}_iY_i$ (plus a random noise $\xi$ for the randomized rule) into the oracle optimal decision $\mathbf{1}\{L(\theta)\ge 0\}$.
    I show that this estimator is optimal in the sense of minimizing the worst-case squared bias (over $\Theta$) among all estimators of $L(\theta)$ subject to a certain bound on the variance (Theorem \ref{theorem:minimax-bias}).
    Furthermore, I show that this estimator places more importance on bias than variance compared with the minimax affine MSE estimator in \cite{donoho1994} (Theorem \ref{theorem:donoho}).

    Fourth, while this paper's main focus is on optimal rules under partial identification, my results are novel even under point identification for problems with restricted parameter spaces.
    When the welfare contrast $L(\theta)$ is point identified, the minimax regret rule is shown to always be a nonrandomized threshold rule, with its form depending on the strength and type of restrictions.
    For example, consider a linear regression model in which $\boldsymbol{m}(\theta)=\boldsymbol{X}\theta$ for some fixed $n\times k$ design matrix $\boldsymbol{X}$, and suppose $\Theta=\{\theta\in\mathbb{R}^k: \|\theta\|_p\le C\}$ for some known constants $C\ge 0$ and $p\ge 1$, where $\|\cdot\|_p$ denotes the $L_p$-norm.
    The results of this paper imply that the minimax regret rule makes decisions based on the sign of $L(\hat\theta)$, where $\hat\theta$ is an estimator for $\theta$ that resolves the bias-variance tradeoff in a certain way (e.g., the ridge estimator with the regularization parameter depending on $C$ for $p=2$).
    The result of \cite{hirano2009asymptotics} applies to problems with no parameter restrictions (i.e., $\Theta=\mathbb{R}^k$), but does not apply to ones with restricted parameter spaces.

    I demonstrate the practical relevance of my framework through an application to the problem of eligibility cutoff choice.
    In Section \ref{section:app}, I consider a situation in which the eligibility for treatment (e.g., social, educational, or welfare programs) is determined based on whether the value of an individual's characteristic exceeds a certain cutoff, as in RD setups.
    The policymaker wants to change the cutoff to a specific new value if the welfare effect of the cutoff change is positive. Here, the welfare effect is defined as the average treatment effect across units whose treatment status would be changed under the new cutoff.
    In this context, a decision rule maps the data collected under the status quo cutoff to the probability of changing the cutoff.
    My results can be used to find an optimal decision rule for this problem, under various restrictions on the conditional mean potential outcome function that enable extrapolation from one side of the cutoff to the other.
    For illustration, I assume that the function satisfies Lipschitz continuity with a known Lipschitz constant (i.e., a bound on the first derivatives), which leads to partial identification of the welfare effect.
    Applying my general results, I show that the minimax regret rule makes a decision based on the difference between a weighted average of observed outcomes for the treated units and that for the untreated units, with the weight vector determined by the choice of the Lipschitz constant.
    The rule can easily be computed by solving finite-dimensional convex optimization problems.

	Finally, in Section \ref{section:empirical}, I apply this rule to the Burkinab\'{e} Response to Improve Girls' Chances to Succeed (BRIGHT) program, a school construction program in Burkina Faso \citep{Kazianga2013bright}.
	With the aim of improving educational outcomes in rural villages, the program constructed primary schools in 132 villages from 2005 to 2008.
	To allocate schools, the Ministry of Education first computed a score that summarized village characteristics for each of the nominated 293 villages, then selected the highest-ranking villages to receive a school.
        Consider a policymaker who uses data collected after the completion of this program to decide whether to scale it up.
        As a hypothetical policy question, I consider whether to construct schools in the top 20\% of previously ineligible villages, using the enrollment rate as the welfare measure.
        I impose the Lipchitz constraint on the counterfactual enrollment rates across villages.
	To consider policy costs, I assume that it is optimal to scale up the program if its cost-effectiveness is better than that of a similar policy.
        Applying my theoretical results,
        I find that the minimax regret rule nonrandomly decides not to scale up the program for a plausible range of the Lipschitz constant.

\paragraph{Related Literature.}
This paper contributes to the literature on minimax regret statistical treatment choice under point identification \citep[e.g.,][]{hirano2009asymptotics,Stoye2009minimax,Stoye2012minimax,tetenov2012asymmetric} and partial identification \citep[e.g.,][]{Manski2007missing,Stoye2012minimax}.
\cite{Stoye2012minimax} derives a minimax regret rule in settings in which the experiment has imperfect validity, considering both Bernoulli and Gaussian models.
My result generalizes \citeauthor{Stoye2012minimax}'s \citeyearpar{Stoye2012minimax} in Gaussian models (Proposition 7(iii)) by allowing for multivariate samples, three- or higher-dimensional parameters, and various forms of parameter spaces.
Recently, \cite{ishihara2021meta} consider the problem of deciding whether to introduce a new policy based on results from multiple external studies.
They restrict attention to the class of nonrandomized threshold rules based on a weighted average of the sample and propose a way to numerically minimize the maximum regret.
In contrast, I do not impose any restrictions on decision rules, and use the hardest one-dimensional subfamily argument to analytically derive a minimax regret rule. My approach does not involve numerically minimizing the maximum regret and, for some problems, offers a closed-form expression for a minimax regret rule.
Since the initial version of this paper was circulated, there have been some advances in the literature.
\cite{kitagawa2023partial,kitagawa2022nonlinear} derive minimax fractional rules under squared welfare regret loss for both point and partial identification settings.
\cite{olea2023partial} point out the nonuniqueness of minimax regret rules for problems in which my minimax regret rule is randomized, and propose the least randomizing rule.
Their work and mine are complementary; their analysis relies on the existence of a minimax regret rule based on a weighted sum of the sample, while I prove the existence of such a rule and derive its formula.


Broadly, this paper contributes to the literature on treatment choice and policy learning under partial identification, which has been growing in econometrics and statistics \citep[e.g.,][]{Manski2000ambiguity,manski2009diversified,Manski2010vaccine,manski2011are,manski2011partial,Manski2020decision,chamberlain2011bayesian,kasy2018taxation,russell2020policy,Mo2020robust,Kallus2020confounding,dadamo2022orthogonal,ben-michael2022safe,adjaho2023,christensen2022discrete,Kido2023locally}.\footnote{An extensive literature examines the problem of learning treatment allocation policies that map an individual's covariates to a treatment. See \cite{Manski2004hetero,Dehejia2005decision,Stoye2009minimax,Stoye2012minimax,qian2011ind,Bhattacharya2012budget,Kitagawa2018EWM,kitagawa2021equal,Athey2021policy}; and \cite{Mbakop2021penalized}, among others, for point identification settings.
My approach can be applied to partial identification settings in which the choice set consists of two treatment assignment policies.
}
Many recent studies, either implicitly or explicitly, consider the worst-case welfare (or welfare loss) over the identified set given the point-identified parameter as the loss function of a statistical decision problem.
By contrast, this paper directly uses the welfare loss as the loss function without taking its worst-case value given the point-identified parameter, following the standard minimax regret criterion.
Furthermore, they derive finite-sample bounds on the expected loss, its convergence rates, or the local asymptotic optimality of their proposed rules, while this paper derives a finite-sample exact optimal rule.
Another difference, which is empirically relevant, is that many studies consider welfare parameters whose upper and lower bounds can be estimated at a parametric rate. In contrast, this paper covers welfare parameters whose bounds cannot be estimated at a parametric rate, such as the average treatment effect at a point of the running variable in a nonparametric RD setup.

The problem of eligibility cutoff choice considered in this paper is related to the literature on extrapolation away from the cutoff in RD designs, including
\cite{rokkanen2015rd,angrist2015rd,dong2015rd,Bertanha2020rd,Bertanha2020many,Bennett2020rd}; and \cite{Cattaneo2020multi}.
Unlike these papers, I explicitly consider the decision problem of whether to change the cutoff and derive an optimal decision rule.
Recently, \cite{zhang2022rd} consider the problem of learning cutoff-based policies under multi-cutoff designs and propose a maximin policy that is guaranteed to perform no worse than the existing policy.

\section{Problem Setup and Examples}\label{section:setup}

In this section, I introduce my framework and provide examples to illustrate it.
Section \ref{subsection:dgp} describes the statistical model that generates the data available to the policymaker. Section \ref{subsection:choice} defines the policymaker's action set and the associated social welfare functions.
Section \ref{section:optimality} introduces the minimax regret criterion as an optimality criterion for decision rules.
Section \ref{subsection:setup-relation} provides a brief discussion of the relationship between this framework and existing ones.
Finally, Section \ref{section:existing} illustrates the framework using two examples.




\subsection{Data-generating Model}\label{subsection:dgp}
Suppose that the policymaker observes a sample $\boldsymbol{Y}=(Y_1,...,Y_n)'\in \mathbb{R}^n$ of the form
\begin{align}
	\boldsymbol{Y} \sim {\cal N}(\boldsymbol{m}(\theta), \boldsymbol\Sigma),\label{model}
\end{align}
where $\theta$ is an unknown parameter that lies in a known subset $\Theta$ of a vector space $\mathbb{V}$; $\boldsymbol{m}:\mathbb{V}\rightarrow \mathbb{R}^n$ is a known linear function; and $\boldsymbol\Sigma$ is a known, positive-definite $n\times n$ matrix.
I allow $\theta$ to be an infinite-dimensional parameter such as a function.

The linearity of $\boldsymbol{m}$ is not necessarily restrictive. If we specify $\theta$ so that it contains each of the expected values of $Y_1,...,Y_n$ as its element, $\boldsymbol m$ is a function that extracts those expected values from $\theta$, which is linear in $\theta$.

This model allows the expected value of $\boldsymbol{Y}$ to depend on other observed variables such as covariates and treatment by treating them as fixed and subsuming them into $\boldsymbol m$ and $\boldsymbol \Sigma$.
For example, a regression model with fixed regressors
\begin{align*}
	Y_i=f(x_i)+U_i,~~~U_i\sim {\cal N}(0,\sigma^2(x_i)) ~\text{independent across $i$}
\end{align*}
is a special case in which $\boldsymbol{Y}=(Y_1,...,Y_n)'$, $\theta=f$, $\Theta$ is a class of functions, $\boldsymbol m(f)=(f(x_1),...,f(x_n))'$, and $\boldsymbol \Sigma={\rm diag}(\sigma^2(x_1),...,\sigma^2(x_n))$.

The normality of $\boldsymbol Y$ and the assumption of known variance are restrictive, but are often imposed to deliver finite-sample optimality results for statistical decision problems.
In some cases, the normal model is motivated as an approximation to a finite-sample problem.
Suppose that we observe an $n$-dimensional vector of statistics derived from the original data, which is an asymptotically normal estimator of its population counterpart.
For example, the mean outcome difference between the treatment and control groups in a randomized experiment is a statistic that is asymptotically normal for the population mean difference.
If we regard the $n$-dimensional vector of statistics as $\boldsymbol Y$, the normal model (\ref{model}) can be viewed as an asymptotic approximation.
Also, in Appendix \ref{appendix:asymptotics}, I consider an asymptotic framework in which the distribution of the error $\boldsymbol Y-\boldsymbol m(\theta)$ is unknown and the sample size $n$ goes to infinity. I propose a feasible decision rule and derive conditions under which its maximum regret over a class of distributions converges to that of a minimax regret rule as $n\rightarrow\infty$.

I assume that the parameter space $\Theta$ is convex and centrosymmetric (i.e., $\theta\in\Theta$ implies $-\theta\in\Theta$) throughout the paper.
Typical parameter spaces considered in empirical analyses are convex.
For example, in the regression model above, classes of functions with bounded derivatives are convex.
The centrosymmetry simplifies the minimax analysis; see Remark \ref{remark:centrosymmetry} in Section \ref{section:hardest} for the role of centrosymmetry. However, it rules out some shape restrictions.
In the regression model above, the class of convex (or concave) functions is noncentrosymmetric.

\subsection{Policy Choice Problem}\label{subsection:choice}
Now, suppose that the policymaker is interested in choosing between two policies, policy $1$ and policy $0$, to maximize social welfare.
Suppose that the welfare resulting from implementing policy $a\in\{0,1\}$ under $\theta$ is $W_a(\theta)$, where $W_a:\mathbb{V}\rightarrow \mathbb{R}$ is a known function specified by the policymaker.
The welfare contrast between policy $1$ and policy $0$ is given by
$$
L(\theta)\coloneqq W_1(\theta) - W_0(\theta).
$$
I assume that $L:\mathbb{V}\rightarrow \mathbb{R}$ is a linear function.
The optimal policy under $\theta$ is policy 1 if $L(\theta)>0$, policy 0 if $L(\theta)<0$, and either if $L(\theta)=0$.

One example of a welfare criterion is a weighted average of an outcome across individuals.
For example, suppose a policy could change the outcome of each individual.
Suppose also that we specify $\theta=(f_1(\cdot),f_0(\cdot))$, where $f_a(x)$ represents the counterfactual mean outcome under policy $a$ across individuals whose observed covariates are $x$.
The welfare under policy $a$ can be defined, for example, by the population mean outcome $W_a(\theta)=\int f_a(x)dP_X$, where $P_X$ is the probability measure of covariates and is assumed to be known.
In this case, the welfare contrast $L(\theta)=\int [f_1(x)-f_0(x)]dP_X$ is linear in $\theta$.
On the other hand, the linearity of $L$ may rule out welfare criteria that depend on the distribution of the counterfactual outcome other than the mean.
See \cite{kitagawa2021equal} for such welfare criteria.

Importantly, this framework allows for cases in which $L(\theta)$ is not point identified.
Let ${\cal M}\coloneqq\{\boldsymbol{m}(\theta):\theta\in\Theta\}\subset\mathbb{R}^n$ denote the set of possible values of the reduced-form parameter $\boldsymbol{m}(\theta)$.
The {\it identified set} of $L(\theta)$ when $\boldsymbol m(\theta)=\boldsymbol \mu\in \mathbb{R}^n$ is defined as
$$
I(\boldsymbol \mu)\coloneqq \{L(\theta):\boldsymbol m(\theta)=\boldsymbol \mu, \theta\in \Theta\}.
$$
$I(\boldsymbol \mu)$ may contain multiple elements for some or all $\boldsymbol \mu\in{\cal M}$.
If $I(\boldsymbol \mu)$ contains both positive and negative values, the superior policy is ambiguous even without sampling uncertainty.

\subsection{Optimality Criterion}\label{section:optimality}

This paper's goal is to provide an optimal decision rule for using data to make a policy decision.
A (randomized) {\it decision rule} is a measurable function $\delta:\mathbb{R}^n\rightarrow [0,1]$, where $\delta(\boldsymbol{y})$ represents the probability of choosing policy $1$ when the realization of the sample $\boldsymbol{Y}$ is $\boldsymbol y$.\footnote{In contexts in which \textit{fractional} treatment allocations, which assign treatment to a fraction of individuals in a population, are permitted, we can also interpret $\delta$ as a \textit{fractional} rule. Here, $\delta(\boldsymbol y)$ represents the fraction of individuals to whom we would assign treatment. If the welfare of treating a fraction $a\in [0,1]$ is defined as $W_a(\theta)=W_0(\theta)+a(W_1(\theta)-W_0(\theta))$, the regret of a fractional rule is equal to that of a randomized rule.}
I consider the minimax regret criterion as an optimality criterion for decision rules.
To introduce it, I first define
the {\it welfare regret loss} for policy choice $a\in\{0,1\}$ under $\theta$ as
\begin{align*}
	l(a,\theta)&\coloneqq \max_{a'\in\{0,1\}}W_{a'}(\theta)-W_a(\theta)=\begin{cases}
		L(\theta)\cdot(1-a)~~&\text{ if } L(\theta)\ge 0,\\
		-L(\theta)\cdot a~~&\text{ if } L(\theta)< 0.
	\end{cases}
\end{align*}
The welfare regret loss $l(a,\theta)$ is the difference in welfare between the optimal policy and policy $a$ under $\theta$.
If the policymaker chooses the superior policy, they do not incur any loss; otherwise, they incur a loss of the absolute value of the welfare contrast $L(\theta)$.

The {\it risk} or {\it regret} of decision rule $\delta$ under $\theta$ is the expected welfare regret loss
$$
R(\delta,\theta)\coloneqq\begin{cases}
L(\theta)(1-\mathbb{E}_\theta[\delta(\boldsymbol{Y})])~~&\text{ if } L(\theta)\ge 0,\\
-L(\theta)\mathbb{E}_\theta[\delta(\boldsymbol{Y})]~~&\text{ if } L(\theta)< 0,
\end{cases}
$$
where $\mathbb{E}_\theta$ denotes the expectation taken with respect to $\boldsymbol{Y}$ under $\theta$.
Using the regret as a performance measure allows one to consider not only the error probabilities ($1-\mathbb{E}_\theta[\delta(\boldsymbol{Y})]$ or $\mathbb{E}_\theta[\delta(\boldsymbol{Y})]$) but also the potential welfare loss ($|L(\theta)|$).

The regret of a decision rule can vary with $\theta$ over the parameter space $\Theta$. Generally, no rule uniformly dominates all other rules.
The minimax regret criterion aggregates the regret over $\Theta$ by considering the {\it maximum} or {\it worst-case regret}, defined as
$\sup_{\theta\in\Theta}R(\delta,\theta)$.
Let ${\cal R}(\Theta)$ denote the {\it minimax risk} or {\it minimax regret}, defined as
\begin{align}
{\cal R}(\Theta)\coloneqq\inf_{\delta\in{\cal D}}\sup_{\theta\in\Theta}R(\delta,\theta), \label{eq:minimax}
\end{align}
where ${\cal D}$ denotes the set of all decision rules.
Under the minimax regret criterion, we aim to derive a {\it minimax regret} decision rule $\delta^*$, which satisfies
$\sup_{\theta\in\Theta}R(\delta^*,\theta)={\cal R}(\Theta)$.

To further understand the minimax regret criterion, note that the above minimax problem can equivalently be written as
$\inf_{\delta\in{\cal D}}\sup_{\theta\in\Theta}\left(\max_{a\in\{0,1\}}W_a(\theta)-U(\delta,\theta)\right)$,
where $U(\delta,\theta)\coloneqq W_1(\theta)\mathbb{E}_\theta[\delta(\boldsymbol{Y})]+W_0(\theta)(1-\mathbb{E}_\theta[\delta(\boldsymbol{Y})])$ is the expected welfare of decision rule $\delta$ under $\theta$.
Thus, the minimax regret criterion optimizes the uniform closeness of the expected welfare to the maximum attainable welfare.
When the minimax risk ${\cal R}(\Theta)$ is small, a minimax regret rule uniformly achieves near-optimal expected welfare across all parameter values.

\sloppy

Alternative optimality criteria include the maximin criterion, which maximizes the worst-case expected welfare.
    Specifically, one aims to solve $\sup_{\delta\in{\cal D}}\inf_{\theta\in\Theta}U(\delta,\theta)$.
    It has been pointed out that the maximin criterion is unreasonably pessimistic and can lead to pathological decision rules \citep{Manski2004hetero,Stoye2009minimax}; see \cite{olea2023partial} for such a result in the setting of this paper.\footnote{For example, suppose the policymaker perfectly knows that the welfare of the status quo policy ($a=0$) is given by $W_0(\theta)=w_0$ for some constant $w_0$.
    If the welfare of the new policy ($a=1$) can be less than $w_0$ under at least one parameter value (i.e., $\inf_{\theta\in\Theta}W_1(\theta)<w_0$), then the decision rule that always maintains the status quo regardless of the data (i.e., $\delta(\boldsymbol{Y})=0$) is optimal under the maximin criterion.}
    By contrast, the minimax regret criterion optimizes the worst-case expected welfare relative to what is achievable at a given parameter value, which often leads to nontrivial decision rules.
    Another alternative is 
    the Bayes criterion, which optimizes the average risk over a prior.
    That is, one aims to solve $\sup_{\delta\in{\cal D}}\int U(\delta,\theta)d\pi(\theta)$ (or equivalently $\inf_{\delta\in{\cal D}}\int R(\delta,\theta)d\pi(\theta)$), where $\pi$ is a prior on $\Theta$. 
    See, for example, \cite{chamberlain2011bayesian} and \cite{kasy2018taxation} for Bayesian treatment choice.
    This paper focuses on the minimax regret criterion, following prior treatment choice studies \citep[e.g.,][]{Manski2004hetero,Manski2007missing,hirano2009asymptotics,Stoye2009minimax,Stoye2012minimax,tetenov2012asymmetric,Kitagawa2018EWM}. Studying optimal rules under the Bayes or other possible criteria is beyond the scope of this paper.

\subsection{Relation to Existing Frameworks}\label{subsection:setup-relation}    
    Here, I discuss the relationship between the above framework and existing ones. For minimax regret treatment choice, the framework in this paper generalizes the univariate Gaussian problems with two-dimensional parameters in \cite{Stoye2012minimax} to accommodate multivariate samples, parameters of three or higher dimensions, and various types of parameter restrictions.
    It also generalizes the limiting version of minimax regret problems under parametric models studied by \cite{hirano2009asymptotics}
    to accommodate partially identified welfare contrasts and restricted parameter spaces.
    In the setting described in Section \ref{section:existing}, \cite{ishihara2021meta} derive a minimax regret rule within the class of nonrandomized threshold rules based on a weighted average of the sample, namely $\delta(\boldsymbol{Y})=\mathbf{1}\{\sum_{i=1}^nw_iY_i\ge 0\}$ with $\sum_{i=1}^nw_i=1$.
    Their result does not require the convexity of the parameter space, but instead assumes that the parameter space is invariant to the addition of vectors of ones, which excludes bounded parameter spaces.\footnote{A generalization of their invariance assumption to the general setup of this paper is as follows: There exists $\iota\in\Theta$ such that $L(\iota)=1$ and $\theta+c\iota\in\Theta$ for all $\theta\in\Theta$ and $c\in\mathbb{R}$. Under this condition and the centrosymmetry of $\Theta$, it is possible to extend their approach to derive a minimax regret rule among rules of the form $\delta(\boldsymbol{Y})=\mathbf{1}\{\boldsymbol{w}'\boldsymbol{Y}\ge 0\}$ with $\boldsymbol{w}\in\mathbb{R}^n$ and $\boldsymbol{w}'\boldsymbol{m}(\iota)=1$, in the general setup of this paper.}
    The setting of \cite{ishihara2021meta} with any convex parameter space, whether bounded or unbounded, is a special case of my setting.

    \cite{donoho1994} studies the optimal estimation of a linear functional of $\theta$ in a more general version of this paper's model, which allows for infinite-dimensional Gaussian models and noncentrosymmetric parameter spaces.
    \cite{donoho1994} derives minimax estimators and confidence intervals within the class of affine procedures, using squared error, absolute error, or the length of a fixed-length two-sided confidence interval as a loss function.
    Using the framework of \cite{donoho1994}, \cite{armstrong2018optimal} provide a one-sided confidence interval that minimizes the maximum $\beta$th quantile of the excess length among all one-sided confidence intervals of a given confidence level.
    In contrast to these studies, this paper focuses on a binary decision problem under welfare regret loss.
    Section \ref{section:relation} discusses the connection between minimax estimation and minimax regret treatment choice in detail.



\subsection{Examples}\label{section:existing}


I illustrate my framework using two examples.
The first is a special case of \citeauthor{ishihara2021meta}'s \citeyearpar{ishihara2021meta} setup, which I use as a running example to illustrate theoretical results in Section \ref{section:minimax}.
The second is policy choice using observational data under unconfoundedness.

\begin{example}[Evidence Aggregation \citep{ishihara2021meta}]\label{example:intersection}
    Consider a policymaker who is interested in deciding whether to introduce a new policy to a specific local population based on causal evidence of similar policies implemented in other populations.
    We observe an $n$-dimensional sample $\boldsymbol Y\sim {\cal N}(\boldsymbol m(\theta), \boldsymbol\Sigma)$,  where $\theta=(\theta_1,...,\theta_n,\theta_{n+1})'\in \Theta\subset \mathbb{R}^{n+1}$, $\boldsymbol m(\theta)=(\theta_1,...,\theta_n)'$, and $\boldsymbol{\Sigma}={\rm diag}(\sigma_1^2,...,\sigma_n^2)$.
    The welfare contrast is given by $L(\theta)=\theta_{n+1}$. Here, $\theta_{n+1}$ is the average welfare effect of a new policy on the target population; $\theta_1,...,\theta_n$ are the average welfare effects on $n$ study populations; and $Y_1,...,Y_n$ are estimators for $\theta_1,...,\theta_n$.
    $\Theta$ imposes restrictions on the differences between $\theta_i$'s, so that $\theta_{n+1}$ is point or partially identified from $\boldsymbol m(\theta)$.
    To illustrate my results in Section \ref{section:minimax}, I focus on a simple case in which $n=2$ and $\sigma_1^2=\sigma_2^2=\sigma^2$ for some $\sigma^2>0$. Also, I specify
    $$
    \Theta=\{\theta\in\mathbb{R}^3: |\theta_1-\theta_3|\le C_1, |\theta_2-\theta_3|\le C_2\}
    $$
    for some known constants $C_1,C_2\ge 0$.
    For $i=1,2$, $C_i$ reflects prior knowledge about how similar study population $i$ and the target population are in terms of the average welfare effect.
    In this example, ${\cal M}=\{\boldsymbol{m}(\theta):\theta\in\Theta\}=\{(\theta_1,\theta_2)'\in\mathbb{R}^2:|\theta_1-\theta_2|\le C_1+C_2\}$. The identified set of $L(\theta)$ when $\boldsymbol{m}(\theta)=\boldsymbol{\mu}\in {\cal M}$ is given by the following intersection bounds:
    $$
    I(\boldsymbol{\mu})=\{\theta_3:\boldsymbol{m}(\theta)=\boldsymbol{\mu},\theta\in\Theta\}=\left[\max\{\mu_1-C_1,\mu_2-C_2\},\min\{\mu_1+C_1,\mu_2+C_2\}\right].
    $$
    Note that the upper bound (and the lower bound) is not differentiable with respect to $\boldsymbol{\mu}$ if $\mu_1+C_1=\mu_2+C_2$ (and $\mu_1-C_1=\mu_2-C_2$, respectively).
    My framework covers problems with nondifferentiable upper and lower bounds on the welfare contrast.
    

\end{example}

    \begin{example}[Choice of Treatment Assignment Policy under Unconfoundedness]\label{example:unconfoundedness}
        Consider a policymaker interested in choosing who should be treated based on an individual’s observable covariates.
        Suppose each member $i$ of the population is characterized by potential outcomes $Y_i(1)$ and $Y_i(0)$ with and without treatment; a vector of covariates $X_i\in{\cal X}\subset\mathbb{R}^k$; and a treatment indicator $D_i\in\{0,1\}$.
    The policymaker observes a random sample $\{(Y_i,X_i,D_i)\}_{i=1}^n$, where $Y_i=Y_i(1)D_i+Y_i(0)(1-D_i)$ is the realized outcome.
    Let $f(x,d)=\mathbb{E}[Y_i(d)|X_i=x]$ and $\sigma^2(x,d)={\rm Var}(Y_i(d)|X_i=x)$ for $(x,d)\in {\cal X}\times\{0,1\}$.
    Assume that the unconfoundedness condition, $(Y_i(1),Y_i(0))\indep D_i|X_i$, holds, so that $f(x,d)=\mathbb{E}[Y_i|X_i=x,D_i=d]$ and $\sigma^2(x,d)={\rm Var}(Y_i|X_i=x,D_i=d)$ for $x\in{\cal X}_d$ and $d\in\{0,1\}$, where ${\cal X}_d$ denotes the support of $X_i$ conditional on $D_i=d$.
    On the other hand, I do not impose the overlap condition, $0<\mathbb{P}(D_i=1|X_i)<1$, so the conditional average treatment effect $f(x,1)-f(x,0)$ is generally not point identified without additional structure.
    I condition on the realized values $\{(x_i,d_i)\}_{i=1}^n$ of $\{(X_i,D_i)\}_{i=1}^n$ to obtain a regression model with fixed regressors
    $$
    Y_i=f(x_i,d_i)+U_i,
    $$
    where $U_i$ is independent across $i$, $\mathbb{E}[U_i]=0$, and ${\rm Var}(U_i)=\sigma^2(x_i,d_i)$.
    This model fits into my framework by assuming $U_i$ is normal and setting $\boldsymbol Y=(Y_1,...,Y_n)'$, $\theta=f$, $\boldsymbol m(f)=(f(x_1,d_1),...,f(x_n,d_n))'$, and $\boldsymbol \Sigma={\rm diag}(\sigma^2(x_1,d_1),...,\sigma^2(x_n,d_n))$.

    Now, suppose the policymaker must decide between two treatment assignment policies, $\pi_1$ and $\pi_0$, where each policy $\pi_a:{\cal X}\rightarrow[0,1]$, $a\in\{0,1\}$, specifies the probability of assigning treatment to individuals with covariates $x\in{\cal X}$.
    For example, $\pi_0$ may be the status quo policy that generates the treatment assignment of the units in the sample (i.e., $\pi_0(X_i)=\mathbb{P}(D_i=1|X_i)$), and $\pi_1$ may be a new policy; in this case, if $\pi_0$ is a deterministic policy (i.e., $\pi_0(x)\in\{0,1\}$ for all $x$), then there is no overlap between the supports of covariates for the treated and untreated groups in the sampling population (i.e., ${\cal X}_1\cap{\cal X}_0=\varnothing$).
    Alternatively, $\pi_0$ may assign treatment to no one (i.e., $\pi_0(x)=0$ for all $x$), and $\pi_1$ may assign treatment to everyone (i.e., $\pi_1(x)=1$ for all $x$).
    Suppose the welfare under policy $a\in\{0,1\}$ is an average of the conditional mean potential outcome across different values of covariates
	$$
	W_a(f) = \int [f(x,1)\pi_a(x)+f(x,0)(1-\pi_a(x))]d\nu(x)
	$$
	for some known measure $\nu$.
	The welfare contrast between the two policies is
	$$
	L(f)=W_1(f)-W_0(f)=\int (\pi_1(x)-\pi_0(x))[f(x,1)-f(x,0)]d\nu(x).
	$$

	To point or partially identify $L(f)$ under imperfect overlap, we need to impose restrictions on $f$. Suppose that $f\in{\cal F}$, where ${\cal F}$ is a known convex and centrosymmetric function class and plays the role of the parameter space $\Theta$.
    Possible function classes include the class of functions with a known bound on derivatives.
    In Section \ref{section:app}, I consider a special case of this setting, in which $x$ is scalar, $\pi_0$ is the status quo cutoff-based policy that generates the data, and $\pi_1$ is a new cutoff-based policy.
    I derive a minimax regret rule under the assumption that $f$ belongs to the Lipschitz class with a known Lipschitz constant.
    \end{example}
	
\section{Minimax Regret Rules in General Setup}\label{section:minimax}

    In this section, I solve the minimax regret problem by using the hardest one-dimensional subfamily argument, which \cite{donoho1994} used to solve minimax affine estimation problems.
    This approach consists of three steps.
    The first step is to solve one-dimensional subproblems, in which the parameter space is restricted to a one-dimensional linear bounded subfamily.
    The second step is to search for the hardest one-dimensional subproblem, defined as the one with the highest minimax risk.
    The final step is to show that a minimax rule for the hardest one-dimensional subproblem is also minimax optimal for the original problem.
    
    A key distinction between minimax regret treatment choice and minimax affine estimation lies in the structure of their risk functions.
    For estimation, standard risk functions such as mean squared error (MSE) can be decomposed into bias and variance. In contrast, the regret can be decomposed into the error probability and the potential welfare loss.
    Consequently, substantially different arguments are required for each of the above three steps.


I normalize $\boldsymbol\Sigma=\sigma ^2\boldsymbol I_n$ for some $\sigma>0$ throughout this section, where $\boldsymbol I_n$ is the identity matrix.
This normalization is without loss of generality, since $\boldsymbol \Sigma$ is known.\footnote{Specifically, let $\tilde{\boldsymbol Y}=\boldsymbol{\Sigma}^{-1/2}\boldsymbol{Y}$ and $\tilde{\boldsymbol{m}}(\theta)=\boldsymbol{\Sigma}^{-1/2}\boldsymbol{m}(\theta)$ so that 
$\tilde{\boldsymbol Y}\sim {\cal N}(\tilde{\boldsymbol{m}}(\theta), \boldsymbol I_n)$.
For any rule $\delta(\boldsymbol{Y})$, its regret in the problem with $(L,\boldsymbol{m},\Theta, \boldsymbol{\Sigma})$ is the same as the regret of the rule $\tilde \delta(\tilde{\boldsymbol{Y}})$ in the problem with $(L,\tilde{\boldsymbol{m}},\Theta, \boldsymbol{I}_n)$, where $\tilde \delta(\tilde{\boldsymbol{Y}})=\delta(\boldsymbol{\Sigma}^{1/2}\tilde{\boldsymbol{Y}})=\delta(\boldsymbol{Y})$.}
I use the following assumption to derive a minimax regret rule.

\begin{assumption}\label{assumption:problem}
    (i) $L:\mathbb{V}\rightarrow\mathbb{R}$ and $\boldsymbol m:\mathbb{V}\rightarrow\mathbb{R}^n$ are linear; (ii) $\Theta$ is a nonempty, convex, and centrosymmetric subset of $\mathbb{V}$; (iii) $L(\theta)\neq 0$ for at least one $\theta\in\Theta$; (iv) $\sup I(\boldsymbol{0})<\infty$.
\end{assumption}

The first two conditions are introduced in Section \ref{section:setup}.
It is straightforward to see that ${\cal M}=\{\boldsymbol{m}(\theta):\theta\in\Theta\}$ is a nonempty, convex, and centrosymmetric subset of $\mathbb{R}^n$ under these two conditions.
These conditions also imply the following relationship between the lower and upper bounds on the welfare contrast: $\sup I(\boldsymbol{\mu})=-\inf I(-\boldsymbol{\mu})$ for all $\boldsymbol{\mu}\in {\cal M}$.
By this symmetry, it is sufficient to focus on $\sup I(\boldsymbol{\mu})$ in the analysis below.
The third condition excludes the trivial case in which $L(\theta)=0$ for any $\theta\in\Theta$ and hence the worst-case regret of any decision rule is zero.
The fourth condition is also necessary to obtain nontrivial results: If $\sup I(\boldsymbol{0})=\infty$, the regret of any decision rule is unbounded on $\{\theta\in\Theta:\boldsymbol m(\theta)=\boldsymbol 0\}$, and hence the worst-case regret of any rule is infinity.

    In the following, I first solve one-dimensional subproblems in Section \ref{section:one-dim}.
    Next, I characterize the hardest one-dimensional subproblem in Section \ref{section:hardest} and present a minimax regret rule for the original problem in Section \ref{section:mmr}.

\subsection{Minimax Regret Rules for One-dimensional Subproblems}\label{section:one-dim}

First, I consider one-dimensional subproblems.
To define a one-dimensional subproblem, take any $\bar\theta\in \Theta$ such that $L(\bar\theta)\ge 0$.
A {\it one-dimensional subfamily}, denoted by $[-\bar\theta,\bar\theta]$, is defined as the set of all convex combinations of $\bar\theta$ and $-\bar\theta$:
$$
[-\bar\theta,\bar\theta]\coloneqq\{\theta\in \mathbb{V}:\theta=\lambda\bar\theta,\lambda\in [-1,1]\}.
$$
$[-\bar\theta,\bar\theta]$ is a subset of $\Theta$, since $\Theta$ is convex and centrosymmetric.
Given $\bar\theta$, $[-\bar\theta,\bar\theta]$ can be viewed as a one-dimensional parameter space with a scalar parameter $\lambda\in[-1,1]$.
A {\it one-dimensional subproblem} is the problem of finding a minimax regret rule for $[-\bar\theta,\bar\theta]$, whose maximum regret over $[-\bar\theta,\bar\theta]$ equals ${\cal R}([-\bar\theta,\bar\theta])$, where ${\cal R}([-\bar\theta,\bar\theta])=\inf_{\delta\in{\cal D}}\sup_{\theta\in [-\bar\theta,\bar\theta]}R(\delta,\theta)$.

The following result derives minimax regret rules for one-dimensional subproblems.
Let $\|\cdot\|$ denote the Euclidean norm and $\Phi$ denote the cumulative distribution function of a standard normal random variable.

\begin{lemma}[Minimax Regret Rules for One-dimensional Subproblems]\label{lemma:one-dim0}
    Suppose Assumption \ref{assumption:problem} holds, and consider a one-dimensional subproblem for $[-\bar\theta,\bar\theta]$, where $\bar \theta\in\Theta$ and $L(\bar\theta)\ge 0$.
    Then, the following holds.
    \begin{enumerate}[label=(\roman*)]
        \item If $\boldsymbol{m}(\bar\theta)\neq \boldsymbol{0}$,
	then the decision rule
	$
	\delta^*(\boldsymbol{Y})=\mathbf{1}\left\{\boldsymbol{m}(\bar \theta)'\boldsymbol{Y}\ge 0\right\}
	$
	is minimax regret.
    \item If $\boldsymbol{m}(\bar\theta)=\boldsymbol{0}$, then any decision rule $\delta^*$ such that
	$
	\mathbb{E}[\delta^*(\boldsymbol{Y})]=1/2,
	$
    where $\boldsymbol Y\sim {\cal N}(\boldsymbol 0,\sigma^2 \boldsymbol I_n)$,
	is minimax regret.
	\item \label{lemma:one-dim0:mmrisk} The minimax risk is given by
	\begin{align*}
{\cal R}([-\bar \theta,\bar \theta])
=\begin{cases}
L(\bar\theta)\Phi\left(-\frac{\|\boldsymbol m(\bar\theta)\|}{\sigma}\right) ~~ &\text{ if } \|\boldsymbol m(\bar\theta)\|\le \tau^*\sigma,\\
\tau^*\sigma \frac{L(\bar\theta)}{\|\boldsymbol m(\bar\theta)\|}\Phi\left(-\tau^*\right) &\text{ if } \|\boldsymbol m(\bar\theta)\|> \tau^*\sigma,
\end{cases}
\end{align*}
where $\tau^*\in \arg\max_{t\ge 0}t\Phi(-t)$, which is unique ($\tau^*\approx 0.752$).
\end{enumerate}
\end{lemma}
\begin{proof}
	See Appendix \ref{proof:lemma:one-dim}.
\end{proof}

Lemma \ref{lemma:one-dim0} provides minimax regret rules for subproblem $[-\bar \theta,\bar \theta]$ separately for the following two cases: (i) $\boldsymbol{m}(\bar\theta)\neq \boldsymbol{0}$ and (ii) $\boldsymbol{m}(\bar\theta)= \boldsymbol{0}$.
In case (i), a linear threshold rule based on $\boldsymbol{m}(\bar \theta)'\boldsymbol{Y}$ is minimax regret. On the other hand, in case (ii), any decision rule that chooses each policy with probability one-half over the distribution of $\boldsymbol Y\sim {\cal N}(\boldsymbol 0,\sigma^2 \boldsymbol I_n)$ is minimax regret.
For example, a linear threshold rule $\delta(\boldsymbol{Y})=\mathbf{1}\left\{\boldsymbol{w}'\boldsymbol{Y}\ge 0\right\}$ for any $\boldsymbol{w}\neq \boldsymbol{0}$, a probit-like randomized rule $\delta(\boldsymbol{Y})=\Phi(\boldsymbol{w}'\boldsymbol{Y})$ for any $\boldsymbol{w}\neq \boldsymbol{0}$, and a data-independent randomized rule $\delta(\boldsymbol{Y})=1/2$ are all minimax regret. There exist infinitely many minimax regret rules for this case.

To gain intuition for Lemma \ref{lemma:one-dim0}, 
I outline the derivation for case (i) and provide the proof for case (ii) when $L(\bar\theta)>0$.
\paragraph{Case (i): $L(\bar\theta)>0$ \textnormal{and} $\boldsymbol{m}(\bar\theta)\neq \boldsymbol{0}$.}
Under $\theta=\lambda\bar\theta$, $\boldsymbol{Y}\sim {\cal N}(\lambda \boldsymbol{m}(\bar\theta),\sigma^2\boldsymbol{I}_n)$ and the welfare constrast is $\lambda L(\bar\theta)$ by the linearity of $\boldsymbol{m}$ and $L$.
Viewing $\lambda\in [-1,1]$ as the underlying parameter of $[-\bar\theta,\bar\theta]=\{\theta\in \mathbb{V}:\theta=\lambda\bar\theta,\lambda\in [-1,1]\}$,
one can show that the scalar statistic
$T(\boldsymbol{Y})=\frac{\boldsymbol m(\bar\theta)'\boldsymbol Y}{\|\boldsymbol m(\bar\theta)\|^2}\sim {\cal N}\left(\lambda,\frac{\sigma^2}{\|\boldsymbol m(\bar\theta)\|^2}\right)$
is a sufficient statistic of $\boldsymbol Y$ for $\lambda$.
Since the class of decision rules that only depend on a sufficient statistic is essentially complete\footnote{A class ${\cal C}$ of decision rules is {\it essentially complete} if, for any decision rule $\delta\notin{\cal C}$, there is a decision rule $\delta'\in{\cal C}$ such that $R(\delta,\theta)\ge R(\delta',\theta)$ for all $\theta\in\Theta$.} \citep[Theorem 1 in Chapter 1]{Berger1985book}, it is justified to restrict one's attention to rules that depend on $\boldsymbol{Y}$ only through $T(\boldsymbol{Y})\in\mathbb{R}$.

With this restricted class of rules, the minimax regret problem for $[-\bar \theta,\bar \theta]$ is equivalent to a {\it univariate} problem in which we observe a univariate sample $T\sim {\cal N}\left(\lambda,\frac{\sigma^2}{\|\boldsymbol m(\bar\theta)\|^2}\right)$ and the welfare contrast is $\lambda L(\bar\theta)$ for $\lambda\in [-1,1]$.
By a mild extension of the results of \cite{hirano2009asymptotics} and \cite{tetenov2012asymmetric} for univariate problems with unbounded parameter spaces to ones with bounded parameter spaces, the simple threshold rule $\delta(T)=\mathbf{1}\left\{T\ge 0\right\}$ is minimax regret for this univariate problem.
Consequently, 
$\delta^*(\boldsymbol Y)=\mathbf{1}\{T(\boldsymbol{Y})\ge 0\}=\mathbf{1}\left\{\boldsymbol{m}(\bar \theta)'\boldsymbol{Y}\ge 0\right\}$ is minimax regret for the original one-dimensional subproblem $[-\bar \theta,\bar \theta]$.

For the minimax risk, a simple calculation shows that the regret of $\delta^*$ under $\theta=\lambda\bar\theta$ is
\begin{align*}
    R(\delta^*,\lambda\bar\theta)&=|\lambda| L(\bar\theta) \cdot \Phi\left(-|\lambda|\cdot \|\boldsymbol{m}(\bar \theta)\|/\sigma\right).
\end{align*}
The first factor $|\lambda|L(\bar\theta)$ is the welfare loss when $\delta^*$ chooses the inferior policy under $\lambda\bar\theta$, which is increasing in $|\lambda|$.
The second factor $\Phi\left(-|\lambda|\cdot \|\boldsymbol{m}(\bar \theta)\|/\sigma\right)$ is the probability of choosing the inferior policy, which decreases in $|\lambda|$.
The regret $R(\delta^*,\lambda\bar\theta)$ is shown to be a bimodal function of $\lambda$ symmetric around zero, globally maximized at $\lambda\in \{-\tau^*\sigma/\|\boldsymbol{m}(\bar \theta)\|,\tau^*\sigma/\|\boldsymbol{m}(\bar \theta)\|\}$.
Maximizing this function over $\lambda\in [-1,1]$ yields the minimax risk in Lemma \ref{lemma:one-dim0}\ref{lemma:one-dim0:mmrisk}.

\paragraph{Case (ii): $L(\bar\theta)>0$ \textnormal{and} $\boldsymbol{m}(\bar\theta)= \boldsymbol{0}$.}
By the linearity of $\boldsymbol{m}$, $\boldsymbol{m}(\theta)=\boldsymbol{0}$ for any $\theta\in [-\bar\theta,\bar\theta]$. 
For a given rule $\delta$, the probability of choosing policy 1, $\mathbb{E}_{\theta}[\delta(\boldsymbol Y)]$, is constant over $\theta\in [-\bar\theta,\bar\theta]$, so that
\begin{align*}
    R(\delta,\theta)
    &=L(\theta)(1-\mathbb{E}[\delta(\boldsymbol Y)])\mathbf{1}\{L(\theta)\ge 0\}+(-L(\theta))\mathbb{E}[\delta(\boldsymbol Y)]\mathbf{1}\{L(\theta)< 0\},
\end{align*}
where $\boldsymbol Y\sim {\cal N}(\boldsymbol 0,\sigma^2 \boldsymbol I_n)$.
Consequently, the maximum regret can be calculated as follows:
$$
\sup_{\theta\in[-\bar\theta,\bar\theta]}R(\delta,\theta)=\begin{cases}
L(\bar \theta) (1-\mathbb{E}[\delta(\boldsymbol{Y})]) ~~\text{(at $\theta=\bar\theta$)}&\text{ if } \mathbb{E}[\delta(\boldsymbol{Y})]<1/2,\\
L(\bar \theta) /2  \quad\quad\quad\quad\hspace{0.62em}~~~\text{(at $\theta\in\{-\bar\theta,\bar\theta\}$)}&\text{ if } \mathbb{E}[\delta(\boldsymbol{Y})]=1/2,\\
L(\bar \theta)\mathbb{E}[\delta(\boldsymbol{Y})] \quad\quad~~~\hspace{.2em}\text{(at $\theta=-\bar\theta$)}&\text{ if } \mathbb{E}[\delta(\boldsymbol{Y})]>1/2.
\end{cases}
$$
Thus, any rule $\delta^*$ with $\mathbb{E}[\delta^*(\boldsymbol{Y})]=1/2$ is minimax regret, and 
the minimax risk is $L(\bar \theta)/2$.

\subsection{Hardest One-dimensional Subproblem}\label{section:hardest}

Now, I search for the {\it hardest one-dimensional subfamily} $[-\bar{\theta}^*,\bar{\theta}^*]\subset\Theta$, which satisfies $
{\cal R}([-\bar{\theta}^*,\bar{\theta}^*])=\sup_{\bar\theta\in\Theta:L(\bar\theta)\ge 0}{\cal R}([-\bar \theta,\bar \theta])$.
The key to characterizing the hardest one-dimensional subfamily is the {\it modulus of continuity}, defined as
\begin{align}
\omega(\epsilon)\coloneqq \sup\{L(\theta): \|\boldsymbol{m}(\theta)\|\le \epsilon,\theta\in\Theta\},~~~\epsilon\ge 0. \label{eq:def-mod}
\end{align}
The modulus of continuity and its variants have been used in constructing minimax estimators and confidence intervals on linear functionals in Gaussian models \citep{donoho1991,donoho1994,low1995tradeoff,armstrong2018optimal}.\footnote{\cite{donoho1994} defines the modulus of continuity as $\tilde \omega(\epsilon)= \sup\{|L(\theta)-L(\tilde\theta)|: \|\boldsymbol{m}(\theta-\tilde\theta)\|\le \epsilon,\theta,\tilde\theta\in\Theta\}$.
	If $\Theta$ is convex and centrosymmetric, the relationship $\tilde\omega(\epsilon)=2\omega(\epsilon/2)$ holds.}
Under Assumption \ref{assumption:problem}, $\omega(\epsilon)$ is the value of a convex optimization problem.
If $\Theta$ is closed, this problem typically has a solution.
\footnote{See \citet[Lemma 2]{donoho1994} for sufficient conditions for the existence of a solution.}
In addition, $\omega(\cdot)$ is nonnegative and nondecreasing by construction.
Furthermore, the modulus of continuity has the following properties.
\begin{lemma}\label{lemma:mod}
    Under Assumption \ref{assumption:problem}, $\omega(\epsilon)<\infty$ for all $\epsilon\ge 0$, and $\omega(\cdot)$ is concave and continuous on $[0,\infty)$ and right differentiable at $0$.
\end{lemma}
\begin{proof}
	See Lemma \ref{lemma:superdifferential} in Appendix \ref{appendix:lemma}.
\end{proof}

The following result shows that the hardest one-dimensional subfamily can be obtained by solving an optimization problem that involves the modulus of continuity.
Define the right derivative of $\omega(\cdot)$ at $0$ as $\omega'(0)\coloneqq\lim_{\epsilon\downarrow 0}\frac{\omega(\epsilon)-\omega(0)}{\epsilon}$.
Let $\phi$ denote the probability density function of a standard normal random variable.
For $\epsilon\ge 0$, I say that {\it $\theta_\epsilon\in \Theta$ attains the modulus of continuity at $\epsilon$} if $L(\theta_\epsilon)=\omega(\epsilon)$ and $\|\boldsymbol{m}(\theta_\epsilon)\|\le \epsilon$.


\begin{lemma}[Hardest One-dimensional Subproblem]\label{lemma:hardest-1d}
	Under Assumption \ref{assumption:problem}, the following holds.
\begin{enumerate}[label=(\roman*)]
    \item \label{lemma:hardest-1d:risk}
    The largest minimax risk among one-dimensional subproblems is given by
    \begin{align}
    \sup_{\bar\theta\in\Theta:L(\bar\theta)\ge 0}{\cal R}([-\bar \theta,\bar \theta])=\sup_{\epsilon\in[0,\tau^*\sigma]}\omega(\epsilon)\Phi(-\epsilon/\sigma).\label{eq:hardest}
    \end{align}
    \item \label{lemma:hardest-1d:epsilon}
    There exists a unique solution $\epsilon^*$ to the right-hand side of \eqref{eq:hardest}.
    Furthermore, $\epsilon^*>0$ if and only if $\sigma\omega'(0)>2\phi(0)\omega(0)$.
    \item \label{lemma:hardest-1d:sol}
    If there exists $\theta_{\epsilon^*}\in\Theta$ that attains the modulus of continuity at $\epsilon^*$, then $[-\theta_{\epsilon^*},\theta_{\epsilon^*}]$ is the hardest one-dimensional subfamily. That is,
     $$
    {\cal R}([-\theta_{\epsilon^*},\theta_{\epsilon^*}])=\sup_{\bar\theta\in\Theta:L(\bar\theta)\ge 0}{\cal R}([-\bar \theta,\bar \theta]).
    $$
	Furthermore, $\boldsymbol{m}(\theta_{\epsilon^*})$ does not depend on the choice of $\theta_{\epsilon^*}$ among potentially multiple $\theta$'s that attain the modulus of continuity at $\epsilon^*$, and $\|\boldsymbol{m}(\theta_{\epsilon^*})\|=\epsilon^*$.
 \end{enumerate}
\end{lemma}
\begin{proof}
	See Appendix \ref{proof:lemma:hardest-1d}.
\end{proof}

To provide intuition for this result, consider a convenient case in which $\omega(\epsilon)=\sup_{\theta\in\Theta:\|\boldsymbol{m}(\theta)\|= \epsilon} L(\theta)$ for each $\epsilon\ge 0$. In other words, suppose, for illustration, that the supremum remains the same if the inequality constraint $\|\boldsymbol{m}(\theta)\|\le \epsilon$ is replaced by the equality constraint $\|\boldsymbol{m}(\theta)\|= \epsilon$.
In this case, the largest minimax risk among one-dimensional subproblems can be calculated as follows:
\begin{align*}
&
\sup_{\bar\theta\in\Theta:L(\bar\theta)\ge 0}{\cal R}([-\bar \theta,\bar \theta])=\sup_{\epsilon\ge 0}\sup_{\bar\theta\in\Theta:\|\boldsymbol m(\bar\theta)\|=\epsilon,L(\bar\theta)\ge 0}{\cal R}([-\bar \theta,\bar \theta]) \\
&=\sup\left\{\sup_{\epsilon\in [0, \tau^*\sigma]}\sup_{\bar\theta\in\Theta:\|\boldsymbol m(\bar\theta)\|=\epsilon}L(\bar\theta)\Phi\left(-\frac{\|\boldsymbol m(\bar\theta)\|}{\sigma}\right),\sup_{\epsilon>\tau^*\sigma}\sup_{\bar\theta\in\Theta:\|\boldsymbol m(\bar\theta)\|=\epsilon} \dfrac{\tau^*\sigma L(\bar\theta)}{\|\boldsymbol m(\bar\theta)\|}\Phi\left(-\tau^*\right)\right\}\\
&=\sup\left\{\sup_{\epsilon\in [0, \tau^*\sigma]}\omega(\epsilon)\Phi(-\epsilon/\sigma),\sup_{\epsilon>\tau^*\sigma}\frac{\tau^*\sigma\omega(\epsilon)}{\epsilon}\Phi(-\tau^*)\right\},
\end{align*}
where 
the second equality uses Lemma \ref{lemma:one-dim0}\ref{lemma:one-dim0:mmrisk} and the third equality uses the assumption that $\omega(\epsilon)=\sup_{\theta\in\Theta:\|\boldsymbol{m}(\theta)\|= \epsilon} L(\theta)$.
To simplify the last expression, note that $\frac{\omega(\epsilon)}{\epsilon}$ is continuous and nonincreasing on $(0,\infty)$ by the concavity of $\omega(\cdot)$, and hence $\sup_{\epsilon>\tau^*\sigma}\frac{\tau^*\sigma\omega(\epsilon)}{\epsilon}\Phi(-\tau^*)=\omega(\tau^*\sigma)\Phi(-\tau^*)\le \sup_{\epsilon\in[0,\tau^*\sigma]}\omega(\epsilon)\Phi(-\epsilon/\sigma)$.
As a result,
$$
\sup_{\bar\theta\in\Theta:L(\bar\theta)\ge 0}{\cal R}([-\bar \theta,\bar \theta])=\sup_{\epsilon\in[0,\tau^*\sigma]}\omega(\epsilon)\Phi(-\epsilon/\sigma).
$$

In light of the above derivation, the problem of searching for the hardest one-dimensional subfamily can be interpreted as the following problem by an adversarial Nature. Nature optimizes $\epsilon\in[0,\tau^*\sigma]$, which represents a level of the strength of the signal provided by the sample $\boldsymbol{Y}$ within subfamily $[-\bar\theta,\bar\theta]$.
The signal strength is measured by $\|\boldsymbol m(\bar\theta)\|$: The larger $\|\boldsymbol m(\bar\theta)\|$ is, the more information the sample $\boldsymbol Y$ provides about the sign of $L(\theta)$, and
the smaller the error probability $\Phi(-\|\boldsymbol m(\bar\theta)\|/\sigma)$ is.
The modulus of continuity $\omega(\epsilon)= \sup\{L(\bar \theta): \|\boldsymbol{m}(\bar \theta)\|= \epsilon,\bar \theta\in\Theta\}$ then represents the maximum potential welfare loss (i.e., $|L(\bar \theta)|$) among subfamilies subject to a given level of signal strength.
Nature finally searches for the best level $\epsilon^*\in\arg\max_{\epsilon\in[0,\tau^*\sigma]}\omega(\epsilon)\Phi(-\epsilon/\sigma)$, which optimizes the balance between the error probability and maximum potential welfare loss.
If $\epsilon^*>0$, the sample $\boldsymbol{Y}$ is informative within the hardest subfamily $[-\theta_{\epsilon^*},\theta_{\epsilon^*}]$.
On the other hand,
if $\epsilon^*=0$, the sample $\boldsymbol{Y}$ is uninformative within the hardest subfamily.

Lemma \ref{lemma:hardest-1d}\ref{lemma:hardest-1d:epsilon} implies that $\epsilon^*>0$ if and only if $\omega'(0)/\omega(0)$ or $\sigma$ is sufficiently large.
This is intuitive from the perspective of Nature's problem of optimizing the signal strength described above.
The larger $\omega'(0)/\omega(0)(=\left.\frac{\partial}{\partial\epsilon}\log\omega(\epsilon)\right\vert_{\epsilon=0})$, the larger the percentage increase in maximum potential welfare loss associated with an increase in the signal level from $0$, and thus the greater Nature's incentive to choose a nonzero level of signal strength. Similarly, the larger $\sigma$ is, the noisier the sample $\boldsymbol{Y}$ becomes, leading to a smaller decrease in the error probability associated with an increase in the signal level, and consequently, a greater incentive for Nature to choose a nonzero signal strength.

\begin{remark}[Role of the Centrosymmetry of $\Theta$]\label{remark:centrosymmetry}
    In Lemma \ref{lemma:hardest-1d}, I derive the hardest one-dimensional subfamily among centrosymmetric one-dimensional subfamilies $[-\bar\theta,\bar\theta]$.
    Under the centrosymmetry of the original parameter space $\Theta$, one can show that this subfamily is also hardest among all one-dimensional subfamilies of the form $[\bar\theta_0,\bar\theta_1]=\{(1-\lambda)\bar\theta_0+\lambda\bar\theta_1:\lambda\in [0,1]\}\subset\Theta$, including noncentrosymmetric ones.
    If $\Theta$ is noncentrosymmetric, it is also necessary to consider noncentrosymmetric subfamilies to characterize the hardest one.
    Although it is possible to solve noncentrosymmetric subproblems, their minimax risks do not take a simple form as in Lemma \ref{lemma:one-dim0}.\footnote{More specifically, the minimax risk of a subproblem $[\bar\theta_0,\bar\theta_1]$ depends on both $\bar\theta_0$ and $\bar\theta_1$ even when $\bar\theta_1-\bar\theta_0$ is held fixed.
    This sharply contrasts with the minimax affine estimation problems studied by \cite{donoho1994}, in which the minimax risk of a subproblem $[\bar\theta_0,\bar\theta_1]$ depends on $\bar\theta_0$ and $\bar\theta_1$ only through $\bar\theta_1-\bar\theta_0$ (in particular $L(\bar\theta_1-\bar\theta_0)$ and $\boldsymbol{m}(\bar\theta_1-\bar\theta_0)$).} Consequently, it is not straightforward to apply a similar argument to the one above to characterize the hardest one-dimensional subfamily.
    I leave the extension to noncentrosymmetric parameter spaces to future work.
\end{remark}

I illustrate the results in Lemma \ref{lemma:hardest-1d} using Example \ref{example:intersection}.
\setcounter{example}{0}
\begin{example}[Continued]
Without loss of generality, I assume $C_1\ge C_2$. 
The modulus of continuity is given by
$\omega(\epsilon)=\sup\{\theta_3:\theta\in\mathbb{R}^3, |\theta_1-\theta_3|\le C_1, |\theta_2-\theta_3|\le C_2, \theta_1^2+\theta_2^2\le\epsilon^2\}$.
Solving the optimization problem on the right-hand side yields
$$
\omega(\epsilon)=\begin{cases}
    \epsilon+C_2 ~~ &\text{ if } 0\le \epsilon\le C_1-C_2,\\
    \frac{1}{2}\left(\left(2\epsilon^2-(C_1-C_2)^2\right)^{1/2}+C_1+C_2\right) ~~ &\text{ if } \epsilon> C_1-C_2,
\end{cases}
$$
and the solution $\theta_\epsilon=(\theta_{\epsilon,1},\theta_{\epsilon,2},\theta_{\epsilon,3})'$ is given by
\begin{align}
\theta_\epsilon=\begin{cases}
    (0,\epsilon,\omega(\epsilon))' ~~ &\text{ if } 0\le \epsilon\le C_1-C_2,\\
    (\theta_{\epsilon,1},\theta_{\epsilon,1}+C_1-C_2,\omega(\epsilon))' ~~ &\text{ if } \epsilon> C_1-C_2,
\end{cases}
\label{eq:intersection-theta}
\end{align}
where $\theta_{\epsilon,1}=\frac{1}{2}\left(\left(2\epsilon^2-(C_1-C_2)^2\right)^{1/2}-C_1+C_2\right)$ if $\epsilon> C_1-C_2$.
By Lemma \ref{lemma:hardest-1d}, the hardest subfamily is $[-\theta_{\epsilon^*},\theta_{\epsilon^*}]$, where $\epsilon^*\in \arg\max_{\epsilon\in[0,\tau^*\sigma]}\omega(\epsilon)\Phi(-\epsilon/\sigma)$.
If $C_1>C_2$, then $\omega(0)=C_2$ and $\omega'(0)=1$, and therefore the condition $\sigma\omega'(0)>2\phi(0)\omega(0)$ corresponds to $\sigma>2\phi(0)C_2$.
If $C_1=C_2$, then $\omega(0)=C_2$ and $\omega'(0)=\frac{1}{\sqrt{2}}$, and therefore the condition corresponds to $\frac{\sigma}{\sqrt{2}}>2\phi(0)C_2$.
In both cases, $\epsilon^*>0$ if and only if $C_2$ is small relative to $\sigma$.
\end{example}

\subsection{Main Result: Minimax Regret Rule for the Full Problem}\label{section:mmr}

In the final step, I propose a minimax regret rule for the hardest one-dimensional subproblem $[-\theta_{\epsilon^*},\theta_{\epsilon^*}]$ and verify that it is indeed minimax regret for the full problem.
For the case in which $\|\boldsymbol{m}(\theta_{\epsilon^*})\|=\epsilon^*>0$, I consider the rule $\delta(\boldsymbol{Y})=\mathbf{1}\left\{\boldsymbol{m}(\theta_{\epsilon^*})'\boldsymbol{Y}\ge 0\right\}$, which is minimax regret for $[-\theta_{\epsilon^*},\theta_{\epsilon^*}]$ by Lemma \ref{lemma:one-dim0}.
On the other hand,
for the case in which $\|\boldsymbol{m}(\theta_{\epsilon^*})\|=\epsilon^*=0$, there exist infinitely many minimax regret rules for $[-\theta_{\epsilon^*},\theta_{\epsilon^*}]$.
Among them, I consider a rule that can be viewed as a continuous extension of the rule $\delta(\boldsymbol{Y})=\mathbf{1}\left\{\boldsymbol{m}(\theta_{\epsilon^*})'\boldsymbol{Y}\ge 0\right\}$ from the case with $\epsilon^*>0$ to the case with $\epsilon^*=0$.

To construct the rule, first define $\bar I(\boldsymbol{\mu})$ as the upper bound on the welfare contrast when the reduced-form parameter is $\boldsymbol{\mu}$:
    $$
    \bar I(\boldsymbol{\mu})\coloneqq 
    \sup I(\boldsymbol{\mu})=\sup\{L(\theta):\boldsymbol m(\theta)=\boldsymbol \mu, \theta\in \Theta\},~~~\boldsymbol{\mu}\in \mathbb{R}^n,
    $$
where I use the convention that $\sup I(\boldsymbol{\mu})=-\infty$ when 
$I(\boldsymbol{\mu})$ is empty.
Next, consider a sequence $\{\boldsymbol{\mu}_\epsilon\}_{\epsilon\in (0,\bar\epsilon)}$ indexed by positive real numbers $\epsilon\in (0,\bar\epsilon)$ with some $\bar\epsilon>0$ such that
\begin{align}
\boldsymbol \mu_\epsilon\in \arg\max_{\boldsymbol{\mu}\in{\cal M}:\|\boldsymbol{\mu}\|\le\epsilon}\bar I(\boldsymbol{\mu}),~~~\epsilon\in (0,\bar\epsilon).\label{eq:mu}
\end{align}
Lastly, define
\begin{align}
\boldsymbol{w^*}\coloneqq\lim_{\epsilon\downarrow 0}\frac{\boldsymbol \mu_\epsilon}{\epsilon}.\label{eq:wstar}
\end{align}
In Lemma \ref{lemma:superdifferential} in Appendix \ref{appendix:lemma}, I show that the following holds under Assumption \ref{assumption:problem}. First, there exists a sequence $\{\boldsymbol{\mu}_\epsilon\}_{\epsilon\in (0,\bar\epsilon)}$ that satisfies \eqref{eq:mu}. Second, if $\omega'(0)>0$, the limit $\boldsymbol{w^*}$ exists and does not depend on the choice of $\{\boldsymbol{\mu}_\epsilon\}_{\epsilon\in (0,\bar\epsilon)}$ among potentially multiple sequences that satisfy \eqref{eq:mu}; that is, $\boldsymbol{w^*}$ is uniquely defined. Third, $\boldsymbol{w}^*$ is the direction in which the directional derivative of $\bar I(\cdot)$ at $\boldsymbol{0}$ is maximized among unit vectors.
In other words, $\boldsymbol w^*$ is the ``least favorable'' direction, in the sense that the upper bound on the welfare contrast increases the most if the reduced-form parameter $\boldsymbol{m}(\theta)$ is changed from $\boldsymbol{0}$ in the direction $\boldsymbol w^*$.

The vector $\boldsymbol{w}^*$ relates to the modulus of continuity $\omega(\cdot)$ as follows: $\omega(\epsilon)=\sup_{\theta\in\Theta:\|\boldsymbol m(\theta)\|\le\epsilon}L(\theta)=\sup_{\boldsymbol{\mu}\in{\cal M}:\|\boldsymbol{\mu}\|\le\epsilon}\bar I(\boldsymbol{\mu})$; and
if $\theta_\epsilon\in\Theta$ attains the modulus of continuity at $\epsilon$, then setting $\boldsymbol \mu_\epsilon=\boldsymbol{m}(\theta_\epsilon)$ satisfies \eqref{eq:mu}, and $\boldsymbol{w}^*=\lim_{\epsilon\downarrow 0}\frac{\boldsymbol{m}(\theta_\epsilon)}{\epsilon}$.

The following result derives a minimax regret rule for the original problem, which covers both the case in which $\epsilon^*>0$ (i.e., $\sigma\omega'(0)> 2\phi(0)\omega(0)$) and the case in which $\epsilon^*=0$ (i.e., $\sigma\omega'(0)\le 2\phi(0)\omega(0)$).
The proof and a technical discussion are provided in Appendix \ref{section:proof}.

\begin{theorem}[Minimax Regret Rule for the Full Problem]\label{theorem:main}
    Suppose that Assumption \ref{assumption:problem} holds.
    Let $\epsilon^*\in\arg\max_{\epsilon\in[0,\tau^*\sigma]}\omega(\epsilon)\Phi(-\epsilon/\sigma)$ and suppose that there exists $\theta_{\epsilon^*}\in\Theta$ that attains the modulus of continuity at $\epsilon^*$.
	Then, the following decision rule is minimax regret:
 	\begin{align}
	\delta^*(\boldsymbol Y)=\begin{cases}\mathbf{1}\left\{\boldsymbol{m}(\theta_{\epsilon^*})'\boldsymbol{Y}\ge 0\right\} & ~~\text{if } \omega'(0)>0, 2\phi(0)\frac{\omega(0)}{\omega'(0)} < \sigma,\\
	\mathbf{1}\left\{(\boldsymbol{w}^*)'\boldsymbol{Y}\ge 0\right\} & ~~\text{if } \omega'(0)>0, 2\phi(0)\frac{\omega(0)}{\omega'(0)} = \sigma,\\
	\Phi\left(\dfrac{(\boldsymbol{w}^*)'\boldsymbol{Y}}{((2\phi(0)\omega(0)/\omega'(0))^2-\sigma^2)^{1/2}}\right)& ~~\text{if } \omega'(0)>0,  2\phi(0)\frac{\omega(0)}{\omega'(0)} > \sigma,\\
        1/2 & ~~\text{if } \omega'(0)=0.
	\end{cases}\label{eq:mmr}
	\end{align}  
    Furthermore,
	the minimax risk is given by
    $$
    {\cal R}(\Theta)=R(\delta^*,-\theta_{\epsilon^*})=R(\delta^*,\theta_{\epsilon^*})=\omega(\epsilon^*)\Phi(-\epsilon^*/\sigma).
    $$
\end{theorem}



Theorem \ref{theorem:main} yields the following implications.
First, the minimax regret rule in Theorem \ref{theorem:main} depends on the sample only through a weighted sum, $\boldsymbol{m}(\theta_{\epsilon^*})'\boldsymbol{Y}$ or $(\boldsymbol{w}^*)'\boldsymbol{Y}$, even though no such restrictions are imposed.
The weights can be calculated by solving convex optimization problems; see Section \ref{subsection:computation} for a computational procedure.

Second, the rule is nonrandomized or randomized, depending on whether the condition $2\phi(0)\frac{\omega(0)}{\omega'(0)}\le \sigma$ holds.
This condition is related to the strength of the identifying restrictions.
Under Assumption \ref{assumption:problem}, the closure of $I(\boldsymbol{0})$ is shown to be
$[-\omega(0),\omega(0)]$.\footnote{Since $L$ and $\boldsymbol m$ are linear and $\Theta$ is centrosymmetric, $-\omega(0)=\inf\{L(\theta):\boldsymbol m(\theta)=\boldsymbol 0, \theta\in\Theta\}$.
	Moreover, for any $\alpha\in (-\omega(0),\omega(0))$, we can find $\theta\in\Theta$ such that $L(\theta)=\alpha$ and $\boldsymbol m(\theta)=\boldsymbol 0$ by the linearity of $L$ and $\boldsymbol m$ and the convexity of $\Theta$.}
We can thus interpret $\omega(0)$ as half the length of the identified set of $L(\theta)$ when $\boldsymbol m(\theta)=\boldsymbol 0$.

If $L(\theta)$ is point identified, the length of the identified set is zero, so $2\phi(0)\frac{\omega(0)}{\omega'(0)}\le \sigma$. Therefore, the minimax regret rule is always nonrandomized under point identification.
Even if $L(\theta)$ is not point identified, when the identified set is small relative to the noise level $\sigma$ (holding $\omega'(0)$ fixed), the condition $2\phi(0)\frac{\omega(0)}{\omega'(0)}\le\sigma$ holds and the rule is nonrandomized.
On the other hand, when the identified set is large relative to $\sigma$, the minimax rule is randomized.

The randomized rule can equivalently be written as $\delta^*(\boldsymbol{Y})=\mathbb{P}\left((\boldsymbol{w}^*)'\boldsymbol{Y}+\xi\ge 0|\boldsymbol{Y}\right)$, where $\xi|\boldsymbol{Y}\sim {\cal N}(0,(2\phi(0)\omega(0)/\omega'(0))^2-\sigma^2)$.
 This rule can be implemented by first adding an independent noise $\xi$ to a scalar statistic $(\boldsymbol{w}^*)'\boldsymbol{Y}$ and then making a decision according to the sign of $(\boldsymbol{w}^*)'\boldsymbol{Y}+\xi$.
This addition artificially increases the standard deviation of $(\boldsymbol{w}^*)'\boldsymbol{Y}$ from $\sigma$ to $2\phi(0)\omega(0)/\omega'(0)$, which is the threshold at which we switch from a nonrandomized rule to a randomized rule.
The larger $\omega(0)$ is, the larger the variance of $\xi$ is and the more dependent the choice is on the noise.
As a result, given any realization of $\boldsymbol Y$, the probabilities of choosing policy 1 and policy 0 approach $1/2$ as $\omega(0)$ increases, which suggests that the decisions become more mixed if we impose weaker restrictions on $\Theta$.



I apply Theorem \ref{theorem:main} to derive a minimax regret rule for Example \ref{example:intersection}.
\setcounter{example}{0}
\begin{example}[Continued]
First, consider the case in which $C_1>C_2$. In this case, for any $\epsilon\in (0,C_1-C_2)$, $\omega(\epsilon)=\epsilon+C_2$ and $\theta_\epsilon=\left(0,\epsilon,\omega(\epsilon)\right)'$. Let $\boldsymbol{\mu}_\epsilon=\left(0,\epsilon\right)'$, which satisfies \eqref{eq:mu} for every $\epsilon\in (0,C_1-C_2)$, so that $\boldsymbol{w}^*=\lim_{\epsilon\downarrow 0}\frac{\boldsymbol \mu_\epsilon}{\epsilon}=\left(0,1\right)'$.
By Theorem \ref{theorem:main}, the following decision rule is minimax regret:
 	\begin{align*}
	\delta^*(\boldsymbol Y)=\begin{cases}\mathbf{1}\left\{\theta_{\epsilon^*,1}Y_1+\theta_{\epsilon^*,2}Y_2\ge 0\right\} & ~~\text{if } 2\phi(0)C_2 < \sigma,\\
	\mathbf{1}\left\{Y_2\ge 0\right\} & ~~\text{if } 2\phi(0)C_2 = \sigma,\\
	\Phi\left(\dfrac{Y_2}{((2\phi(0)C_2)^2-\sigma^2)^{1/2}}\right)& ~~\text{if } 2\phi(0)C_2 > \sigma,
	\end{cases}
	\end{align*}
where $\epsilon^*\in\arg\max_{\epsilon\in[0,\tau^*\sigma]}\omega(\epsilon)\Phi(-\epsilon/\sigma)$ and $\theta_\epsilon$ is given by \eqref{eq:intersection-theta}.
If $C_2$ is sufficiently small relative to $\sigma$, the rule is nonrandomized. Whether it gives a nonzero weight to both $Y_1$ and $Y_2$ depends on $(C_1,C_2)$:
If $C_1$ is sufficiently close to $C_2$ so that $C_1-C_2<\epsilon^*$, then $\theta_{\epsilon^*,1}=\frac{1}{2}\left(\left(2(\epsilon^*)^2-(C_1-C_2)^2\right)^{1/2}-C_1+C_2\right)>0$; otherwise, $\theta_{\epsilon^*,1}=0$ even if the rule is nonrandomized.
If $C_2$ is sufficiently large so that $2\phi(0)C_2 > \sigma$, the rule is randomized.

Next, consider the case in which $C_1=C_2$.
In this case, for any $\epsilon>0$, $\omega(\epsilon)=\frac{1}{\sqrt{2}}\epsilon+C_2$ and $\theta_\epsilon=\left(\frac{\epsilon}{\sqrt{2}},\frac{\epsilon}{\sqrt{2}},\omega(\epsilon)\right)'$.
Let $\boldsymbol{\mu}_\epsilon=\left(\frac{\epsilon}{\sqrt{2}},\frac{\epsilon}{\sqrt{2}}\right)'$, which satisfies \eqref{eq:mu} for every $\epsilon>0$, so that $\boldsymbol{w}^*=\lim_{\epsilon\downarrow 0}\frac{\boldsymbol \mu_\epsilon}{\epsilon}=\left(\frac{1}{\sqrt{2}},\frac{1}{\sqrt{2}}\right)'$.
By Theorem \ref{theorem:main}, the following decision rule is minimax regret:
 	\begin{align*}
	\delta^*(\boldsymbol Y)=\begin{cases}\mathbf{1}\left\{\frac{1}{\sqrt{2}}(Y_1+Y_2)\ge 0\right\} & ~~\text{if } 2\sqrt{2}\phi(0)C_2 \le \sigma,\\
	\Phi\left(\dfrac{\frac{1}{\sqrt{2}}(Y_1+Y_2)}{((2\sqrt{2}\phi(0)C_2)^2-\sigma^2)^{1/2}}\right)& ~~\text{if } 2\sqrt{2}\phi(0)C_2 > \sigma.
	\end{cases}
	\end{align*}
The rule always gives equal weights to $Y_1$ and $Y_2$, unlike in the case in which $C_1>C_2$.
\end{example}

Below, I discuss the role of randomization in Section \ref{section:randomization} and the relationship between Theorem \ref{theorem:main} and existing results in Section \ref{section:stoye}. Section \ref{section:proof-strategy} introduces the proof strategy.
    Finally, Section \ref{subsection:computation} provides the procedure for computing the minimax regret rule.

\subsubsection{The Role of Randomization}\label{section:randomization}
Randomization plays the role of reducing the probability of choosing the inferior policy under parameter values at which the error probability of an original rule exceeds one-half.
If the maximum regret of the original rule is attained at such parameter values, randomization can lead to a reduction in maximum regret.

To illustrate this, it is useful to compare the nonrandomized rule $\delta^*_{\rm NR}(\boldsymbol{Y})=\mathbf{1}\left\{(\boldsymbol{w}^*)'\boldsymbol{Y}\ge 0\right\}$ with the randomized rule $\delta^*$ in Theorem \ref{theorem:main}.
In Appendix \ref{appendix:discussion}, I show that, under a mild condition on the local behavior of $\bar I(\boldsymbol{\mu})$ around $\boldsymbol{0}$, if $2\phi(0)\frac{\omega(0)}{\omega'(0)}>\sigma$, there exists $\theta^*\in\Theta$ such that $L(\theta^*)\ge 0$ (i.e., the optimal policy is policy 1),
\begin{align}
&1-\mathbb{E}_{\theta^*}[\delta^*_{\rm NR}(\boldsymbol{Y})]=\Phi\left(-\frac{(\boldsymbol{w}^*)'\boldsymbol{m}(\theta^*)}{\sigma}\right)>1/2,\label{eq:randomization-error}\\
~\text{ and }~&R(\delta^*_{\rm NR},\theta^*)=L(\theta^*)\Phi\left(-\frac{(\boldsymbol{w}^*)'\boldsymbol{m}(\theta^*)}{\sigma}\right)>\mathcal{R}(\Theta).
\label{eq:randomization}
\end{align}
Condition \eqref{eq:randomization-error} means that the error probability of $\delta^*_{\rm NR}$ exceeds one-half under $\theta^*$.
Condition \eqref{eq:randomization} says that the regret of $\delta^*_{\rm NR}$ under $\theta^*$ exceeds the minimax risk over $\Theta$, and therefore $\delta^*_{\rm NR}$ is not minimax regret.

Now, consider the randomized rule $\delta^*$ in Theorem \ref{theorem:main}.
Under any $\theta^*$ that satisfies \eqref{eq:randomization-error} and \eqref{eq:randomization}, randomization reduces the error probability and therefore regret of $\delta^*_{\rm NR}$. Specifically, a simple calculation shows that if $2\phi(0)\frac{\omega(0)}{\omega'(0)}>\sigma$,
$$
1-\mathbb{E}_{\theta^*}[\delta^*(\boldsymbol{Y})]=\Phi\left(-\frac{(\boldsymbol{w}^*)'\boldsymbol{m}(\theta^*)}{2\phi(0)\omega(0)/\omega'(0)}\right)<\Phi\left(-\frac{(\boldsymbol{w}^*)'\boldsymbol{m}(\theta^*)}{\sigma}\right)=1-\mathbb{E}_{\theta^*}[\delta^*_{\rm NR}(\boldsymbol{Y})],
$$
and hence $R(\delta^*,\theta^*)<R(\delta^*_{\rm NR},\theta^*)$.
Although randomization may increase regret under other values of $\theta\in\Theta$,
it turns out that $\delta^*$ achieves a smaller maximum regret over $\Theta$ than $\delta^*_{\rm NR}$.

\begin{remark}[Approaches to Avoiding Randomization]
    In practice, randomization may not be permitted in some cases due to ethical or legislative constraints.
    An informal way of avoiding randomization is to reconsider the problem specification to reach the regime in which the minimax regret rule in Theorem \ref{theorem:main} is nonrandomized. This can be done by imposing more restrictions on the parameter space and/or respecifying policies 1 and 0.
    Note that this is not a post hoc analysis, since the problem specification and the associated decision rule are still determined before observing the data.
    A more formal way is to consider using the recently proposed least randomizing minimax regret rule by \cite{olea2023partial}, which is a piecewise linear function of $(\boldsymbol{w}^*)'\boldsymbol{Y}$.
    This approach reduces the range of data realizations for which the decision is randomized. However, it does not lead to the complete elimination of randomization.
\end{remark}

\subsubsection{Relation to Existing Results}\label{section:stoye}
Theorem \ref{theorem:main} generalizes \citeauthor{Stoye2012minimax}'s \citeyearpar{Stoye2012minimax} result from univariate problems to multivariate problems.
For the case with randomization, the use of a probit-like rule in Theorem \ref{theorem:main} is inspired by \citeauthor{Stoye2012minimax}'s \citeyearpar{Stoye2012minimax} for univariate problems.
A novelty of my result is to use the hardest one-dimensional subfamily argument to construct a certain scalar statistic, $\boldsymbol{m}(\theta_{\epsilon^*})'\boldsymbol{Y}$ or $(\boldsymbol{w}^*)'\boldsymbol{Y}$, of the multivariate sample $\boldsymbol{Y}$ such that a threshold rule based on the statistic (plus a noise for the case with randomization) is minimax regret.
This sharply contrasts with univariate problems, in which the only natural choice of the statistic is the univariate sample $Y$ itself.

In the setting described in Example \ref{example:intersection}, \cite{ishihara2021meta} propose a way of numerically computing a minimax regret rule within the class of decision rules of the form $\delta(\boldsymbol{Y})=\mathbf{1}\{\sum_{i=1}^nw_iY_i\ge 0\}$ with $\sum_{i=1}^nw_i=1$.
An application of Theorem \ref{theorem:main}
shows that this restricted class contains an unconstrained minimax regret rule when $2\phi(0)\frac{\omega(0)}{\omega'(0)}\le \sigma$ and may not when $2\phi(0)\frac{\omega(0)}{\omega'(0)}>\sigma$.

\subsubsection{Proof Strategy}\label{section:proof-strategy}

I briefly describe the approach to proving Theorem \ref{theorem:main}.
First note that lower and upper bounds on the minimax risk for the full problem are given by
$$
\sup_{\theta\in [-\theta_{\epsilon^*},\theta_{\epsilon^*}]}R(\delta^*,\theta)={\cal R}([-\theta_{\epsilon^*},\theta_{\epsilon^*}])\le {\cal R}(\Theta)\le \sup_{\theta\in\Theta}R(\delta^*,\theta),
$$
where the equality follows from the fact that $\delta^*$ is minimax regret for the subproblem $[-\theta_{\epsilon^*},\theta_{\epsilon^*}]$ by Lemma \ref{lemma:one-dim0}, and the two inequalities hold by the definition of the minimax risk.
To prove that $\delta^*$ is also minimax regret for the full problem, it suffices to show that the above lower and upper bounds coincide.
More specifically, I show that
the maximum regret of $\delta^*$ over $\Theta$ is attained at $-\theta_{\epsilon^*}$ and $\theta_{\epsilon^*}$; that is,
\begin{align}
R(\delta^*,-\theta_{\epsilon^*})=R(\delta^*,\theta_{\epsilon^*})=\sup_{\theta\in\Theta}R(\delta^*,\theta).\label{eq:property}
\end{align}
To show this, I first express the maximum regret of $\delta^*$ over $\Theta$ as 
$$
\sup_{\theta\in\Theta}R(\delta^*,\theta)=\begin{cases}\sup_{\boldsymbol{\mu}\in {\cal M}:\bar I(\boldsymbol{\mu})\ge 0}\bar I(\boldsymbol{\mu})\Phi\left(-\dfrac{\boldsymbol{m}(\theta_{\epsilon^*})'\boldsymbol{\mu}}{\sigma\epsilon^*}\right) & ~~\text{if } \epsilon^*>0,\\
	\sup_{\boldsymbol{\mu}\in{\cal M}:\bar I(\boldsymbol{\mu})\ge 0}\bar I(\boldsymbol{\mu})\Phi\left(-\dfrac{(\boldsymbol{w}^*)'\boldsymbol{\mu}}{2\phi(0)\omega(0)/\omega'(0)}\right)& ~~\text{if } \epsilon^*=0.
	\end{cases}
$$
Here, the objective function on the right-hand side represents the maximum regret of $\delta^*$ over $\{\theta\in\Theta:\boldsymbol{m}(\theta)=\boldsymbol{\mu},L(\theta)\ge 0\}$.
I then show that the supremum is attained at $\boldsymbol{\mu}=\boldsymbol{m}(\theta_{\epsilon^*})$, which proves \eqref{eq:property}.

The arguments used to show \eqref{eq:property} are substantially different from the arguments used by \cite{donoho1994}, who shows a counterpart of \eqref{eq:property} for minimax affine estimation problems.
\citeauthor{donoho1994}'s \citeyearpar{donoho1994} arguments rely on the following property of the maximum risk (such as the maximum MSE) of an affine estimator: The maximum risk and maximum squared bias are attained at the same parameter values.
Since the maximum regret does not have this property, the arguments by \cite{donoho1994} cannot be applied to show \eqref{eq:property} for the minimax regret problem.

\subsubsection{Computation}\label{subsection:computation}
I conclude this section by summarizing the procedure for computing the minimax regret rule for a given problem $(L,\boldsymbol{m},\Theta,\sigma^2\boldsymbol{I}_n)$.
For a general variance $\boldsymbol{\Sigma}$, the minimax regret rule can be obtained by applying the procedure after normalizing $\boldsymbol{Y}$ and $(L,\boldsymbol{m},\Theta,\boldsymbol{\Sigma})$ to $\boldsymbol{\Sigma}^{-1/2}\boldsymbol{Y}$ and $(L,\boldsymbol{\Sigma}^{-1/2}\boldsymbol{m},\Theta,\boldsymbol{I}_n)$, respectively.
The procedure consists of the following steps:
\begin{enumerate}
    \item Compute $\epsilon^*\in \arg\max_{\epsilon\in[0,\tau^*\sigma]}\omega(\epsilon)\Phi(-\epsilon/\sigma)$.\footnote{For numerical optimization, possible algorithms include grid search or ternary search for finding a maximum of a unimodal function. The bisection method can also be used if a closed-form expression for $\omega'(\cdot)$ is available; see Lemma \ref{lemma:dif_omega} in Appendix \ref{appendix:lemma} for the differentiability of $\omega(\cdot)$ and Appendix \ref{appendix:computation} for the bisection method.}
    For a given $\epsilon\ge 0$, $\omega(\epsilon)$ is computed by solving the convex optimization problem \eqref{eq:def-mod}.
    Even if $\theta$ is infinite dimensional, the problem may be reduced to a finite-dimensional problem, as illustrated in Section \ref{section:app}.
    If the upper bound $\bar I(\boldsymbol{\mu})$ is analytically tractable, $\omega(\epsilon)$ can also be computed by solving the convex optimization problem \eqref{eq:mu}.
    \item If $\epsilon^*>0$ (i.e., $\sigma\omega'(0)> 2\phi(0)\omega(0)$), then solve \eqref{eq:def-mod} for $\epsilon=\epsilon^*$ to calculate $\theta_{\epsilon^*}$ and compute $\boldsymbol{m}(\theta_{\epsilon^*})$.
    Alternatively, $\boldsymbol{m}(\theta_{\epsilon^*})$ can be directly computed as a solution to \eqref{eq:mu} for $\epsilon=\epsilon^*$. 
    \item If $\epsilon^*=0$ (i.e., $\sigma\omega'(0)\le 2\phi(0)\omega(0)$), then compute $\boldsymbol{w}^*$, $\omega(0)$, and $\omega'(0)$.
    $\boldsymbol{w}^*$ can be computed by using the definition given by \eqref{eq:wstar} or one of the characterizations \eqref{eq:wstar2} and \eqref{eq:wstar3} in Appendix \ref{subsection:characterization}.
    $\omega'(0)$ can be analytically computed if an analytical expression of $\omega(\epsilon)$ is available for any sufficiently small $\epsilon\ge 0$.
    Alternatively, $\omega'(0)$ can be computed using the characterization in \eqref{eq:omega-dstar} in Appendix \ref{subsection:characterization}.
    \item Compute the decision rule given by \eqref{eq:mmr} in Theorem \ref{theorem:main}.
\end{enumerate}

\section{Relation to Optimal Estimation}\label{section:relation}

In this section, I study the relationship between optimal treatment choice and optimal estimation.
Given an estimator $\hat L$ of the welfare contrast $L(\theta)$, a decision rule can be constructed by plugging the estimator into the oracle optimal decision $\mathbf{1}\{L(\theta)\ge 0\}$: $\delta(\boldsymbol{Y})=\mathbf{1}\{\hat L(\boldsymbol{Y})\ge 0\}$.
Such a rule is called a {\it plug-in} rule.
For example, a plug-in rule can be constructed by using an estimator of $L(\theta)$ that is optimal under some standard criterion for estimation, such as minimax MSE optimality.
The minimax regret rule in Theorem \ref{theorem:main} can also be viewed as a plug-in rule, which uses a (possibly randomized) estimator.
A natural question is whether the estimator used in the minimax regret rule is optimal in a certain sense.
Another question is how the estimator used in the minimax regret rule differs from optimal estimators under standard criteria.
This section aims to answer these two questions.

As a preliminary step, Section \ref{section:low} presents a class of estimators that optimally trade off bias and variance in the estimation of $L(\theta)$.
Using the results, Section \ref{section:interpretation} discusses an interpretation of the minimax regret rule as a plug-in rule based on an estimator that satisfies a certain optimality.
In Section \ref{section:donoho}, I compare this estimator with a minimax affine MSE estimator, which is an existing optimal estimator in the setting of this paper.

Throughout Section \ref{section:relation}, I normalize $\boldsymbol\Sigma=\sigma ^2\boldsymbol I_n$ for some $\sigma>0$ as in Section \ref{section:minimax}.
I further assume that $\omega(\cdot)$ is differentiable on $(0,\infty)$ to simplify the presentation and proof of the results. The results can be modified to allow for nondifferentiability of $\omega(\cdot)$ by using the superdifferentials of $\omega(\cdot)$, which exist by the concavity of $\omega(\cdot)$.
I also note that $\omega(\cdot)$ is differentiable in Example \ref{example:intersection} and for eligibility cutoff choice in Section \ref{section:app}.
See Lemma \ref{lemma:dif_omega} in Appendix \ref{appendix:lemma} for a sufficient condition for the differentiability.


\subsection{Optimal Bias and Variance Tradeoff}\label{section:low}

As a basis for the discussion in Sections \ref{section:interpretation} and \ref{section:donoho}, I introduce a class of estimators that optimally trade off bias and variance.
Let ${\cal C}$ denote the class of all (nonrandomized) estimators for $L(\theta)$ (i.e., measurable functions from $\mathbb{R}^n$ to $\mathbb{R}$).
For estimator $\tilde L\in {\cal C}$, let ${\rm Bias}(\tilde L,\theta)\coloneqq\mathbb{E}_\theta[\tilde L(\boldsymbol{Y})]-L(\theta)$ and $\Var(\tilde L,\theta)\coloneqq\mathbb{E}_\theta[(\tilde L(\boldsymbol{Y})-\mathbb{E}_\theta[\tilde L(\boldsymbol{Y})])^2]$.
For scalar $V\ge 0$, let ${\cal C}(V)$ denote the class of estimators with the maximum variance over $\Theta$ less than or equal to $V$:
${\cal C}(V)\coloneqq\{\tilde L\in {\cal C}:\sup_{\theta\in\Theta}\Var(\tilde L,\theta)\le V\}$.
Consider the following minimax problem:
$$
\inf_{\tilde L\in{\cal C}(V)}\sup_{\theta\in\Theta}{\rm Bias}(\tilde L,\theta)^2.
$$
Solving this problem yields a class of estimators indexed by $V$ that optimally trade off the maximum squared bias and the maximum variance.

\cite{low1995tradeoff} derives estimators that achieve minimax optimality in the above sense for infinite-dimensional Gaussian models.
In Theorem \ref{theorem:low} in Appendix \ref{proof:low}, I extend the result of \cite{low1995tradeoff} to the multivariate Gaussian models in this paper.
The following is a corollary of Theorem \ref{theorem:low}, which translates a class of optimal estimators indexed by $V$ in Theorem \ref{theorem:low} into a class of optimal estimators indexed by $\epsilon\ge 0$.
For $\epsilon\ge 0$, define
        \begin{align*}
        V_\epsilon\coloneqq (\sigma\omega'(\epsilon))^2~~~\text{ and }~~~
            \hat L_{\epsilon}(\boldsymbol{Y})\coloneqq\begin{cases}
            \omega'(0)(\boldsymbol{w}^*)'\boldsymbol{Y} & ~~\text{if } \epsilon=0,\\
            \omega'(\epsilon)\frac{\boldsymbol{m}(\theta_{\epsilon})'}{\|\boldsymbol{m}(\theta_{\epsilon})\|}\boldsymbol{Y} & ~~\text{if } \epsilon>0,
        \end{cases}
        \end{align*}
assuming that $\omega(\cdot)$ is differentiable on $(0,\infty)$ and $\theta_{\epsilon}\in\Theta$ attains the modulus of continuity at $\epsilon$ with $\|\boldsymbol m(\theta_{\epsilon})\|=\epsilon$.
\begin{theorem}[Minimax Optimality of $\hat L_\epsilon$]\label{theorem:minimax-bias}
    Suppose that Assumption \ref{assumption:problem} holds; that $\omega'(0)>0$; that $\omega(\cdot)$ is differentiable on $(0,\infty)$; and that for each $\epsilon> 0$, $\theta_{\epsilon}\in\Theta$ attains the modulus of continuity at $\epsilon$ with $\|\boldsymbol m(\theta_{\epsilon})\|=\epsilon$.
    Then, the following holds.
    \begin{enumerate}[label=(\roman*)]
        \item For each $\epsilon\ge 0$, $\hat L_\epsilon$ has a constant variance of $V_\epsilon$:
        $\Var(\hat L_\epsilon,\theta)=V_\epsilon$ for all $\theta\in\Theta$.
        \item For each $\epsilon\ge 0$, $\hat L_\epsilon$ minimizes the maximum squared bias among all estimators with the maximum variance less than or equal to $V_\epsilon$:
        $$\sup_{\theta\in\Theta}{\rm Bias}(\hat L_{\epsilon},\theta)^2=\inf_{\tilde L\in {\cal C}(V_\epsilon)}\sup_{\theta\in\Theta}{\rm Bias}(\tilde L,\theta)^2.$$
        Furthermore, $\hat L_{0}$ minimizes the maximum squared bias among all estimators:
        $$\sup_{\theta\in\Theta}{\rm Bias}(\hat L_{0},\theta)^2=\inf_{\tilde L\in {\cal C}}\sup_{\theta\in\Theta}{\rm Bias}(\tilde L,\theta)^2.$$
        \item As $\epsilon$ increases, the maximum squared bias of $\hat L_\epsilon$ weakly increases and the variance of $\hat L_\epsilon$ weakly decreases.
    \end{enumerate}
\end{theorem}
\begin{proof}
	See Appendix \ref{proof:low}.
\end{proof}

Theorem \ref{theorem:minimax-bias} shows that the linear estimator $\hat L_{\epsilon}$ minimizes the maximum squared bias among all estimators (including nonlinear ones) with variance bounded by $V_\epsilon$.
Furthermore, Theorem \ref{theorem:minimax-bias} shows that the linear estimator $\hat L_{0}$ achieves the minimum maximum squared bias among all estimators. In contrast, \cite{low1995tradeoff} does not provide a minimax squared bias estimator when variance constraints are absent in infinite-dimensional Gaussian models.\footnote{Specifically, Theorem 2 in \cite{low1995tradeoff} does not provide a minimax squared bias estimator for the range of variance bound $V$ for which $0$ is the unique maximizer of $(\omega(\epsilon)-\epsilon\sqrt{V}/\sigma)$ over $\epsilon\ge 0$.}

\subsection{Interpreting the Minimax Regret Rule as a Plug-in Rule}\label{section:interpretation}

Theorem \ref{theorem:minimax-bias} provides an interpretation of the minimax regret rule $\delta^*$ in Theorem \ref{theorem:main}.
If $2\phi(0)\frac{\omega(0)}{\omega'(0)}<\sigma$, the minimax regret rule is given by $\delta^*(\boldsymbol Y)=\mathbf{1}\{\boldsymbol{m}(\theta_{\epsilon^*})'\boldsymbol Y\ge 0\}=\mathbf{1}\{\hat L_{\epsilon^*}(\boldsymbol Y)\ge 0\}$, where $\epsilon^*\in\arg\max_{\epsilon\in [0,\tau^*\sigma]}\omega(\epsilon)\Phi(-\epsilon/\sigma)$.
This corresponds to a plug-in rule based on the linear estimator $\hat L_{\epsilon^*}$, which minimizes the maximum squared bias among all estimators with variance bounded by $V_{\epsilon^*}$.
In contrast, if $2\phi(0)\frac{\omega(0)}{\omega'(0)}>\sigma$, the minimax regret rule is given by $\delta^*(\boldsymbol Y)=\Phi\left(\frac{(\boldsymbol{w}^*)'\boldsymbol{Y}}{((2\phi(0)\omega(0)/\omega'(0))^2-\sigma^2)^{1/2}}\right)=\Phi\left(\frac{\hat L_0(\boldsymbol{Y})}{((2\phi(0)\omega(0))^2-(\sigma\omega'(0))^2)^{1/2}}\right)$. Equivalently, this can be written as $\delta^*(\boldsymbol Y)=\mathbb{P}\left(\hat L_0(\boldsymbol{Y})+\xi\ge 0|\boldsymbol{Y}\right)$, where $\xi|\boldsymbol{Y}\sim {\cal N}(0,(2\phi(0)\omega(0))^2-(\sigma\omega'(0))^2)$.
Thus, this rule can be interpreted as a plug-in rule based on the {\it randomized} estimator $\hat L_0(\boldsymbol{Y})+\xi$ for $L(\theta)$. This estimator is constructed by adding a mean-zero Gaussian noise to the linear estimator $\hat L_0$, which minimizes the maximum squared bias among all estimators.

The above interpretation highlights a key distinction between minimax regret treatment choice and minimax estimation.
The fact that a plug-in rule based on a randomized estimator can be minimax regret suggests that increasing the variance of an estimator while holding the bias constant may reduce the maximum regret of the resulting plug-in rule.
This means that the maximum regret cannot be expressed as an increasing function of the maximum squared bias and the variance.\footnote{This observation aligns with a result from \cite{ishihara2021meta}, who show that in their setting, the maximum regret of a 
linear threshold rule $\delta(\boldsymbol{Y})=\mathbf{1}\{\boldsymbol{w}'\boldsymbol{Y}\ge 0\}$ with $\sum_{i=1}^nw_i=1$ can be written as a function of the maximum absolute bias and the variance of $\boldsymbol{w}'\boldsymbol{Y}$, though the function is not monotonic in variance.}
This observation sharply contrasts with certain performance measures used in minimax affine estimation.
For example, the maximum MSE of an affine estimator $\tilde L(\boldsymbol{Y})=c+\boldsymbol{w}'\boldsymbol{Y}$ is given by
$\sup_{\theta\in \Theta}\mathbb{E}_\theta[(\tilde L(\boldsymbol{Y})-L(\theta))^2]=\sup_{\theta\in\Theta}{\rm Bias}(\tilde L,\theta)^2+\sigma^2\|\boldsymbol{w}\|^2$,
which is increasing in both the maximum squared bias and the variance.

\subsection{Comparison with Minimax Affine MSE Estimator}\label{section:donoho}

To shed further light on the distinction between minimax regret treatment choice and minimax estimation, I compare the linear estimator used by the minimax regret rule with a minimax affine MSE estimator.
To introduce the latter,
let ${\cal C}_{\rm affine}$ denote the class of all affine estimators of $L(\theta)$: ${\cal C}_{\rm affine}\coloneqq\{\tilde L\in {\cal C}:\tilde L(\boldsymbol{Y})=c+\boldsymbol{w}'\boldsymbol{Y}, c\in\mathbb{R},\boldsymbol{w}\in\mathbb{R}^n\}$.
An estimator $\hat L\in {\cal C}_{\rm affine}$ is {\it minimax affine MSE} if
$\sup_{\theta\in \Theta}\mathbb{E}_\theta[(\hat L(\boldsymbol{Y})-L(\theta))^2]=\inf_{\tilde L\in {\cal C}_{\rm affine}}\sup_{\theta\in \Theta}\mathbb{E}_\theta[(\tilde L(\boldsymbol{Y})-L(\theta))^2]$.

Given the results in Theorem \ref{theorem:minimax-bias}, one approach to finding a minimax affine MSE estimator is to minimize the maximum MSE of $\hat{L}_\epsilon$ over $\epsilon \geq 0$. However, for comparison with the minimax regret rule, the following theorem instead relies on an alternative characterization of the minimax affine MSE estimator based on \citeauthor{donoho1994}'s \citeyearpar{donoho1994} approach.


\begin{theorem}[Comparison with Minimax Affine MSE Estimator]\label{theorem:donoho}
    Suppose that Assumption \ref{assumption:problem} holds, that $\omega'(0)>0$, and that $\omega(\cdot)$ is differentiable on $(0,\infty)$.
    Also, suppose that there exists $\epsilon_{\rm MSE}\in\arg\max_{\epsilon\ge 0}\frac{\sigma^2}{\epsilon^2+\sigma^2}\omega(\epsilon)^2$ and that $\theta_{\epsilon_{\rm MSE}}\in\Theta$ attains the modulus of continuity at $\epsilon_{\rm MSE}$ with $\|\boldsymbol m(\theta_{\epsilon_{\rm MSE}})\|=\epsilon_{\rm MSE}$.
    Then, the following holds.
    \begin{enumerate}[label=(\roman*)]
        \item\label{theorem:donoho-original} $\epsilon_{\rm MSE}>0$. Furthermore, $\hat L_{\epsilon_{\rm MSE}}(\boldsymbol{Y})=\omega'(\epsilon_{\rm MSE})\frac{\boldsymbol{m}(\theta_{\epsilon_{\rm MSE}})'}{\|\boldsymbol{m}(\theta_{\epsilon_{\rm MSE}})\|}\boldsymbol{Y}$ is a minimax affine MSE estimator.
        \item\label{theorem:donoho-comparison} Let $\epsilon^*\in\arg\max_{\epsilon\in [0,\tau^*\sigma]}\omega(\epsilon)\Phi(-\epsilon/\sigma)$. Then, $\epsilon_{\rm MSE}>\epsilon^*$.
    \end{enumerate}
\end{theorem}

\begin{proof}
	See Appendix \ref{proof:donoho}.
\end{proof}

Theorem \ref{theorem:donoho}\ref{theorem:donoho-original} follows from the results in \cite{donoho1994}.
The optimization problem $\max_{\epsilon\ge 0}\frac{\sigma^2}{\epsilon^2+\sigma^2}\omega(\epsilon)^2$ corresponds to calculating the largest minimax affine MSE over all one-dimensional subfamilies.
The estimator $\hat L_{\epsilon_{\rm MSE}}$ is shown to be minimax affine MSE for the hardest one-dimensional subproblem $[-\theta_{\epsilon_{\rm MSE}},\theta_{\epsilon_{\rm MSE}}]$ and also for the full problem $\Theta$.
Furthermore, the optimal level $\epsilon_{\rm MSE}$ is always positive, regardless of the degree of partial identification or the noise level $\sigma$.
This implies that, in the hardest one-dimensional subproblem for the minimax affine MSE problem, the sample $\boldsymbol{Y}$ is informative about the sign of $L(\theta)$, unlike in the minimax regret problem, where $\boldsymbol{Y}$ can be completely uninformative.

Theorem \ref{theorem:donoho}\ref{theorem:donoho-comparison} compares the optimal level of $\epsilon$ for the minimax affine MSE and minimax regret problems.
Together with Theorem \ref{theorem:minimax-bias}, this result implies that
$\hat L_{\epsilon^*}$ has a maximum squared bias that is no larger and a variance that is no smaller than those of $\hat L_{\epsilon_{\rm MSE}}$.
\footnote{\cite{ishihara2021meta} derive a related result using a different argument in their setting.}

\section{Application to Eligibility Cutoff Choice}\label{section:app}


	In many policy domains, the eligibility for treatment is determined based on an individual's observable characteristics.
    In this section, I demonstrate how my framework can be valuable for using data collected under the status quo eligibility criterion to decide whether to change it to a new one.
    This approach does not require conducting a randomized experiment that directly evaluates the performance of the status quo and new criteria.
    
\subsection{Setup}\label{section:example}
    Consider the following special case of Example \ref{example:unconfoundedness}.
	For each unit $i=1,..,n$, we observe a fixed running variable $x_i\in\mathbb{R}$, a binary treatment status $d_i\in\{0,1\}$, and an outcome $Y_i\in\mathbb{R}$.
	The eligibility for treatment is determined based on whether the running variable exceeds a specific cutoff $c_0\in \mathbb{R}$, so that $d_i=\mathbf{1}\{x_i\ge c_0\}$.
	Suppose
	\begin{align*}
	Y_i=f(x_i,d_i)+U_i,~~~U_i\sim {\cal N}(0,\sigma^2(x_i,d_i)) ~\text{independent across $i$}, 
	\end{align*}
	where $f:\mathbb{R}\times\{0,1\}\rightarrow \mathbb{R}$ is an unknown function and plays the role of the parameter $\theta$ and $\sigma^2(x_i,d_i)>0$.
	We interpret $f(x,d)$ as the conditional mean potential outcome under treatment status $d\in\{0,1\}$ given $x$.
	\sloppy
    We can write the model in a vector form
	$\boldsymbol Y \sim{\cal N}(\boldsymbol m(f), \boldsymbol\Sigma)$,
	where $\boldsymbol Y=(Y_1,...,Y_n)'$, $\boldsymbol m(f)=(f(x_1,d_1),...,f(x_n,d_n))'$, and $\boldsymbol\Sigma={\rm diag}(\sigma^2(x_1,d_1),...,\sigma^2(x_n,d_n))$.
	
	Now, suppose we are interested in changing the cutoff from $c_0$ to a specific value $c_1$.
	For illustration, assume $c_1<c_0$.
	The welfare under the cutoff $c_a$, $a\in\{0,1\}$, is given by
	$$
	W_a(f) = \int [f(x,1)\mathbf{1}\{x\ge c_a\}+f(x,0)\mathbf{1}\{x< c_a\}]d\nu(x)
	$$
	for some known measure $\nu$.
        An implicit assumption made here is that $f$ is invariant to the cutoff change.
    For an illustration of the results, I use an empirical measure as $\nu$, for which the welfare is the unweighted sample average:
	$
	W_a(f) = \frac{1}{n}\sum_{i=1}^n[f(x_i,1)\mathbf{1}\{x_i\ge c_a\}+f(x_i,0)\mathbf{1}\{x_i< c_a\}].
	$
    I define the welfare contrast between the two cutoffs as
	$$
	L(f)=\frac{n}{\tilde n}(W_1(f)-W_0(f))=\frac{1}{\tilde n}\sum_{i=1}^n\mathbf{1}\{c_1\le x_i< c_0\}[f(x_i,1)-f(x_i,0)],
	$$
    where $\tilde n=\sum_{i=1}^n\mathbf{1}\{c_1\le x_i<c_0\}$ denotes the number of units between the two cutoffs $c_1$ and $c_0$, whose treatment status would be changed if the cutoff were changed.
    I scale $W_1(f)-W_0(f)$ by $n/\tilde n$ so that $L(f)$ represents the sample average treatment effect for these units.
    This scaling does not change the form of a minimax regret rule.
    Figure \ref{fig:RD} presents an example of this setting.
    Panel (a) shows an example of the conditional mean potential outcome function $f$.
    In this example, the welfare contrast is given by $L(f)=\frac{1}{2}\sum_{i\in\{2,3\}}[f(x_i,1)-f(x_i,0)]$, which is the sample average treatment effect for $x_2$ and $x_3$.

\begin{figure}[!t]
	\centering
	\caption{Illustration of Eligibility Cutoff Choice} 
	\begin{subfigure}[b]{0.49\textwidth}   
		\centering 
		\includegraphics[width=\textwidth]{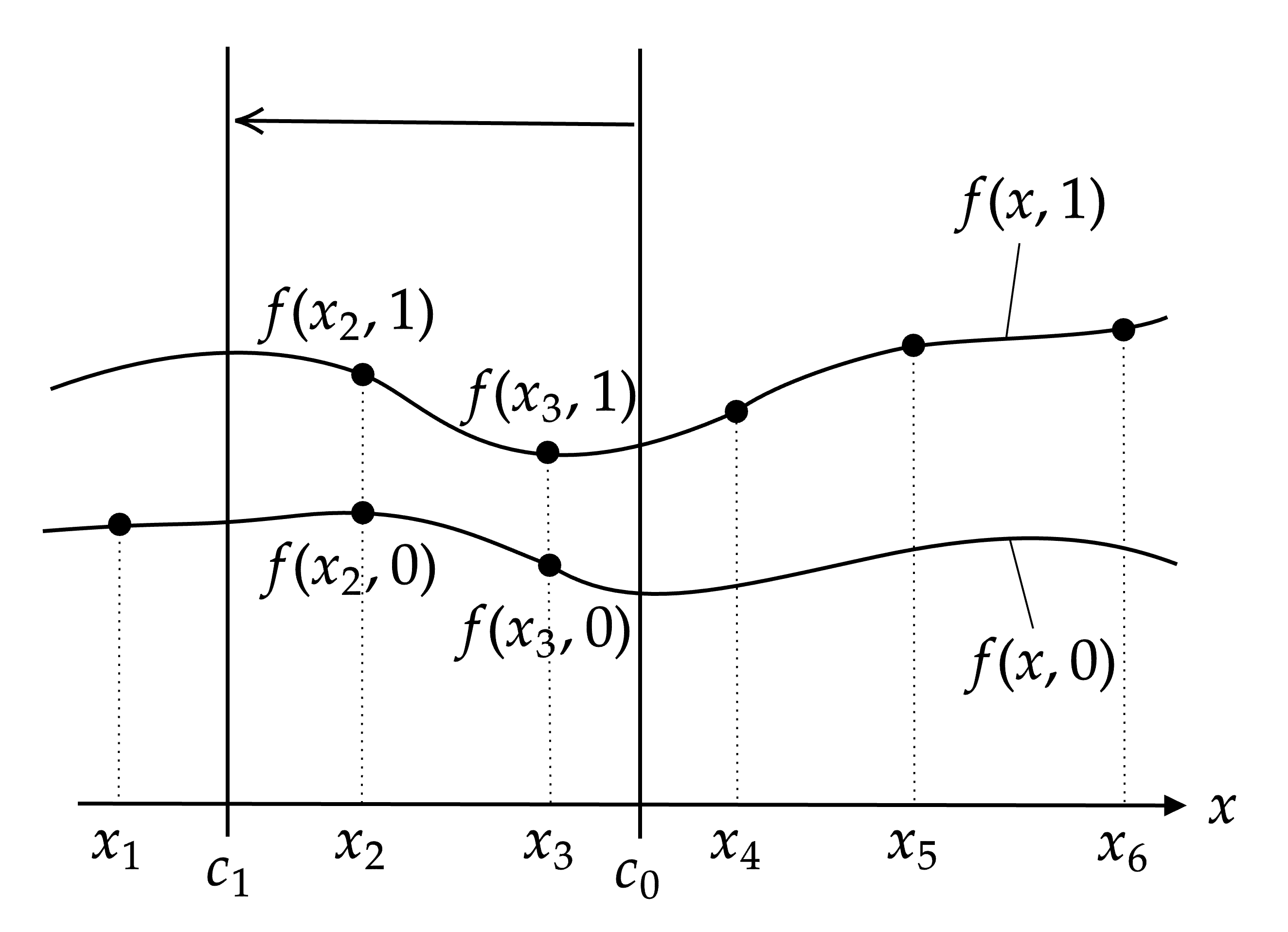}
		\caption[]%
		{Conditional Mean Potential Outcome $f$} 
		\label{subfig:RD1}
	\end{subfigure}
	\begin{subfigure}[b]{0.49\textwidth}   
		\centering 
		\includegraphics[width=\textwidth]{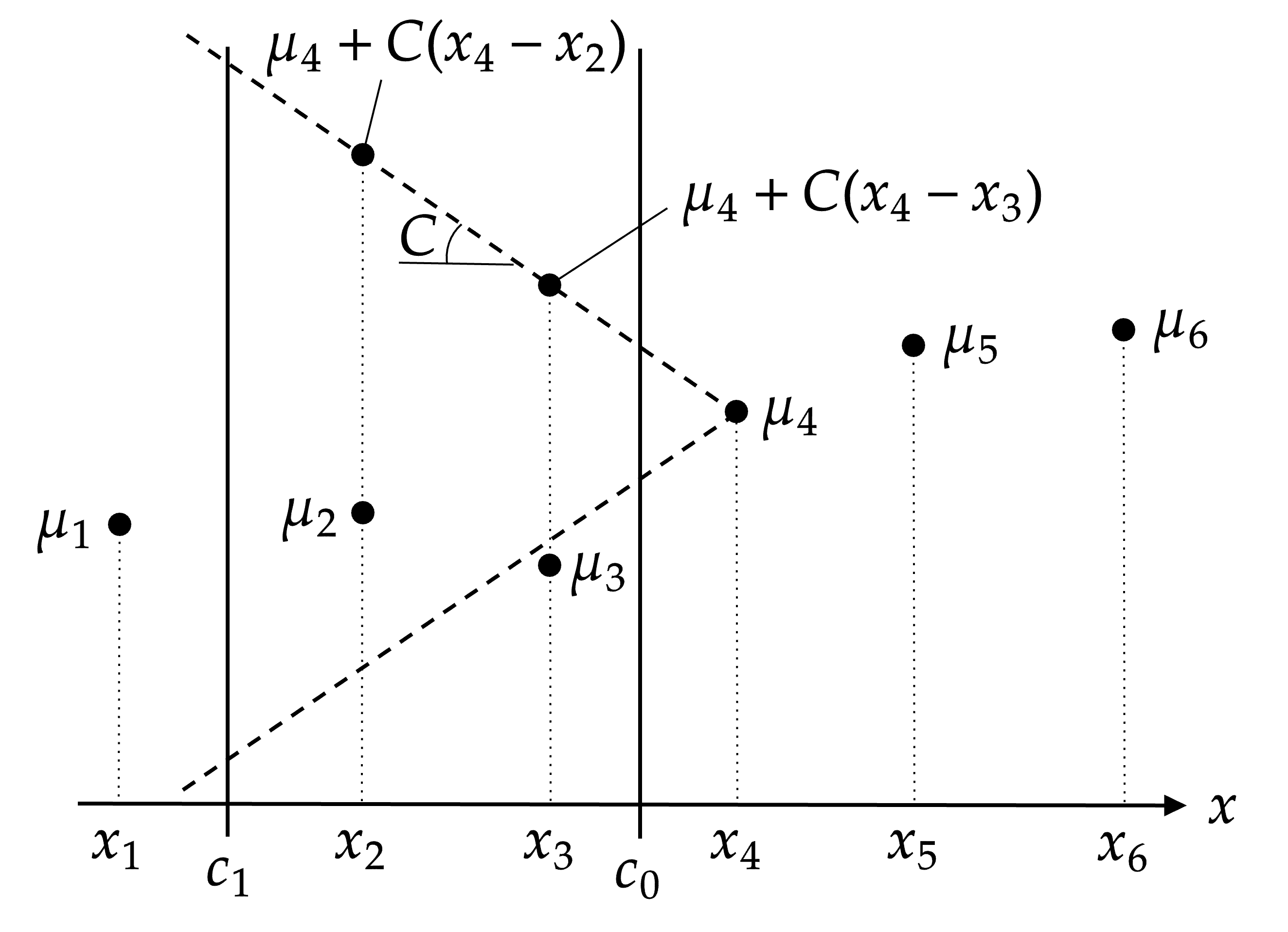}
		\caption[]%
		{Bounds on $f(x,1)$ for $x<0$}
		\label{subfig:RD2}
	\end{subfigure}
	
	\caption*{\scriptsize{\it Notes}: This figure shows an example of the setting of eligibility cutoff choice described in Section \ref{section:example}.
    The two vertical lines indicate the original and new cutoffs $c_0$ and $c_1$.
    Panel (a) shows an example of the conditional mean potential outcome function $f$.
    In Panel (b), the two dashed lines indicate the upper and lower bounds on the function $f(x,1)$ on the range of $x<x_4$ under the assumption that $f\in{\cal F}_{\rm Lip}(C)$.}
	\label{fig:RD}
\end{figure}

	To point or partially identify $L(f)$, suppose that $f\in{\cal F}$ for some function class ${\cal F}$.
	Here, I focus on the {\it Lipchitz class} 
	$$
	{\cal F}_{\rm Lip}(C)=\{f:|f(x,d)-f(\tilde x,d)|\le C|x-\tilde x| \text{ for every } x, \tilde x\in\mathbb{R} \text{ and } d\in\{0,1\}\}.
	$$
	The Lipschitz constraint bounds the maximum possible change in $f(x,d)$ in response to a shift in $x$ by one unit.
    I assume the Lipschitz constant is common for $f(\cdot,1)$ and $f(\cdot,0)$ for simplicity; it is possible to impose two separate constants.
    Other possible function classes include the class of functions with a known bound on the second derivative, as used by \cite{Imbens2019RDD} for inference in RD designs.

    To illustrate how the Lipschitz constraint allows one to partially identify $L(f)$, I present the upper bound on $L(f)$.
    The lower bound can be obtained analogously.
    Let ${\cal M}=\{\boldsymbol m(f):f\in {\cal F}_{\rm Lip}(C)\}$ and $x_{+,{\rm min}}=\min\{x_i:x_i\ge c_0\}$ be the value of $x$ of the treated unit closest to the original cutoff $c_0$.
    The upper bound on $L(f)$ when $\boldsymbol m(f)=\boldsymbol{\mu}\in {\cal M}$ is given by
    \begin{align*}
        \bar I(\boldsymbol{\mu})&=\sup\{L(f):f(x_i,d_i)=\mu_i \text{ for } i=1,...,n, f\in {\cal F}_{\rm Lip}(C)\}\\
        &=\sup\left\{\frac{1}{\tilde n}\sum_{i:c_1\le x_i< c_0}[f(x_i,1)-\mu_i]:f(x_i,1)=\mu_i \text{ for $i$ with $x_i\ge c_0$}, f\in {\cal F}_{\rm Lip}(C)\right\}\\
        &=\frac{1}{\tilde n}\sum_{i:c_1\le x_i< c_0}[\mu_{+,\min}+C(x_{+,{\rm min}} - x_i)-\mu_i],
    \end{align*}
    where I define $\mu_{+,\min}=\mu_i$ for the unit $i$ with $x_i=x_{+,\min}$.
    The second equality holds, since $f(x_i,0)=\mu_i$ for any $i$ with $x_i<c_0$ (i.e., $d_i=0$) and $f(x_i,1)=\mu_i$ for any $i$ with $x_i\ge c_0$ (i.e., $d_i=1$). The last equality holds, since the upper bound on $f(x_i,1)$ for any unit $i$ with $x_i<c_0$ is shown to be $\mu_{+,\min}+C(x_{+,{\rm min}} - x_i)$ under the Lipschitz constraint.
    The upper bound $\bar I(\boldsymbol{\mu})$ increases with the Lipschitz constant $C$ and weakly increases with the size of cutoff change $|c_1-c_0|$ (holding $c_0$ fixed).
    In the example presented in Figure \ref{fig:RD}, $x_4$ is the treated unit closest to the original cutoff $c_0$.
    The two dashed lines in Panel (b) indicate the upper and lower bounds on the function $f(x,1)$ on the range of $x<x_4$.
    In this example, the upper bound on $L(f)$ is given by $\bar I(\boldsymbol{\mu})=\frac{1}{2}\sum_{i\in\{2,3\}}[\mu_4+C(x_4 - x_i)-\mu_i]$.

    Note that we are not interested in the identified set per se, but are interested in using the sample $\boldsymbol{Y}$ to choose between the two cutoffs given the specified function class ${\cal F}$.
    For a decision rule $\delta:\mathbb{R}^n\rightarrow[0,1]$, $\delta(\boldsymbol y)\in [0,1]$ represents the probability of changing the cutoff from $c_0$ to $c_1$ when the realized sample is $\boldsymbol y$.
    Alternatively, we can interpret $\delta(\boldsymbol y)$ as the fraction of individuals to whom we would assign treatment within the units between $c_1$ and $c_0$.
    In Section \ref{section:app-result}, I derive a minimax regret rule when the welfare is the sample average outcome and ${\cal F}={\cal F}_{\rm Lip}(C)$.
    The form of the rule depends on the empirical distribution of $x_i$, the two cutoffs $c_0$ and $c_1$, the Lipschitz constant $C$, and the conditional variance $\sigma^2(x_i,d_i)$, all of which are treated as known.
    In practice, the policymaker must specify $C$ and $\sigma^2(x_i,d_i)$ to implement the rule.
    In Section \ref{section:implementation}, I provide practical guidance on how to specify them.

\subsection{Minimax Regret Rule}\label{section:app-result}


To apply the results in Section \ref{section:minimax}, I normalize $\tilde{\boldsymbol Y}=\boldsymbol\Sigma^{-1/2}\boldsymbol Y=(Y_1/\sigma(x_1,d_1),...,Y_n/\sigma(x_n,d_n))'$ and $\tilde{\boldsymbol m}(f)=\boldsymbol\Sigma^{-1/2}\boldsymbol m(f)=(f(x_1,d_1)/\sigma(x_1,d_1),...,f(x_n,d_n)/\sigma(x_n,d_n))'$, so that $\tilde{\boldsymbol Y} \sim{\cal N}(\tilde{\boldsymbol m}(f), \boldsymbol I_n)$.
Then, $\omega(\epsilon)=\sup\{L(f): \|\tilde {\boldsymbol m}(f)\|\le \epsilon,f\in{\cal F}_{\rm Lip}(C)\}$ is the value of
\begin{align}
	\sup_{f\in{\cal F}_{\rm Lip}(C)}\frac{1}{\tilde n}\sum_{i=1}^n\mathbf{1}\{c_1\le x < c_0\}[f(x_i,1)-f(x_i,0)]~~s.t.~~\sum_{i=1}^n\frac{f(x_i,d_i)^2}{\sigma^2(x_i,d_i)}\le\epsilon^2.\label{eq:mod_lip}
\end{align}
The unknown parameter $f$ is infinite dimensional, but the objective and the norm constraint $\sum_{i=1}^n\frac{f(x_i,d_i)^2}{\sigma^2(x_i,d_i)}\le\epsilon^2$ depend on $f$ only through its values at $(x_1,0),...,(x_n,0),(x_1,1),...,(x_n,1)$.
By a slight modification of Theorem 2.2 in \cite{Armstrong2021ATE}, this optimization problem can be reduced to the following problem:
\begin{align}
	&\max_{(f(x_i,0),f(x_i,1))_{i=1,...,n}\in\mathbb{R}^{2n}} \frac{1}{\tilde n}\sum_{i=1}^n\mathbf{1}\{c_1\le x_i < c_0\}[f(x_i,1)-f(x_i,0)] \label{eq:mod_lip_finite}\\
	s.t. & ~~\sum_{i=1}^n\frac{f(x_i,d_i)^2}{\sigma^2(x_i,d_i)}\le \epsilon^2, ~f(x_i,d)-f(x_j,d)\le C|x_i-x_j|, ~~d\in\{0,1\}, i,j\in\{1,...,n\}.\nonumber
\end{align}
A solution to (\ref{eq:mod_lip_finite}) exists, since the objective function is continuous and the set of the vectors of $2n$ unknowns that satisfy the constraints is closed and bounded.
Once we find a solution $(f(x_i,0),f(x_i,1))_{i=1,...,n}$, we can always find a function $f\in{\cal F}_{\rm Lip}(C)$ that interpolates the points $(x_i,f(x_i,0)),(x_i,f(x_i,1))$, $i=1,...,n$ \citep[Theorem 4]{BELIAKOV2006lipschitz}, which is a solution to the original problem (\ref{eq:mod_lip}).
Problem (\ref{eq:mod_lip_finite}) is a finite-dimensional convex optimization problem with $2n$ unknowns, one quadratic and $2n(n-1)$ linear constraints, and a linear objective function, and can be solved using off-the-shelf convex optimization packages.\footnote{In the empirical application in Section \ref{section:empirical}, I use CVXPY, a Python-embedded modeling language for convex optimization problems \citep{diamond2016cvxpy,agrawal2018rewriting}.}

The following result derives a minimax regret rule.

\begin{proposition}[Minimax Regret Rule for Eligibility Cutoff Choice]\label{prop:eligibility}
    Consider the setup in Section \ref{section:example} with $L(f)=\frac{1}{\tilde n}\sum_{i:c_1\le x_i< c_0}[f(x_i,1)-f(x_i,0)]$ and ${\cal F}={\cal F}_{\rm Lip}(C)$.
    For simplicity, suppose $x_i\neq x_j$ for any $i\neq j$.
    Let $\omega(\epsilon)$ be the value of (\ref{eq:mod_lip_finite}) for $\epsilon\ge 0$, $\epsilon^*\in\arg\max_{\epsilon\in[0,\tau^*]}\omega(\epsilon)\Phi(-\epsilon)$, and $(f_{\epsilon^*}(x_i,0),f_{\epsilon^*}(x_i,1))_{i=1,...,n}$ solve (\ref{eq:mod_lip_finite}) for $\epsilon=\epsilon^*$.
    Then, the following decision rule is minimax regret:
    \begin{align*}
    \delta^*(\boldsymbol Y)=\begin{cases}
    \mathbf{1}\left\{\sum_{i=1}^nf_{\epsilon^*}(x_i,d_i)Y_i/\sigma^2(x_i,d_i)\ge 0\right\}~~&\text{ if }  ~s^*<\bar\sigma,\\
    \mathbf{1}\left\{Y_{+,\min}-\frac{1}{\tilde n}\sum_{i:c_1\le x_i<c_0}Y_i\ge 0\right\}~~&\text{ if } ~s^*=\bar\sigma,\\
    \Phi\left(\dfrac{Y_{+,\min}-\frac{1}{\tilde n}\sum_{i:c_1\le x_i<c_0}Y_i}{((s^*)^2-\bar\sigma^2)^{1/2}}\right)~~&\text{ if } ~s^*>\bar\sigma,
    \end{cases}
\end{align*}
    where $x_{+,{\rm min}}=\min\{x_i:x_i\ge c_0\}$, $Y_{+,\min}=Y_i$ for the unit $i$ with $x_i=x_{+,\rm min}$, $\bar\sigma=(\sigma^2(x_{+,{\rm min}},1)+\frac{1}{\tilde n^2}\sum_{i:c_1\le x_i<c_0}\sigma^2(x_i,0))^{1/2}$, and
    $s^*=2\phi(0)C\frac{1}{\tilde n}\sum_{i:c_1\le x_i<c_0}[x_{+,{\rm min}} - x_i]$.
    Furthermore, the maximum regret of $\delta^*$ over ${\cal F}_{\rm Lip}(C)$ is attained at $-f^*$ and $f^*$, where $f^*$ is any function in ${\cal F}_{\rm Lip}(C)$ that interpolates the points $(x_i,f_{\epsilon^*}(x_i,0)),(x_i,f_{\epsilon^*}(x_i,1))$, $i=1,...,n$.
\end{proposition}

\begin{proof}
    In Appendix \ref{appendix:derivation-wstar}, I derive a solution to (\ref{eq:mod_lip_finite}) for any sufficiently small $\epsilon\ge0$, which provides closed-form expressions for $\omega(0)$, $\omega'(0)$, and $\boldsymbol{w}^*$.
    The results then follow from an application of Theorem \ref{theorem:main} and simple calculations.
\end{proof}

The minimax regret rule is randomized or nonrandomized, depending on $s^*$ and $\bar\sigma$. $s^*$ is 
increasing in the Lipschitz constant $C$ and nondecreasing in the size of cutoff change $|c_1-c_0|$.
$\bar\sigma$ is the standard deviation of $Y_{+,\min}-\frac{1}{\tilde n}\sum_{i:c_1\le x_i<c_0}Y_i$, and therefore increases with the variances of the treated unit closest to the status quo cutoff and the untreated units between the two cutoffs.
If the Lipschitz constant $C$ or the cutoff change is large relative to the variances so that $s^*>\bar\sigma$, the minimax regret rule is a randomized rule based on $Y_{+,\min}-\frac{1}{\tilde n}\sum_{i:c_1\le x_i<c_0}Y_i$.
In view of the results in Section \ref{section:low}, $Y_{+,\min}-\frac{1}{\tilde n}\sum_{i:c_1\le x_i<c_0}Y_i$ can be interpreted as an estimator of $L(f)$ that minimizes the maximum squared bias over ${\cal F}_{\rm Lip}(C)$ among all estimators.
As $C$ or the cutoff change increases, $(s^*)^2-\bar\sigma^2$ increases, and the decision is more randomized given the realization of the estimator.

On the other hand, if $s^*<\bar\sigma$, the minimax regret rule is a nonrandomized rule based on $\sum_{i=1}^nf_{\epsilon^*}(x_i,d_i)Y_i/\sigma^2(x_i,d_i)$.
Using the results in Appendix \ref{appendix:derivation-wstar}, we can show that if $s^*$ is marginally below $\bar\sigma$,
$\sum_{i=1}^nf_{\epsilon^*}(x_i,d_i)Y_i/\sigma^2(x_i,d_i)$ is proportional to $Y_{+,\min}-\frac{1}{\tilde n}\sum_{i:c_1\le x_i<c_0}Y_i$.
As the Lipschitz constant $C$ or cutoff change decreases or the variances increase so that $s^*$ becomes sufficiently smaller than $\bar\sigma$,
the weighted sum assigns nonzero weights to some of the units with $x_i<c_1$ or $x_i>x_{+,\rm min}$.
In Section \ref{section:empirical}, I numerically examine the relationship between the weights and the choice of $C$ in the empirical application; see Figure \ref{fig:weight}.

\subsection{Practical Issues}\label{section:implementation}

In practice, the conditional variance $\sigma^2(x_i,d_i)$ is unknown.
A feasible version of the minimax regret rule is obtained by using a consistent estimator in place of the true $\sigma^2(x_i,d_i)$.
The conditional variance can be estimated, for example, by applying a local linear regression to the squared residuals \citep{fan1998variance} or by the nearest-neighbor variance estimator \citep{abadie2006matching}.
In the case in which unit $i$ represents a group of individuals and $Y_i$ is the sample mean outcome within group $i$, as in the empirical application in Section \ref{section:empirical}, it is natural to use the conventional standard error of the sample mean as $\sigma(x_i,d_i)$.

Implementation of the minimax regret rule requires choosing the Lipschitz constant $C$.
In principle, it is not possible to choose $C$ that applies to both sides of the cutoff $c_0$ in a data-driven way, since we only observe outcomes either under treatment or under no treatment on each side.
It is, however, possible to estimate a lower bound on $C$.
If $f\in{\cal F}_{\rm Lip}(C)$ is differentiable, a lower bound on $C$ is given by $\max\left\{\max_{\tilde x\ge c_0}\left\vert\frac{\partial f(\tilde x,1)}{\partial x}\right\vert,\max_{\tilde x< c_0}\left\vert\frac{\partial f(\tilde x,0)}{\partial x}\right\vert\right\}$, since $\left\vert\frac{\partial f(\tilde x,d)}{\partial x}\right\vert\le C$ for all $\tilde x$ and $d$.
To estimate the lower bound, we could estimate the derivatives $\frac{\partial f(\tilde x,1)}{\partial x}$ for $\tilde x\ge c_0$ and $\frac{\partial f(\tilde x,0)}{\partial x}$ for $\tilde x< c_0$ by a local polynomial regression and then take the maximum of their absolute values.
In practice, I suggest choosing the initial value of $C$ by estimating the lower bound or using application-specific knowledge, and considering a range of plausible values of $C$ to conduct a sensitivity analysis.

\section{Empirical Application}\label{section:empirical}

I now illustrate my approach in an empirical application to consider whether to scale up the BRIGHT program in Burkina Faso.

\subsection{Background and Data}

With the aim of improving children's---especially girls'---educational outcomes in rural villages,
the BRIGHT program constructed well-resourced village-based schools with three classrooms for grades 1 to 3 in 132 villages from 47 departments\footnote{Departments are the third-level administrative divisions of Burkina Faso, below regions and provinces.} during the period 2005 to 2008.
The Ministry of Education determined the villages in which schools would be built through the following process.
First, 293 villages were nominated based on low school enrollment rates.
Second, the Ministry administered a survey in each village and assigned each village a score using a set formula.
	The formula attached a large weight to the estimated number of children to be served from the nominated and neighboring villages, giving additional weight to girls.
The Ministry then ranked villages within each department and selected the top half of the villages to receive a school.
For further details on the BRIGHT program and allocation process, see \cite{levy2009bright} and \cite{Kazianga2013bright}.

Since the school allocation was determined at department level, the cutoff score for program eligibility differed across departments.
Following \cite{Kazianga2013bright}, I define the {\it relative score} as the score for each village minus the cutoff score for the department the village belongs to.
As a result, a village is eligible for the program when the relative score is larger than zero.
\cite{Kazianga2013bright} use the relative score as a running variable and evaluate the causal effect of the program on educational outcomes using an RD design.

I use the replication data for \citeauthor{Kazianga2013bright}'s \citeyearpar{Kazianga2013bright} results \citep{Kazianga2019data} and consider whether we should expand the program.
The dataset contains survey results on 30 households from 287 nominated villages, for a total sample of 23,282 children between the ages of 5 and 12.
The survey was conducted in 2008---namely, 2.5 years after the start of the program.
Table \ref{table:summary} in Appendix \ref{appendix:empirical_result} reports summary statistics on child educational outcomes and characteristics.

I consider school enrollment as the target outcome.
Since the score and eligibility are determined at village level, I use the village-level mean outcome---namely, the enrollment rate for each village.
This setting fits into the setup in Section \ref{section:example}, where $i$ represents a village, $Y_i$ is the sample enrollment rate of village $i$, $d_i$ is program eligibility, and $x_i$ is the relative score.
The original cutoff is $c_0=0$; that is, $d_i=\mathbf{1}\{x_i\ge 0\}$.
The parameter is a function $f:\mathbb{R}\times\{0,1\}\rightarrow \mathbb{R}$, where $f(x,d)$ represents the counterfactual enrollment rate conditional on the relative score if the eligibility status were set to $d\in\{0,1\}$. 
Since $Y_i$ is a village-level sample mean, it is plausible to assume that $Y_i$ is approximately normally distributed.
I use the conventional standard error of the sample mean as the standard deviation of $Y_i$.\footnote{The sample enrollment rate is zero in 21 out of 287 villages.
	I exclude these villages from the analysis, since the standard error of $Y_i$ is zero.}

\subsection{Hypothetical Policy Choice Problem}\label{section:hypothetical}

Suppose we are evaluating the program to decide whether to scale it up.
Specifically, consider the following decision problem.
The counterfactual policy is to build BRIGHT schools in previously ineligible villages whose relative scores are in the top 20\%,
which corresponds to lowering the cutoff from $0$ to $-0.256$; in Section \ref{section:sensitivity}, I examine the sensitivity of the result to the choice of the new cutoff.
I use the average enrollment rate across villages as the welfare criterion, so that the welfare effect of this policy relative to the status quo is
$$
L(f)=\frac{1}{\tilde n}\sum_{i=1}^n\mathbf{1}\{-0.256\le x_i<0\}[f(x_i,1)-f(x_i,0)],
$$
where $\tilde n=\sum_{i=1}^n\mathbf{1}\{-0.256\le x_i<0\}$ is the number of villages that would receive a school under the new policy.
Some assumptions underlying this choice of $L(f)$ are: (i) we consider the set of villages in the sample rather than a new set of villages, and (ii) the counterfactual enrollment rate function $f$ remains constant over time between the period when the BRIGHT program was implemented and the period when the program is expanded.\footnote{Another underlying assumption is no spillover effects. The plausibility of this assumption can be indirectly verified, for example, based on how isolated each village is from other villages.}
In principle, these assumptions can potentially be relaxed by a suitable choice of $L(f)$ and the function class ${\cal F}$, while I focus on the above choice of $L(f)$ for a simple illustration of my approach.

When deciding whether to implement the policy, it is important to consider the benefit relative to the cost.
\cite{Kazianga2013bright} provide an estimate of the cost of constructing a BRIGHT school, which is \$4,758 per village.\footnote{I assume that the cost is known and constant across villages.
    If village-level cost data are available, my framework allows for unknown and heterogeneous costs by introducing the cost model on top of the outcome model.}
To incorporate the cost in the decision problem, suppose that the policymaker cares about the cost-effectiveness of this new policy relative to similar programs.
Cost-effectiveness is defined as the ratio of the policy cost to the increase in the enrollment.
I assume that it is optimal to implement the policy if its cost-effectiveness is smaller than that of a benchmark policy, denoted by $CE_0$---that is,
$$
\frac{\text{\$4,758}}{416\cdot\frac{1}{\tilde n}\sum_{i=1}^n\mathbf{1}\{-0.256\le x_i<0\}[f(x_i,1)-f(x_i,0)]}\le CE_0,
$$
where $416$ is the number of children per village.
The denominator represents the increase in the average enrollment across villages that would receive a BRIGHT school under the new policy.
For concreteness, I set the benchmark cost-effectiveness to $\$83.77$, which is the cost-effectiveness of a school construction program in Indonesia \citep{duflo2001school,Kazianga2013bright}.\footnote{The cost per village and the cost-effectiveness of a school construction program in Indonesia are found in Tables A18 and A20, respectively, in Online Appendix of \cite{Kazianga2013bright}.
	I compute the number of children per village by dividing the total enrollment by the enrollment rate reported in Table A17 in Online Appendix of \cite{Kazianga2013bright}.}
The above condition with $CE_0=\$83.77$ is equivalent to
$$
\frac{1}{\tilde n}\sum_{i=1}^n\mathbf{1}\{-0.256\le x_i<0\}[f(x_i,1)-0.137-f(x_i,0)]\ge 0.
$$
My method can be used to consider this decision problem by setting the outcome to $Y_i-0.137d_i$, where $0.137$ can be viewed as the policy cost measured in the unit of the enrollment rate.
Alternatively, $0.137$ can be viewed as the effect on the enrollment rate of implementing a benchmark policy with the same cost as the new policy.
I present the results for this scenario as well as for a benchmark scenario in which we ignore the policy cost.

I implement my method assuming that the counterfactual outcome function $f$ belongs to the Lipschitz class ${\cal F}_{\rm Lip}(C)$.\footnote{It is possible to incorporate the natural bound of $[0,1]$ on enrollment rates by setting the outcome to $Y_i-0.5$ and assuming $f(x,d)\in [-0.5,0.5]$ for all $(x,d)$ in addition to the Lipschitz constraint.
I find that this adjustment does not change the minimax regret rule and its maximum regret for the range of the Lipschitz constant $C$ considered in this analysis.}
The Lipschitz constant $C$ represents the maximum possible change in the enrollment rate in response to a one-unit change in the relative score.
While the relative score is computed based on multiple village-level characteristics, it is largely based on the estimated number of students to be served.
All other characteristics being equal, a one-unit increase in the relative score corresponds to around 100 additional children in the village, where the average number of children per village is 416.
To obtain a reasonable range of $C$, I estimate a lower bound on $C$ using the method described in Section \ref{section:implementation}, which yields the lower bound estimate of $0.149$.\footnote{I estimate $\frac{\partial f(x,0)}{\partial x}$ at $x\in\{-2.5,-2.45,...,-0.05\}$ and $\frac{\partial f(x,1)}{\partial x}$ at $x\in\{0.05,0.1,...,2.5\}$ by local quadratic regression and take the maximum of their absolute values.
	For local quadratic regression, I use the MSE-optimal bandwidth selection procedure of \cite{calonic2018bian}, which can be implemented by R package ``nprobust.''}
I present the results for $C\in\{0.05,0.1,...,0.95,1\}$ and examine their sensitivity to the choice of $C$.

\subsection{Results}

Figure \ref{fig:choice_mmr} plots $\delta^*(\boldsymbol Y)$, the probability of choosing the new policy computed by the minimax regret rule, against the Lipschitz constant $C$.
When $C< 0.6$, the minimax regret rule is nonrandomized.
It chooses the new policy in the no-cost scenario and maintains the status quo in the scenario in which the policy cost is $0.137$.
When $C\ge 0.6$, on the other hand, the minimax regret rule is randomized.
The decisions become more mixed as $C$ increases.
Given that the estimate of the lower bound on $C$ is 0.149, the minimax regret rule is nonrandomized when $C$ is less than four times the estimated lower bound.
Under this reasonable range of $C$, the optimal decision is the same in each scenario.


\begin{figure}[t]
	\centering
	\caption{Optimal Decisions: Probability of Choosing the New Policy} 
	\includegraphics[width=0.5\textwidth]{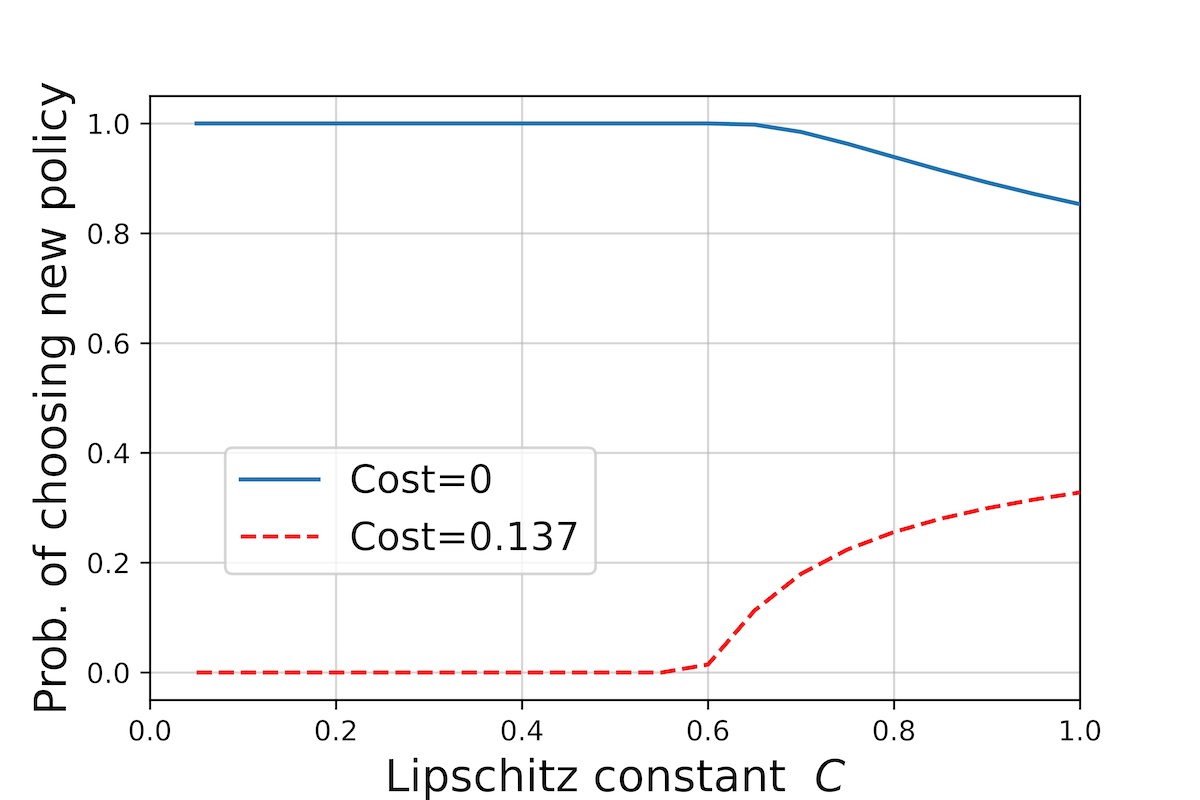}
	\vspace{1em}
	\caption*{\scriptsize{\it Notes}: This figure shows the probability of choosing the new policy computed by the minimax regret rule.
		The new policy is to construct BRIGHT schools in previously ineligible villages whose relative scores are in the top 20\%.
		The solid line shows results for the scenario in which we ignore the policy cost.
		The dashed line shows results for the scenario in which the policy cost measured in the unit of the enrollment rate is 0.137.
		I report results for the range $[0.05,0.1,...,0.95,1]$ of the Lipschitz constant $C$.}
	\label{fig:choice_mmr}
\end{figure}

\begin{figure}[t]
	\centering
	\caption{Weight to Each Village Attached by the Minimax Regret Rule} 
	\begin{subfigure}[b]{0.49\textwidth}   
		\centering 
		\includegraphics[width=\textwidth]{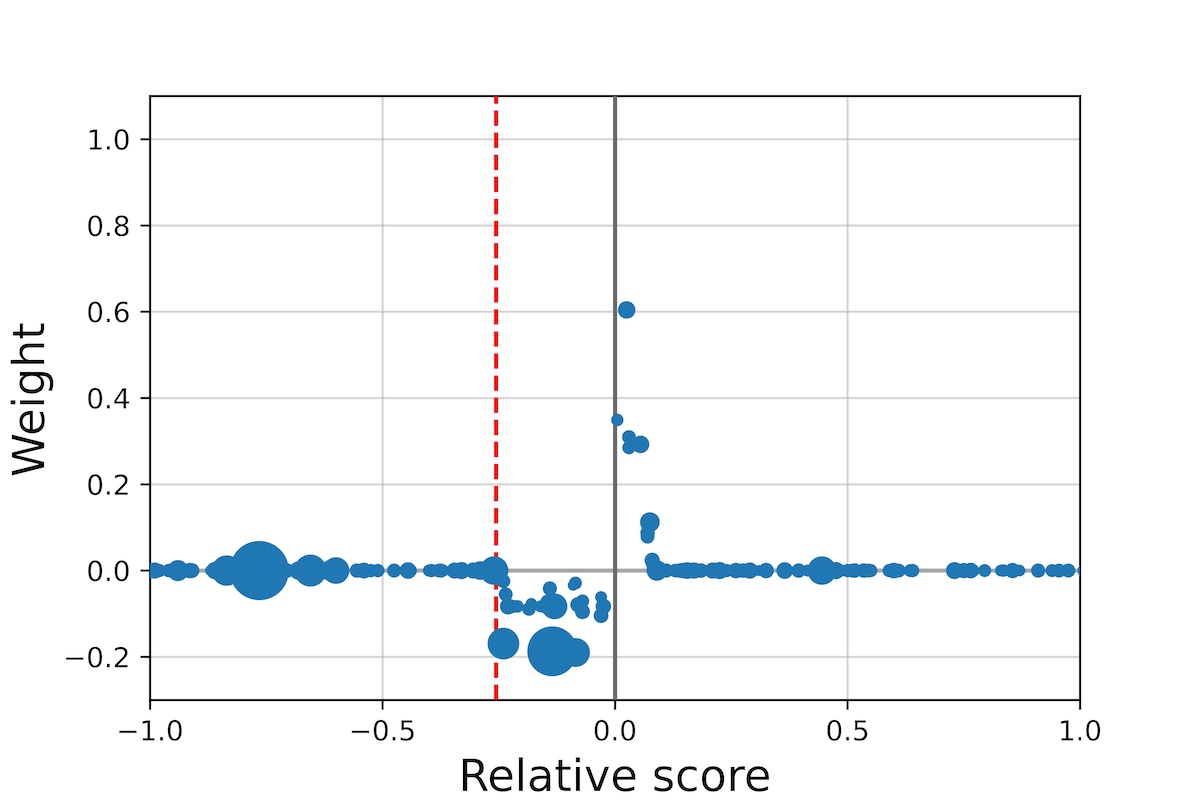}
		\caption[]%
		{$C=0.1$} 
		\label{subfig:mmr_C=01}
	\end{subfigure}
	\begin{subfigure}[b]{0.49\textwidth}   
		\centering 
		\includegraphics[width=\textwidth]{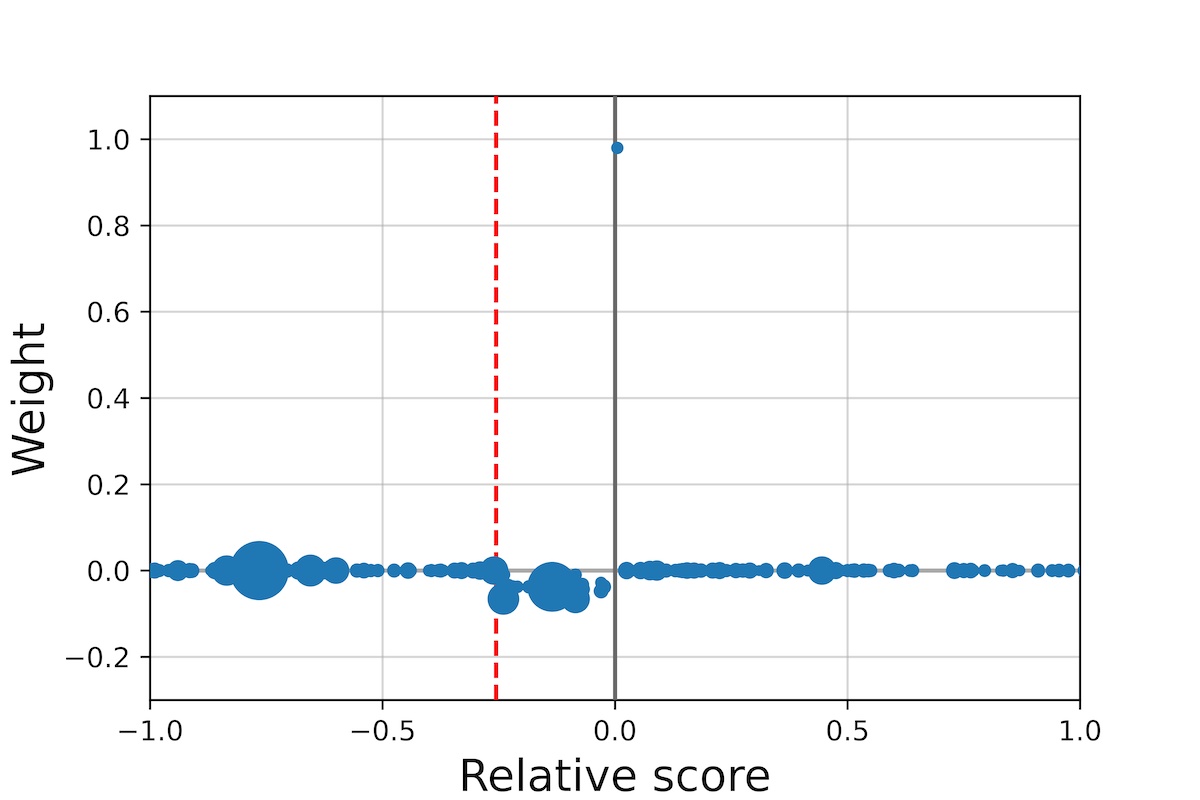}
		\caption[]%
		{$C=0.5$}
		\label{subfig:mmr_C=05}
	\end{subfigure}
	
	\caption*{\scriptsize{\it Notes}: This figure shows the weight $w_i$ attached to each village by the minimax regret rule of the form $\delta^*(\boldsymbol Y)=\mathbf{1}\{\sum_{i=1}^nw_iY_i\ge 0\}$.
		The weights are normalized so that $\sum_{i=1}^nw_i^2=1$.
		The horizontal axis indicates the relative score of each village.
		Each circle corresponds to each village.
		The size of circles is proportional to the inverse of the standard error of the enrollment rate $Y_i$.
		The vertical dashed line corresponds to the new cutoff $-0.256$.
		Panels (a) and (b) show the results when the Lipschitz constant $C$ is 0.1 and 0.5, respectively.}
	\label{fig:weight}
\end{figure}

If the minimax regret rule is nonrandomized, the rule is of the form $\delta^*(\boldsymbol Y)=\mathbf{1}\{\sum_{i=1}^nw_iY_i\ge 0\}$ for some weights $w_i$'s.
Panels (a) and (b) of Figure \ref{fig:weight} plot the weight $w_i$ attached to each village against the relative score $x_i$ for $C=0.1$ and $C=0.5$, respectively.
In the plots, the size of circles is proportional to the inverse of the standard error of the enrollment rate $Y_i$.
For both $C=0.1$ and $C=0.5$, a few treated units just above the original cutoff (the solid vertical line) receive a positive weight, the untreated units between the original cutoff and the new cutoff (the dashed vertical line) receive a negative weight, and no other units receive any weight.
When $C=0.1$, the weight tends to be larger for units with a smaller standard error.
When $C=0.5$, a positive weight is attached only to the treated unit closest to the original cutoff.
Also, the weights on the untreated units between the two cutoffs are almost identical.
This situation corresponds to the minimax regret rule of the form $\delta^*(\boldsymbol Y)=\mathbf{1}\left\{Y_{+,\min}-\frac{1}{\tilde n}\sum_{i:c_1\le x_i<c_0}Y_i\ge 0\right\}$ discussed in Section \ref{section:app-result}.

\subsubsection{Comparison with Alternative Rules}

I compare the minimax regret rule with several plug-in decision rules of the form $\delta(\boldsymbol{Y})=\mathbf{1}\{\hat L(\boldsymbol{Y})\ge 0\}$, where $\hat L(\boldsymbol{Y})$ is an estimator of the policy effect $L(f)$.
I consider the following three estimators of $L(f)$. (i) The minimax affine MSE estimator \citep{donoho1994}, described in Section \ref{section:donoho}, under the Lipschitz class ${\cal F}_{\rm Lip}(C)$. (ii) The minimax affine MSE estimator under the additional assumption of constant conditional treatment effects. In other words, I construct the estimator assuming that
	$
	{\cal F}=\{f\in{\cal F}_{\rm Lip}(C): f(x,1)-f(x,0)=f(\tilde x,1)-f(\tilde x,0) ~\text{for all}~x,\tilde x\}
	$.
	This estimation corresponds to first nonparametrically estimating the average treatment effect at the original cutoff and then extrapolating the effects on the units between the two cutoffs by the constant effects assumption.
(iii) The polynomial regression estimator \citep{Kazianga2013bright}.\footnote{\cite{Kazianga2013bright} estimate the treatment effect at the cutoff, not the effect on units away from the cutoff. They apply global polynomial regression RD estimators to child-level data. }
	Given the degree of polynomial $p$, I first estimate the model $f(x,d)=\alpha_0+\alpha_1 x+\cdots+\alpha_p x^p +\beta_0 d+\beta_1 d\cdot x +\cdots +\beta_pd\cdot x^p$ by weighted least squares regression using $1/\sigma^2(x_i,d_i)$ as the weight.
	I then estimate $L(f)$ by $\frac{1}{\tilde n}\sum_{i=1}^n\mathbf{1}\{-0.256\le x_i<0\}[\hat f(x_i,1)-\hat f(x_i,0)]$, where $\hat f$ is the estimated function.
	This estimator relies on the functional form of $f$ to extrapolate $f(x_i,1)$ for the untreated units.

\begin{figure}[!t]
	\centering
	\caption{Estimated Effects of the New Policy on the Enrollment Rate} 
	\begin{subfigure}[b]{0.49\textwidth}   
		\centering 
		\includegraphics[width=\textwidth]{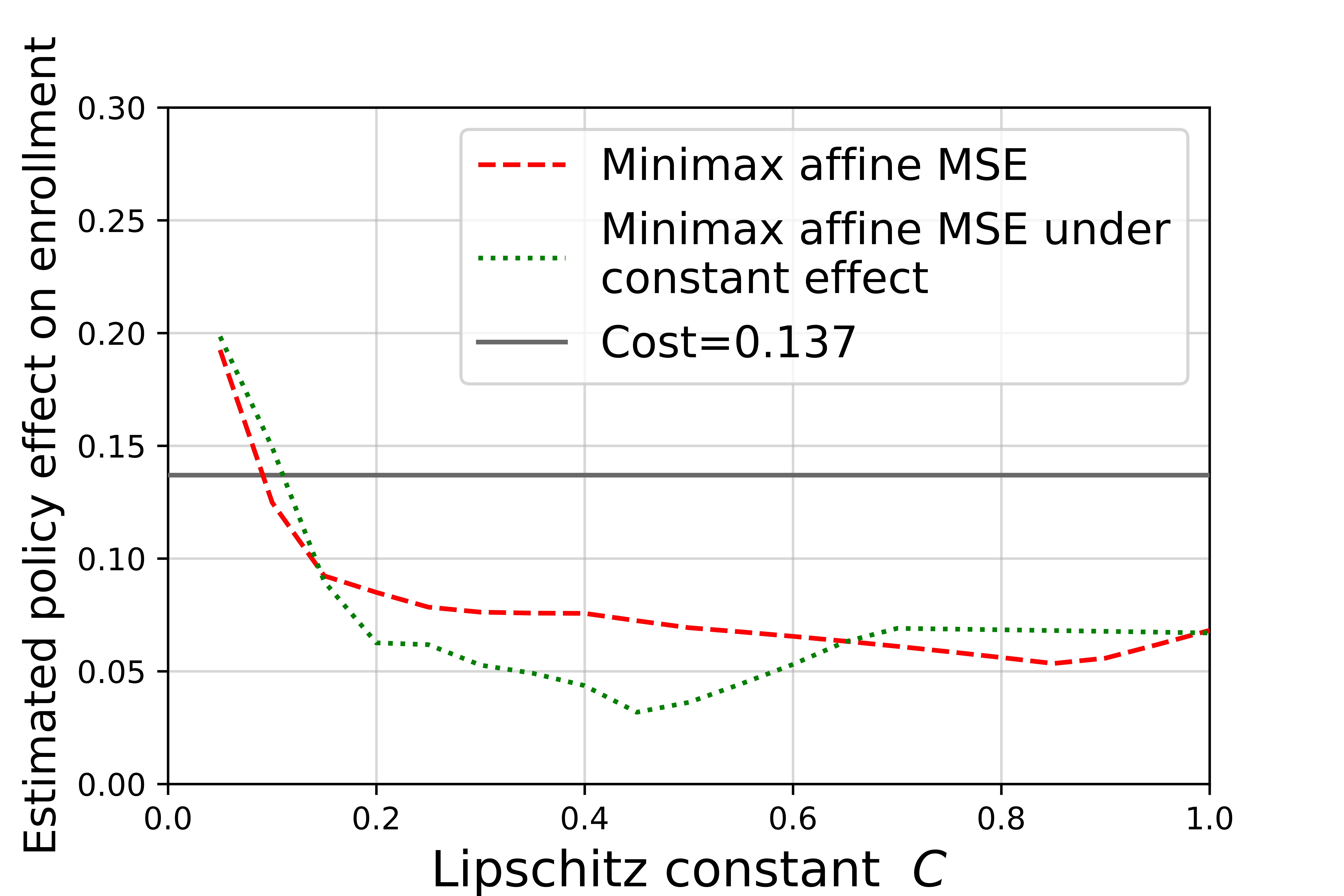}
		\vspace{-0.8em}
		\caption[]%
		{Minimax Affine MSE Estimators Under\\
			$~~~~~~~~~~~~~~~~~~~~~$ Lipschitz Class}
		\label{subfig:mse_est}
	\end{subfigure}
	\begin{subfigure}[b]{0.49\textwidth}   
		\centering 
		\includegraphics[width=\textwidth]{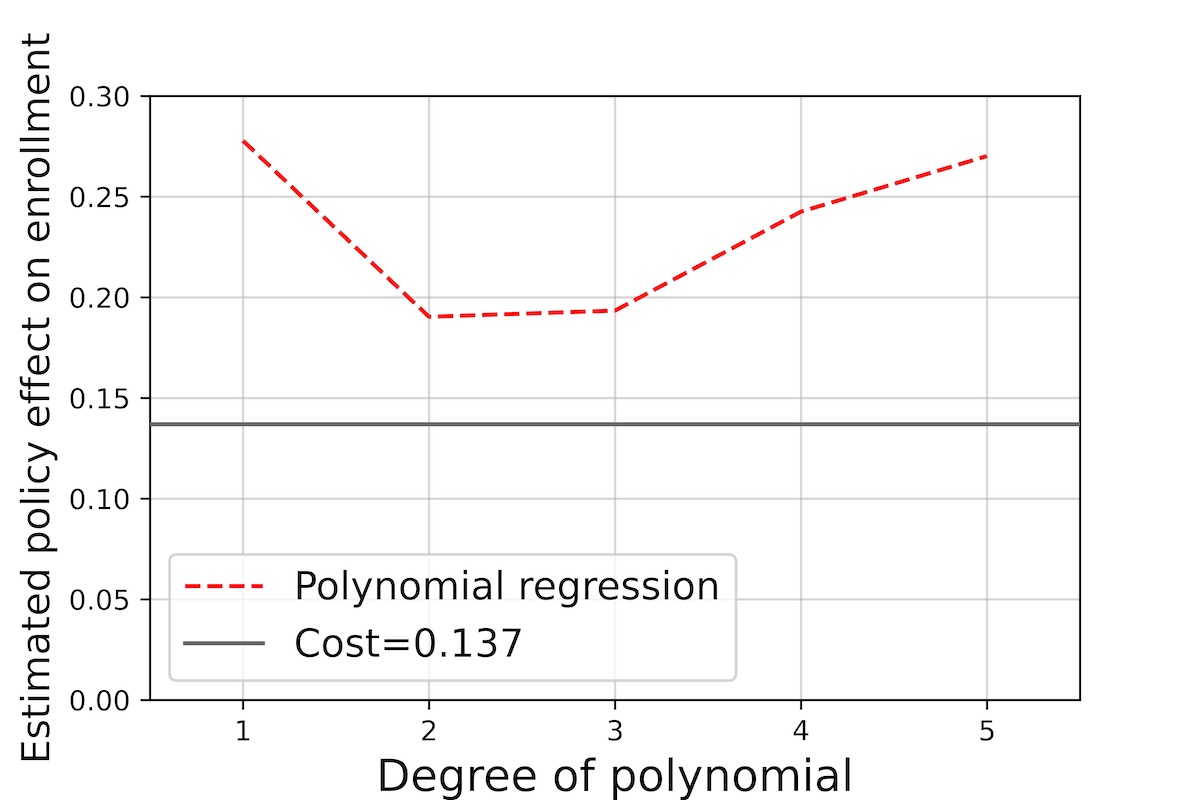}
		\vspace{-0.8em}
		\caption[]%
		{Polynomial Regression Estimators\\
			~} 
		\label{subfig:pol_est}
	\end{subfigure}
	\caption*{\scriptsize{\it Notes}: This figure shows the average effect of the new policy on the enrollment rate across the villages that would receive a school under the new policy.
		Panel (a) reports estimates from the minimax affine MSE estimators with and without the assumption of constant conditional treatment effects.
		I report results for the range $[0.05,0.1,...,0.95,1]$ of the Lipschitz constant $C$.
		Panel (b) reports estimates from the polynomial regression estimators of degrees 1 to 5.
		The horizontal line shows the cost of 0.137, which is my main specification of the policy cost.
	}
	\label{fig:est}
\end{figure}

Panel (a) of Figure \ref{fig:est} reports the estimated policy effects from the minimax affine MSE estimators with and without constant conditional treatment effects.
Overall, these two estimators exhibit a similar pattern.
While the estimated policy effects are larger than the policy cost when $C$ is close to zero, they are smaller than the policy cost when $C$ is moderate or large.
For $C\ge 0.15$, the resulting decisions about whether to choose the new policy are the same as the decision made by the minimax regret rule until $C$ reaches 0.6, where the minimax regret rule starts to randomize.
In contrast, the estimated policy effects from the polynomial regression estimators of degrees $1$ to $5$ exceed the policy cost, as reported in Panel (b) of Figure \ref{fig:est}.
The estimates appear to be close to the simple mean outcome difference between eligible and ineligible villages that can be computed from Table \ref{table:summary} in Appendix \ref{appendix:empirical_result}.
The resulting decisions differ from the decision made by the minimax regret rule.\footnote{The estimators presented here can be written as $\sum_{i=1}^nw_iY_i$ for some weights $w_i$'s. See Figure \ref{fig:weight_plug} in Appendix \ref{appendix:empirical_result} for the plots of these weights. While the minimax affine MSE estimators attach weights to units just above the original cutoff and to units between the two cutoffs, polynomial regression estimators even attach weights to units further from the cutoffs.}

The above decisions are computed from a particular realization of the sample.
To assess the ex ante performance of different decision rules, I compute the maximum regret of these rules when the true function class is ${\cal F}_{\rm Lip}(C)$.\footnote{I compute the maximum regret of the minimax regret rule using the formula in Theorem \ref{theorem:main}.
	For the other rules, I adapt the approach of \cite{ishihara2021meta} to numerically calculate the maximum regret in this setup.}
Panel (a) of Figure \ref{fig:max_reg} reports the result for the minimax regret rule and the plug-in rules based on the minimax affine MSE estimators with and without constant conditional treatment effects.\footnote{
The result for the plug-in rules based on polynomial regression estimators is omitted, since these rules turn out to have significantly larger maximum regret than the other rules.}
The maximum regret of the plug-in MSE rule with constant conditional treatment effects is much larger than that of the other two, especially when the Lipschitz constant $C$ is large.
The plug-in MSE rule without constant conditional treatment effects performs slightly worse than the minimax regret rule.
The ratio of the maximum regret between the two rules is maximized at $C=0.6$, where the minimax regret rule starts to randomize.
The maximized ratio is about 1.233.

\begin{figure}[!t]
	\centering
	\caption{Maximum Regret of the Minimax Regret Rule and Plug-in MSE Rules} 
	\begin{subfigure}[b]{0.49\textwidth}   
		\centering 
        \includegraphics[width=\textwidth]{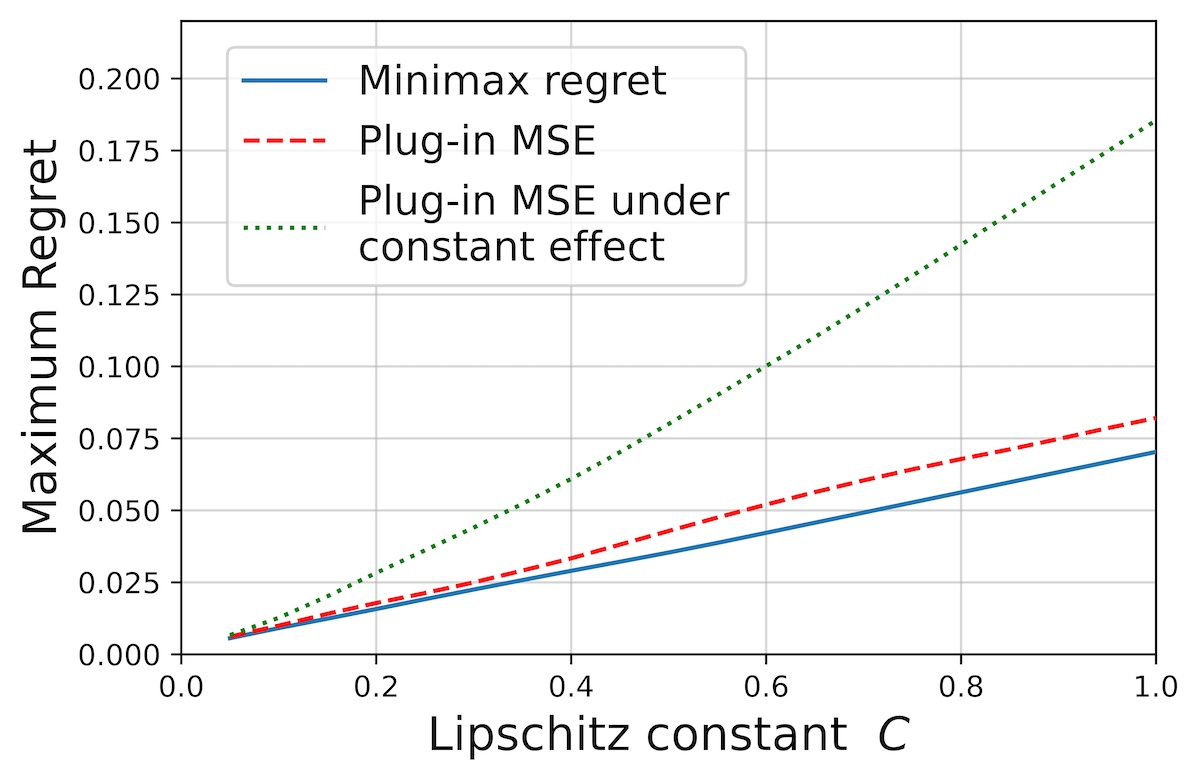}
		\vspace{-0.8em}
		\caption[]%
		{Correctly Specified Lipschitz Constant}
		\label{subfig:max_reg_correct}
	\end{subfigure}
	\begin{subfigure}[b]{0.49\textwidth}   
		\centering 
		\includegraphics[width=\textwidth]{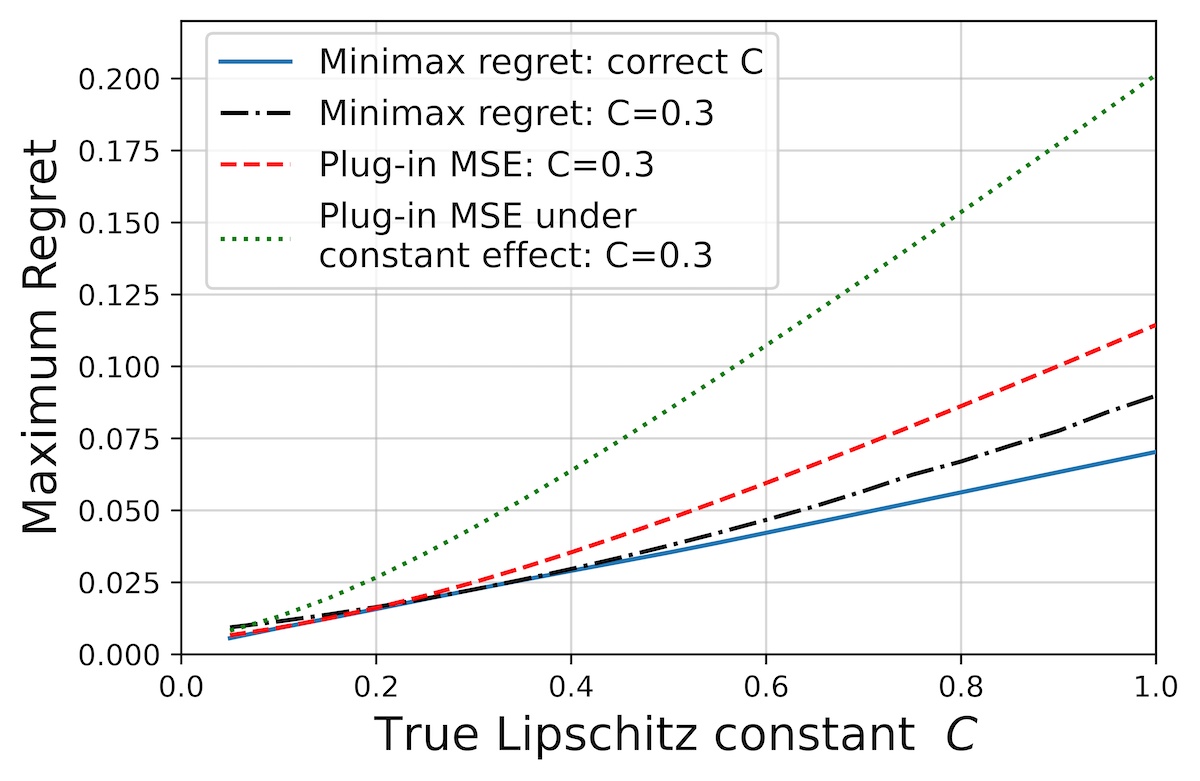}
		\vspace{-0.8em}
		\caption[]%
		{Misspecified Lipschitz Constant} 
		\label{subfig:max_reg_correct_misspecified}
	\end{subfigure}
	\caption*{\scriptsize{\it Notes}: This figure shows the maximum regret of the minimax regret rule and the plug-in rules based on the minimax affine MSE estimators with and without the assumption of constant conditional treatment effects.
    Panel (a) reports results for rules constructed using the true Lipschitz constant $C$.
    Panel (b) reports results for rules constructed assuming $C=0.3$, where the maximum regret is computed by setting the true $C$ to the value on the horizontal axis.
		The solid line in Panel (b) indicates the maximum regret that can be achieved if $C$ is correctly specified.
		In both panels, the maximum regret is normalized so that the unit is the same as that of the enrollment rate.
		I report results for the range $[0.05,0.1,...,0.95,1]$ of the Lipschitz constant $C$.}
	\label{fig:max_reg}
\end{figure}

\subsubsection{Sensitivity Analysis}\label{section:sensitivity}

I conduct several sensitivity analyses to assess how the results depend on the problem specification.
First, I examine the sensitivity of the decision from the minimax regret rule to the choice of the new eligibility cutoff.
Figure \ref{fig:other_cutoffs} in Appendix \ref{appendix:empirical_result} reports the results when the new policy builds schools in the top 10\% or 30\% of previously ineligible villages instead of the top 20\%.
As predicted by the result in Section \ref{section:app-result}, the minimax regret rule switches from a nonrandomized rule to a randomized rule at a smaller Lipschitz constant $C$ when the fraction of the target villages is larger.
When the fraction is 30\%, the rule is nonrandomized and suggests that the new policy is not cost-effective as long as $C$ is less than $0.4$.

Second, I examine the sensitivity of the decision to the policy cost.
I find that the minimax regret rule nonrandomly decides to maintain the status quo as long as the cost exceeds $0.10$, for the range of $C$ between its estimated lower bound of 0.149 and 0.6.

So far, I have constructed decision rules assuming that the Lipschitz constant $C$ is known, which is a crucial assumption in my theoretical analysis.
To assess the sensitivity of the performance to misspecification of $C$, I construct decision rules assuming $C=0.3$ and then compute their maximum regret when the true value of $C$ lies in $\{0.05,0.1,...,0.95,1\}$.
Panel (b) of Figure \ref{fig:max_reg} reports the result.
The solid line indicates the ``oracle'' maximum regret, which can be achieved if we correctly specify $C$.
The result shows that the plug-in MSE rule without constant conditional treatment effects performs slightly better than the minimax regret rule when the true $C$ is close to zero.
On the other hand, the minimax regret rule outperforms the plug-in MSE rule with nonnegligible differences for any value of the true $C$ greater than 0.3.
The result suggests that the minimax regret rule is more robust to misspecification of $C$ toward zero than the plug-in MSE rule.\footnote{The potential superiority of the minimax regret rule seems consistent with the theoretical results in the following way.
As shown in Section \ref{section:app-result}, when the true value of $C$ is large, the oracle minimax regret rule only uses the treated units just above the original cutoff and the untreated units between the original and new cutoffs (see Panel (b) of Figure \ref{fig:weight}).
If the specified $C$ is smaller than the true value, the resulting minimax regret rule is closer to the oracle rule than the plug-in MSE rule, since the minimax regret rule places more importance on the bias than the minimax affine MSE estimator, as discussed in Section \ref{section:donoho}.
Therefore, it is expected that the minimax regret rule performs better than the plug-in MSE rule under misspecification of $C$ toward zero.}


\section{Conclusion}\label{conclusion}

This paper derives an optimal decision rule for a large class of policy decision problems.
The framework introduced in this paper allows for infinite-dimensional parameters, various forms of parameter restrictions, and partial identification of social welfare.
I illustrate my approach through an application to the problem of eligibility cutoff choice in an RD setup.

Another promising application of this framework lies in policy adoption decisions using a difference-in-differences design.
Specifically, consider a group of units that have experienced a policy change and another group that has not. Suppose the policymaker needs to decide whether to implement the new policy for the latter group.
The average policy effect on this group may not be point identified if either the parallel trends assumption is violated or the policy effect varies between the two groups.
My framework can be applied to this problem by imposing a set of restrictions on the degree of the violation of parallel trends and the amount of heterogeneity in policy effects.

Future research may explore several theoretical directions.
First, one of the crucial assumptions in my approach is knowledge of the parameter space, such as the smoothness parameters of a function class.
It would be interesting to investigate the possibility of adaptation over a collection of parameter spaces---that is, achieving near-optimal worst-case regret simultaneously over multiple parameter spaces, as has been studied for estimation and inference problems.
Second, my approach only covers a binary choice problem. It would be challenging but both theoretically and practically important to extend the analysis to a multiple or continuous policy space.

\singlespacing
\bibliographystyle{myecca}
\bibliography{reference}
\onehalfspacing

\clearpage


\pagenumbering{arabic}
\renewcommand*{\thepage}{S-\arabic{page}}
\renewcommand\thefootnote{\alph{footnote}}
\setcounter{footnote}{0}



\begin{center}
    \begin{Large}
		Online Appendix to\\
  
  ``Optimal Decision Rules Under Partial Identification''\\  
	\end{Large}
    ~\\
\begin{large}
 Kohei Yata
 \end{large}
\end{center}

Appendix \ref{section:proof} contains the proof and a discussion of Theorem \ref{theorem:main}.
Appendix \ref{appendix:lemma-proof} contains auxiliary lemmas and proofs of Lemmas \ref{lemma:one-dim0}, \ref{lemma:hardest-1d}, and \ref{theorem:minimax_nonrandom} and Theorems \ref{theorem:minimax-bias} and \ref{theorem:donoho}.
Appendix \ref{appendix:app-details} contains derivations and a computational procedure for Section \ref{section:app}.
Appendix \ref{appendix:asymptotics} presents 
asymptotic properties of a feasible decision rule in the case with unknown error distribution.
Appendix \ref{appendix:empirical_result} contains additional results for the empirical application in Section \ref{section:empirical}.


\appendix

\section{Technical Details on Theorem \ref{theorem:main}}\label{section:proof}
\renewcommand{\theequation}{A.\arabic{equation}}
\setcounter{equation}{0}

I first provide two characterizations of $\boldsymbol{w}^*$ in Appendix \ref{subsection:characterization} and then prove Theorem \ref{theorem:main} in Appendix \ref{subsection:proof}.
In Appendix \ref{appendix:discussion}, I discuss potential existence of other minimax regret rules.

As a preliminary,
let $\bar{\cal M}\subset\mathbb{R}^n$ denote the linear span of ${\cal M}=\{\boldsymbol{m}(\theta):\theta\in\Theta\}$; that is,
$\bar{\cal M}\coloneqq \left\{\sum_{k=1}^K\lambda_k\boldsymbol{\mu}_k:K\in\mathbb{N},\boldsymbol{\mu}_1,...,\boldsymbol{\mu}_K\in {\cal M},\lambda_1,...,\lambda_K\in\mathbb{R}\right\}$.
$\bar{\cal M}$ is a linear subspace of $\mathbb{R}^n$.
If ${\cal M}$ contains $n$ linearly independent vectors, $\bar{\cal M}=\mathbb{R}^n$.
On the other hand,
if the dimension of $\theta$ is lower than $n$, or if there is a linear restriction between the elements of $\theta$, then $\bar{\cal M}$ is a lower-dimensional linear subspace of $\mathbb{R}^n$.\footnote{For example, suppose that $\Theta=\{\theta=(\theta_1,\theta_2)'\in\mathbb{R}^2:\theta_2=a\theta_1\}$ for some constant $a\neq 0$ and $\boldsymbol{m}(\theta)=(\theta_1,\theta_2)'$.
In this case, ${\cal M}=\Theta$, and $\bar{\cal M}=\Theta$, which is a one-dimensional subspace of $\mathbb{R}^2$.}
I equip $\bar{\cal M}$ with the Euclidean metric, the Euclidean norm, the dot product, and the topology induced by the Euclidean metric.

\subsection{Characterizations of $\boldsymbol{w}^*$}\label{subsection:characterization}

I present two analytical expressions for $\boldsymbol{w^*}$.
    These expressions clarify that the rule can be applied to a general class of problems in which the upper bound $\bar I(\cdot)$ is not differentiable at $\boldsymbol{0}$, not only to a specific problem in Example \ref{example:intersection}.
To provide the first one, let $\left.\bar I\right|_{\bar{\cal M}}:\bar{\cal M}\rightarrow[-\infty,\infty]$ be the restriction of $\bar I:\mathbb{R}^n\rightarrow[-\infty,\infty]$ to $\bar{\cal M}$.
Define the {\it superdifferential} of $\left.\bar I\right|_{\bar{\cal M}}$ at $\boldsymbol{0}$ as
\begin{align*}
\left.\partial\bar I\right|_{\bar{\cal M}}(\boldsymbol{0})\coloneqq&~\{\boldsymbol{d}\in\bar{\cal M}:\left.\bar I\right|_{\bar{\cal M}}(\boldsymbol{\mu})\le \left.\bar I\right|_{\bar{\cal M}}(\boldsymbol{0})+\boldsymbol{d}'\boldsymbol{\mu} \text{ for all }\boldsymbol{\mu}\in \bar{\cal M}\}\\
=&~\{\boldsymbol{d}\in\bar{\cal M}:\bar I(\boldsymbol{\mu})\le \bar I(\boldsymbol{0})+\boldsymbol{d}'\boldsymbol{\mu} \text{ for all }\boldsymbol{\mu}\in \bar{\cal M}\}.
\end{align*}
Each element of the superdifferential is called a {\it supergradient}.
Under Assumption \ref{assumption:problem}, $\left.\bar I\right|_{\bar{\cal M}}$ is concave on $\bar{\cal M}$, and $\left.\partial\bar I\right|_{\bar{\cal M}}(\boldsymbol{0})$ is a nonempty, convex, closed, and bounded subset of $\bar{\cal M}$; see Lemma \ref{lemma:superdifferential} in Appendix \ref{appendix:lemma}.
If $\left.\bar I\right|_{\bar{\cal M}}$ is differentiable at $\boldsymbol{0}$, $\left.\partial\bar I\right|_{\bar{\cal M}}(\boldsymbol{0})$ is the singleton set containing the gradient of $\left.\bar I\right|_{\bar{\cal M}}$ at $\boldsymbol{0}$, denoted by $\nabla \left.\bar I\right|_{\bar{\cal M}}(\boldsymbol{0})$.
If $\left.\bar I\right|_{\bar{\cal M}}$ is not differentiable at $\boldsymbol{0}$, $\left.\partial\bar I\right|_{\bar{\cal M}}(\boldsymbol{0})$ contains multiple elements.
In Example \ref{example:intersection}, $\left.\partial\bar I\right|_{\bar{\cal M}}(\boldsymbol{0})=\{(0,1)'\}$ if $C_1>C_2$ and $\left.\partial\bar I\right|_{\bar{\cal M}}(\boldsymbol{0})=\{(\lambda,1-\lambda)':\lambda\in [0,1]\}$ if $C_1=C_2$.

As shown in Lemma \ref{lemma:superdifferential}, it turns out that if $\boldsymbol{0}\notin \left.\partial\bar I\right|_{\bar{\cal M}}(\boldsymbol{0})$, $\boldsymbol{w^*}$ is the unit vector in the same direction as the supergradient with the minimum Euclidean norm:
\begin{align}
\boldsymbol{w}^*= \frac{\boldsymbol d^*}{\|\boldsymbol d^*\|}, ~\text{ where } \boldsymbol d^*\in\arg\min_{\boldsymbol d\in \left.\partial\bar I\right|_{\bar{\cal M}}(\boldsymbol{0})}\|\boldsymbol d\|. \label{eq:wstar2}
\end{align}
Furthermore, the right derivative of $\omega(\cdot)$ at $0$ is given by
\begin{align}
    \omega'(0)=\|\boldsymbol d^*\|.\label{eq:omega-dstar}
\end{align}
Note that if $\left.\bar I\right|_{\bar{\cal M}}$ is differentiable at $\boldsymbol{0}$, $\boldsymbol d^*$ is simply the gradient $\nabla \left.\bar I\right|_{\bar{\cal M}}(\boldsymbol{0})$, so that $\boldsymbol{w}^*$ is the normalized gradient $\frac{\nabla \left.\bar I\right|_{\bar{\cal M}}(\boldsymbol{0})}{\|\nabla \left.\bar I\right|_{\bar{\cal M}}(\boldsymbol{0})\|}$ and $\omega'(0)=\|\nabla \left.\bar I\right|_{\bar{\cal M}}(\boldsymbol{0})\|$.

An alternative, related characterization of
$\boldsymbol{w}^*$ is the direction of {\it steepest} ascent.
Formally, as shown in Lemma \ref{lemma:superdifferential}, $\boldsymbol{w}^*$ is the direction in which the directional derivative of $\bar I(\cdot)$ at $\boldsymbol{0}$ is maximized among unit vectors:
\begin{align}
\boldsymbol{w}^*\in\arg\max_{\boldsymbol{w}\in\bar{\cal M}:\|\boldsymbol{w}\|\le 1} \bar I'(\boldsymbol{0};\boldsymbol{w}), \label{eq:wstar3}
\end{align}
where $\bar I'(\boldsymbol{0};\boldsymbol{w})\coloneqq \lim_{\epsilon\downarrow 0}\frac{\bar I(\epsilon\boldsymbol{w})-\bar I(\boldsymbol{0})}{\epsilon}$ is the directional derivative of $\bar I(\cdot)$ at $\boldsymbol{0}$ in direction $\boldsymbol{w}\in\bar{\cal M}$.


\subsection{Proof of Theorem \ref{theorem:main}}\label{subsection:proof}

As explained in Section \ref{section:proof-strategy}, it suffices to show that \eqref{eq:property} holds. Below, I consider the case in which $\epsilon^*>0$ and the case in which $\epsilon^*=0$ separately.

\paragraph{Case (i): $\epsilon^*>0$ \textnormal{(or, equivalently, $\sigma\omega'(0)>2\phi(0)\omega(0)$)}.}
Let $\boldsymbol{\mu}_{\epsilon^*}=\boldsymbol{m}(\theta_{\epsilon^*})$.
Since $(\boldsymbol{\mu}_{\epsilon^*})'\boldsymbol{Y}\sim {\cal N}((\boldsymbol{\mu}_{\epsilon^*})'\boldsymbol{m}(\theta), \sigma^2\|\boldsymbol{\mu}_{\epsilon^*}\|^2)$,
the maximum regret of $\delta^*$ over $\Theta$ is given by
\begin{align*}
	&\sup_{\theta\in\Theta}R(\delta^*,\theta)
	=\sup_{\theta\in\Theta}\left[ (L(\theta))^+\Phi\left(-\frac{(\boldsymbol{\mu}_{\epsilon^*})'\boldsymbol{m}(\theta)}{\sigma\|\boldsymbol{\mu}_{\epsilon^*}\|}\right)+(-L(\theta))^+\Phi\left(\frac{(\boldsymbol{\mu}_{\epsilon^*})'\boldsymbol{m}(\theta)}{\sigma\|\boldsymbol{\mu}_{\epsilon^*}\|}\right)\right]\\
	&=\sup_{\theta\in\Theta}L(\theta)\Phi\left(-\frac{(\boldsymbol{\mu}_{\epsilon^*})'\boldsymbol{m}(\theta)}{\sigma\|\boldsymbol{\mu}_{\epsilon^*}\|}\right)\\
    &=\sup_{\boldsymbol{\mu}\in {\cal M}}\sup_{\theta\in\Theta:\boldsymbol{m}(\theta)=\boldsymbol{\mu}}L(\theta)\Phi\left(-\frac{(\boldsymbol{\mu}_{\epsilon^*})'\boldsymbol{m}(\theta)}{\sigma\|\boldsymbol{\mu}_{\epsilon^*}\|}\right)=\sup_{\boldsymbol{\mu}\in {\cal M}}\bar I(\boldsymbol{\mu})\Phi\left(-\frac{(\boldsymbol{\mu}_{\epsilon^*})'\boldsymbol{\mu}}{\sigma\|\boldsymbol{\mu}_{\epsilon^*}\|}\right),
\end{align*}
where $x^+=\max\{x,0\}$, the second equality holds by the symmetry of the objective function ($R(\delta^*,\theta)=R(\delta^*,-\theta)$) and the centrosymmetry of $\Theta$, and the last by the definition of $\bar I(\boldsymbol{\mu})$.

Now, consider the regret of $\delta^*$ under $-\theta_{\epsilon^*}$ and $\theta_{\epsilon^*}$.
Note that since $\boldsymbol{m}(\theta_{\epsilon^*})=\boldsymbol{\mu}_{\epsilon^*}$,
$$
L(\theta_{\epsilon^*}) \le \sup_{\theta\in\Theta:\boldsymbol{m}(\theta)=\boldsymbol{\mu}_{\epsilon^*}}L(\theta)=\bar I(\boldsymbol{\mu}_{\epsilon^*}).
$$
On the other hand, since $\|\boldsymbol{\mu}_{\epsilon^*}\|\le\epsilon^*$,
$$
\sup_{\theta\in\Theta:\boldsymbol{m}(\theta)=\boldsymbol{\mu}_{\epsilon^*}}L(\theta)\le \sup_{\theta\in\Theta:\|\boldsymbol{m}(\theta)\|\le\epsilon^*}L(\theta)=L(\theta_{\epsilon^*}).
$$
Hence, $L(\theta_{\epsilon^*})=\bar I(\boldsymbol{\mu}_{\epsilon^*})$.
As a result, the regret of $\delta^*$ under $-\theta_{\epsilon^*}$ and $\theta_{\epsilon^*}$ is given by
$$
R(\delta^*,-\theta_{\epsilon^*})=R(\delta^*,\theta_{\epsilon^*})=L(\theta_{\epsilon^*})\Phi\left(-\frac{(\boldsymbol{\mu}_{\epsilon^*})'\boldsymbol{m}(\theta_{\epsilon^*})}{\sigma\|\boldsymbol{\mu}_{\epsilon^*}\|}\right)=\bar I(\boldsymbol{\mu}_{\epsilon^*})\left(-\frac{(\boldsymbol{\mu}_{\epsilon^*})'\boldsymbol{\mu}_{\epsilon^*}}{\sigma\|\boldsymbol{\mu}_{\epsilon^*}\|}\right).
$$
Therefore, to prove \eqref{eq:property}, it suffices to show that
\begin{align}
\bar I(\boldsymbol{\mu}_{\epsilon^*})\left(-\frac{(\boldsymbol{\mu}_{\epsilon^*})'\boldsymbol{\mu}_{\epsilon^*}}{\sigma\|\boldsymbol{\mu}_{\epsilon^*}\|}\right)=\sup_{\boldsymbol{\mu}\in {\cal M}}\bar I(\boldsymbol{\mu})\Phi\left(-\frac{(\boldsymbol{\mu}_{\epsilon^*})'\boldsymbol{\mu}}{\sigma\|\boldsymbol{\mu}_{\epsilon^*}\|}\right).\label{eq:worst-mu}
\end{align}


\begin{lemma}\label{theorem:minimax_nonrandom}
        Under the assumptions in Theorem \ref{theorem:main}, if $\epsilon^*>0$, then \eqref{eq:worst-mu} holds.
\end{lemma}

\begin{proof}
	See Appendix \ref{proof:theorem:minimax_nonrandom}.
\end{proof}

\paragraph{Case (ii): $\epsilon^*=0$ \textnormal{(or, equivalently, $\sigma\omega'(0)\le2\phi(0)\omega(0)$)}.}
First, consider the case in which $\omega'(0)=0$ and hence $\delta^*(\boldsymbol{Y})=1/2$. In this case, the regret of $\delta^*$ is
$R(\delta^*,\theta)=|L(\theta)|/2$.
On the other hand, the fact that $\omega'(0)=0$ implies that $\omega(\cdot)$ is constant, since $\omega(\cdot)$ is nondecreasing and concave.
Therefore, $|L(\theta)|$ is maximized at $\theta\in\{-\theta_{\epsilon^*},\theta_{\epsilon^*}\}$ over $\Theta$, so \eqref{eq:property} holds.

Next, consider the case in which $\omega'(0)>0$, and let $s^*=2\phi(0)\frac{\omega(0)}{\omega'(0)}\ge \sigma$.
Write $\delta^*$ as $\delta^*(\boldsymbol{Y})=\mathbb{P}\left((\boldsymbol{w}^*)'\boldsymbol{Y}+\xi\ge 0|\boldsymbol{Y}\right)$, where $\xi|\boldsymbol{Y}\sim {\cal N}(0,(s^*)^2-\sigma^2)$.
Since $(\boldsymbol{w}^*)'\boldsymbol{Y}+\xi\sim {\cal N}\left((\boldsymbol{w}^*)'\boldsymbol{m}(\theta),(s^*)^2\right)$, it follows from an argument similar to the one used for case (i) that the maximum regret of $\delta^*$ over $\Theta$ is given by
\begin{align*}
	&\sup_{\theta\in\Theta}R(\delta^*,\theta)=\sup_{\theta\in\Theta}\left[ (L(\theta))^+\Phi\left(-\frac{(\boldsymbol{w}^*)'\boldsymbol{m}(\theta)}{s^*}\right)+(-L(\theta))^+\Phi\left(\frac{(\boldsymbol{w}^*)'\boldsymbol{m}(\theta)}{s^*}\right)\right]\\
	&=\sup_{\theta\in\Theta}L(\theta)\Phi\left(-\frac{(\boldsymbol{w}^*)'\boldsymbol{m}(\theta)}{s^*}\right)
 =\sup_{\boldsymbol{\mu}\in{\cal M}}\bar I(\boldsymbol{\mu})\Phi\left(-\frac{(\boldsymbol{w}^*)'\boldsymbol{\mu}}{s^*}\right).
\end{align*}
Let ${\cal E}_{\boldsymbol{w}^*}=\{(\boldsymbol{w}^*)'\boldsymbol{\mu}:\boldsymbol{\mu}\in{\cal M}\}$.
By Lemma \ref{lemma:superdifferential}, there exists a unique solution $\boldsymbol d^*$ to $\min_{\boldsymbol d\in \left.\partial\bar I\right|_{\bar{\cal M}}(\boldsymbol{0})}\|\boldsymbol d\|$, $\boldsymbol{w}^*= \frac{\boldsymbol d^*}{\|\boldsymbol d^*\|}$, and $\omega'(0)=\|\boldsymbol{d}^*\|$.
The last expression of the above equation can be bounded as follows:
\begin{align*}
    &\sup_{\boldsymbol{\mu}\in{\cal M}}\bar I(\boldsymbol{\mu})\Phi\left(-\frac{(\boldsymbol{w}^*)'\boldsymbol{\mu}}{s^*}\right)\le \sup_{\boldsymbol{\mu}\in{\cal M}}(\bar I(\boldsymbol{0})+(\boldsymbol{d}^*)'\boldsymbol{\mu})\Phi\left(-\frac{(\boldsymbol{w}^*)'\boldsymbol{\mu}}{s^*}\right)\\
    &=\sup_{\boldsymbol{\mu}\in{\cal M}}(\omega(0)+\omega'(0)(\boldsymbol{w}^*)'\boldsymbol{\mu})\Phi\left(-\frac{(\boldsymbol{w}^*)'\boldsymbol{\mu}}{s^*}\right)=\sup_{\epsilon\in{\cal E}_{\boldsymbol{w}^*}}(\omega(0)+\omega'(0)\epsilon)\Phi\left(-\epsilon/s^*\right),
\end{align*}
where the first inequality holds since $\boldsymbol{d}^*\in \left.\partial\bar I\right|_{\bar{\cal M}}(\boldsymbol{0})$ and the first equality holds since $\omega(0)=\bar I(\boldsymbol{0})$ by definition.

Now, consider the regret of $\delta^*$ under $-\theta_{\epsilon^*}$ and $\theta_{\epsilon^*}$.
Since $\theta_{\epsilon^*}$ attains the modulus of continuity at $\epsilon^*=0$, $L(\theta_{\epsilon^*})=\omega(0)$ and $\boldsymbol{m}(\theta_{\epsilon^*})=\boldsymbol{0}$. Hence,
$$
R(\delta^*,-\theta_{\epsilon^*})=R(\delta^*,\theta_{\epsilon^*})=L(\theta_{\epsilon^*})\Phi\left(-\frac{(\boldsymbol{w}^*)'\boldsymbol{m}(\theta_{\epsilon^*})}{s^*}\right)=\omega(0)/2.
$$
Therefore, to prove \eqref{eq:property}, it suffices to show that
\begin{align*}
\omega(0)/2=\sup_{\epsilon\in{\cal E}_{\boldsymbol{w}^*}}(\omega(0)+\omega'(0)\epsilon)\Phi\left(-\epsilon/s^*\right).
\end{align*}
By the choice of $s^*$, it follows from Lemma \ref{lemma:unimodal2} in Appendix \ref{appendix:lemma} that the function $\epsilon\mapsto (\omega(0)+\omega'(0)\epsilon)\Phi\left(-\epsilon/s^*\right)$ is maximized at $0$ over $\mathbb{R}$.
Since $0\in {\cal E}_{\boldsymbol{w}^*}$, we have $\sup_{\epsilon\in{\cal E}_{\boldsymbol{w}^*}}(\omega(0)+\omega'(0)\epsilon)\Phi\left(-\epsilon/s^*\right)=\omega(0)/2$.

\subsection{Discussion}\label{appendix:discussion}

    \paragraph{Possibility of Exisitence of Nonrandomized Minimax Regret Rules.}
    Theorem \ref{theorem:main} does not exclude the possibility that the nonrandomized rule $\delta^*_{\rm NR}(\boldsymbol{Y})=\mathbf{1}\left\{(\boldsymbol{w}^*)'\boldsymbol{Y}\ge 0\right\}$ is minimax regret for the case in which $2\phi(0)\frac{\omega(0)}{\omega'(0)}> \sigma$.
    Since $\sup_{\theta\in\Theta}R(\delta^*_{\rm NR},\theta)=\sup_{\boldsymbol{\mu}\in{\cal M}}\bar I(\boldsymbol{\mu})\Phi\left(-(\boldsymbol{w}^*)'\boldsymbol{\mu}/\sigma\right)$, $\delta^*_{\rm NR}$ is minimax regret if and only if
    $$
    \omega(0)/2=\sup_{\boldsymbol{\mu}\in{\cal M}}\bar I(\boldsymbol{\mu})\Phi\left(-(\boldsymbol{w}^*)'\boldsymbol{\mu}/\sigma\right).
    $$
    This condition may or may not hold depending on the functional form of $\bar I(\boldsymbol{\mu})$.
    A simple example where this condition holds is as follows: $\boldsymbol{\mu}$ is scalar, $\boldsymbol{w}^*=1$, and $\bar I(\boldsymbol{\mu})=a_0\boldsymbol{\mu}\cdot \mathbf{1}\{\boldsymbol{\mu}<0\}+a_1\boldsymbol{\mu}\cdot \mathbf{1}\{\boldsymbol{\mu}\ge 0\}+b$ for some constants $a_0,a_1,b>0$ such that $2\phi(0)\frac{b}{a_0}\le \sigma$ and $2\phi(0)\frac{b}{a_1}\ge \sigma$.
    
    However, if $\bar I(\boldsymbol{\mu})$ is locally smooth in the direction $\boldsymbol{w}^*$ around $\boldsymbol{0}$, then $\delta^*_{\rm NR}$ is not minimax regret.
    Specifically, suppose that $\bar I'(\boldsymbol{0};-\boldsymbol{w}^*)=-\bar I'(\boldsymbol{0};\boldsymbol{w}^*)$.
    Let $f(\epsilon)=\bar I(\epsilon\boldsymbol{w}^*)\Phi\left(-(\boldsymbol{w}^*)'(\epsilon\boldsymbol{w}^*)/\sigma\right)=\bar I(\epsilon\boldsymbol{w}^*)\Phi(-\epsilon/\sigma)$ for $\epsilon$ in a neighborhood of zero.
   If $\bar I'(\boldsymbol{0};-\boldsymbol{w}^*)=-\bar I'(\boldsymbol{0};\boldsymbol{w}^*)$, $f$ is differentiable at $0$ with
    \begin{align*}
        f'(0)=\bar I'(\boldsymbol{0};\boldsymbol{w}^*)\Phi(0)-\bar I(\boldsymbol{0})\phi(0)/\sigma=\omega'(0)/2-\omega(0)\phi(0)/\sigma<0,
    \end{align*}
    where the second equality holds since $\bar I'(\boldsymbol{0};\boldsymbol{w}^*)=\omega'(0)$ by Lemma \ref{lemma:superdifferential} and the inequality holds since $2\phi(0)\frac{\omega(0)}{\omega'(0)}> \sigma$.
    As a result, there exists $\epsilon<0$ such that
    $f(\epsilon)>f(0)=\omega(0)/2$, and hence $\bar I(\epsilon\boldsymbol{w}^*)\Phi\left(-(\boldsymbol{w}^*)'(\epsilon\boldsymbol{w}^*)/\sigma\right)>\omega(0)/2$.
    Therefore, $\delta^*_{\rm NR}$ is not minimax regret.
    Note that the existence of $\epsilon<0$ such that $\bar I(\epsilon\boldsymbol{w}^*)\Phi\left(-(\boldsymbol{w}^*)'(\epsilon\boldsymbol{w}^*)/\sigma\right)>\omega(0)/2$ implies the existence of $\theta^*\in\Theta$ that satisfies 
    \eqref{eq:randomization-error} and \eqref{eq:randomization} in Section \ref{section:randomization}.

    The condition $\bar I'(\boldsymbol{0};-\boldsymbol{w}^*)=-\bar I'(\boldsymbol{0};\boldsymbol{w}^*)$ holds if $\left.\bar I\right|_{\bar{\cal M}}$ is differentiable at $\boldsymbol{0}$.
    This condition may hold, however, even if $\left.\bar I\right|_{\bar{\cal M}}$ is not differentiable at $\boldsymbol{0}$. For example, it holds in Example \ref{example:intersection} even for the case in which $C_1=C_2$.
    
    \paragraph{Possibility of Existence of Minimax Regret Rules Based on Other Weighted Sums.}
    It is the key that $\delta^*$ is based on a unit vector $\boldsymbol{w}^*$ that is in the direction of $\boldsymbol{d}^*$, the supergradient of $\left.\bar I\right|_{\bar{\cal M}}$ at $\boldsymbol{0}$ with the least norm.
    If $\left.\bar I\right|_{\bar{\cal M}}$ is not differentiable at $\boldsymbol{0}$, one could consider using another $\boldsymbol{d}\in\left.\partial\bar I\right|_{\bar{\cal M}}(\boldsymbol{0})$ instead of $\boldsymbol{d}^*$ to construct a rule $\delta_{\boldsymbol{w},s}(\boldsymbol{Y})=\mathbb{P}(\boldsymbol{w}'\boldsymbol{Y}+\xi\ge 0|\boldsymbol{Y})$, where $\boldsymbol{w}=\frac{\boldsymbol{d}}{\|\boldsymbol{d}\|}$ and $\xi|\boldsymbol{Y}\sim {\cal N}(0,s^2-\sigma^2)$ for some $s\ge \sigma$.
    Such a rule may or may not be minimax regret if $2\phi(0)\frac{\omega(0)}{\omega'(0)}\ge \sigma$. To see this, note that by an argument similar to the one used in Appendix \ref{subsection:proof}, $\delta_{\boldsymbol{w},s}$ is minimax regret if
    \begin{align*}
    \omega(0)/2=\sup_{\epsilon\in{\cal E}_{\boldsymbol{w}}}(\omega(0)+\|\boldsymbol{d}\|\epsilon)\Phi\left(-\epsilon/s\right),
    \end{align*}
    where ${\cal E}_{\boldsymbol{w}}=\{(\boldsymbol{w})'\boldsymbol{\mu}:\boldsymbol{\mu}\in{\cal M}\}$.
    This condition holds if and only if $s=2\phi(0)\frac{\omega(0)}{\|\boldsymbol{d}\|}$ by Lemma \ref{lemma:unimodal2} in Appendix \ref{appendix:lemma}.
    For problems with $2\phi(0)\frac{\omega(0)}{\|\boldsymbol{d}\|}\ge \sigma$, one can construct $\delta_{\boldsymbol{w},s}$ with such an $s$, and this rule is minimax regret.
    However, note that $2\phi(0)\frac{\omega(0)}{\|\boldsymbol{d}\|}<2\phi(0)\frac{\omega(0)}{\|\boldsymbol{d}^*\|}=2\phi(0)\frac{\omega(0)}{\omega'(0)}$, since $\boldsymbol{d}^*$ is the unique solution to $\min_{\boldsymbol d\in \left.\partial\bar I\right|_{\bar{\cal M}}(\boldsymbol{0})}\|\boldsymbol d\|$ and $\omega'(0)=\|\boldsymbol{d}^*\|$.
    Therefore, if $2\phi(0)\frac{\omega(0)}{\omega'(0)}\ge \sigma$, it is possible that $2\phi(0)\frac{\omega(0)}{\|\boldsymbol{d}\|}<\sigma$.
    For such problems, it is not possible to use $s=2\phi(0)\frac{\omega(0)}{\|\boldsymbol{d}\|}$ to construct $\delta_{\boldsymbol{w},s}$, and alternative rules with $s\ge \sigma$ are not necessarily minimax regret.

    This argument provides a justification for using $\boldsymbol{w}^*$: The resulting rule $\delta^*$ is minimax regret for any problems whenever $\epsilon^*=0$.
    At the same time, it highlights the possibility that there exist minimax regret rules based on $\boldsymbol{w}'\boldsymbol{Y}$ with $\boldsymbol{w}\neq \boldsymbol{w}^*$ for some problems in which the upper bound $\left.\bar I\right|_{\bar{\cal M}}$ is not differentiable at $\boldsymbol{0}$.\footnote{This contrasts with the point made by \citeappendix{olea2023partial} that, for a given minimax regret rule based on $\boldsymbol{w}'\boldsymbol{Y}$ (such as $\delta^*$), it is possible to construct other minimax regret rules based on the same $\boldsymbol{w}'\boldsymbol{Y}$.}

\section{Auxiliary Lemmas and Proofs of Main Results}\label{appendix:lemma-proof}
\renewcommand{\theequation}{B.\arabic{equation}}
\setcounter{equation}{0}

\subsection{Auxiliary Lemmas}\label{appendix:lemma}

\begin{lemma}\label{lemma:unimodal2}
	Let $g(t)=(at+b)\Phi\left(-t/c+d\right)$, where $a>0$, $b\ge 0$, $c>0$, and $d\in\mathbb{R}$.
    Then, there exists a unique solution $t^*$ to $c\frac{1-\Phi\left(t/c-d\right)}{\phi\left(t/c-d\right)}=t+\frac{b}{a}$, and $g(t)$ is strictly increasing on $(-\infty,t^*)$, strictly decreasing on $(t^*,\infty)$, and uniquely maximized at $t^*$ over $\mathbb{R}$.
    Moreover, $t^*<0$ if $c<2\phi(0)\frac{b}{a}$ and $d=0$, $t^*=0$ if $c=2\phi(0)\frac{b}{a}$ and $d=0$, and $t^*>0$ if $c>2\phi(0)\frac{b}{a}$ and $d=0$.
\end{lemma}
\begin{proof}
        By differentiating $g(t)$, we have
	$g'(t)=a\Phi\left(-t/c+d\right)-\frac{at+b}{c}\phi\left(-t/c+d\right)=\frac{a}{c}\left[c\frac{1-\Phi\left(t/c-d\right)}{\phi\left(t/c-d\right)}-t-\frac{b}{a}\right]\phi\left(t/c-d\right)$,
	where the second equality holds since $\Phi(x)=1-\Phi(-x)$ and $\phi(x)=\phi(-x)$.
	By the fact that the Mills ratio $\frac{1-\Phi(x)}{\phi(x)}$ of a standard normal random variable is strictly decreasing, $c\frac{1-\Phi\left(t/c-d\right)}{\phi\left(t/c-d\right)}$ is strictly decreasing in $t$.
	In addition, $c\frac{1-\Phi\left(t/c-d\right)}{\phi\left(t/c-d\right)}$ is continuous in $t$.
    Therefore, there exists a unique solution $t^*$ to $c\frac{1-\Phi\left(t/c-d\right)}{\phi\left(t/c-d\right)}=t+\frac{b}{a}$.
    Moreover, $t^*<0$ if $c\frac{1-\Phi\left(-d\right)}{\phi\left(-d\right)}<\frac{b}{a}$, $t^*=0$ if $c\frac{1-\Phi\left(-d\right)}{\phi\left(-d\right)}=\frac{b}{a}$, and $t^*>0$ if $c\frac{1-\Phi\left(-d\right)}{\phi\left(-d\right)}>\frac{b}{a}$.
    Also, $g'(t)>0$ for all $t\in(-\infty,t^*)$, $g'(t^*)=0$, and $g'(t)< 0$ for all $t\in(t^*,\infty)$.
    The statement then follows.
\end{proof}

\begin{lemma}\label{lemma:quasi-concave}
	Let $\psi(a,b)=a\Phi(-b)$.
	Then, $\psi(a,b)$ is strictly quasi-concave on $(0,\infty)\times \mathbb{R}$.
\end{lemma}

\begin{proof}
	Take any $a_0,a_1>0$ and $b_0,b_1\in\mathbb{R}$ such that $(a_0,b_0)\neq (a_1,b_1)$.
	I show that $\psi(a_0+\lambda(a_1-a_0),b_0+\lambda(b_1-b_0))>\min\{\psi(a_0,b_0),\psi(a_1,b_1)\}$ for all $\lambda\in(0,1)$.
	
	First, suppose that $a_0\le a_1$ and $b_0\ge b_1$.
	Since either $a_0<a_1$ or $b_0>b_1$ or both must hold, $\psi(a_0+\lambda(a_1-a_0),b_0+\lambda(b_1-b_0))=(a_0+\lambda(a_1-a_0))\Phi(-b_0-\lambda(b_1-b_0))$ is strictly increasing in $\lambda$.
	It then follows that $\psi(a_0+\lambda(a_1-a_0),b_0+\lambda(b_1-b_0))>\psi(a_0,b_0)$.
	Likewise, if $a_0\ge a_1$ and $b_0\le b_1$, then $\psi(a_0+\lambda(a_1-a_0),b_0+\lambda(b_1-b_0))>\psi(a_1,b_1)$.
	
	Now suppose that $a_0<a_1$ and $b_0<b_1$.
	Note that the set $\{(a_0,b_0)+\lambda(a_1-a_0,b_1-b_0):\lambda\in(0,1)\}$ is equivalent to
	$\left\{\left(0,b_0-a_0\frac{b_1-b_0}{a_1-a_0}\right)+t\left(\frac{a_1-a_0}{b_1-b_0},1\right):t\in\left(a_0\frac{b_1-b_0}{a_1-a_0},a_1\frac{b_1-b_0}{a_1-a_0}\right)\right\}$.
	We have
	$\psi\left(\left(0,b_0-a_0\frac{b_1-b_0}{a_1-a_0}\right)+t\left(\frac{a_1-a_0}{b_1-b_0},1\right)\right)=t\left(\frac{a_1-a_0}{b_1-b_0}\right)\Phi\left(-b_0+a_0\frac{b_1-b_0}{a_1-a_0}-t\right)=\left(\frac{a_1-a_0}{b_1-b_0}\right)g\left(t\right)$,
	where $g\left(t\right)=t\Phi\left(-b_0+a_0\frac{b_1-b_0}{a_1-a_0}-t\right)$.
	Lemma \ref{lemma:unimodal2} implies that the minimum of $g(t)$ over an interval $[\underline{t},\bar t]\subset \mathbb{R}_+$ is attained only at $\underline{t}$ or $\bar t$ or both.
	Hence, for all $t\in\left(a_0\frac{b_1-b_0}{a_1-a_0},a_1\frac{b_1-b_0}{a_1-a_0}\right)$,
	$g\left(t\right)>\min\left\{g\left(a_0\frac{b_1-b_0}{a_1-a_0}\right), g\left(a_1\frac{b_1-b_0}{a_1-a_0}\right)\right\}$.
	Thus, for all $t\in\left(a_0\frac{b_1-b_0}{a_1-a_0},a_1\frac{b_1-b_0}{a_1-a_0}\right)$,
	$\psi\left(\left(0,b_0-a_0\frac{b_1-b_0}{a_1-a_0}\right)+t\left(\frac{a_1-a_0}{b_1-b_0},1\right)\right)>\left(\frac{a_1-a_0}{b_1-b_0}\right)\min\left\{g\left(a_0\frac{b_1-b_0}{a_1-a_0}\right), g\left(a_1\frac{b_1-b_0}{a_1-a_0}\right)\right\}=\min\{\psi(a_0,b_0),\psi(a_1,b_1)\}$.
	Therefore, $\psi(a_0+\lambda(a_1-a_0),b_0+\lambda(b_1-b_0))>\min\{\psi(a_0,b_0),\psi(a_1,b_1)\}$ for all $\lambda\in(0,1)$.
	The same argument holds for the case in which $a_0>a_1$ and $b_0>b_1$.
\end{proof}

The following lemma derives a minimax regret rule for a class of univariate problems.

\begin{lemma}[Minimax Regret Rules for Univariate Problems]\label{lemma:bounded_normal}
	Consider a minimax regret problem in which $\Theta=[-\tau,\tau]$ for some $\tau> 0$; $\boldsymbol{m}(\theta)=\theta$; $L(\theta)=\theta$; and $\boldsymbol{\Sigma}=\sigma^2>0$.
	Then, the decision rule
	$
	\delta^*(Y)=\mathbf{1}\left\{Y\ge 0\right\}
	$
	is minimax regret.
	The minimax risk, denoted by ${\cal R}_{{\rm uni}}(\sigma;[-\tau,\tau])$, is given by
	$
	{\cal R}_{{\rm uni}}(\sigma;[-\tau,\tau])=R(\delta^*,-\theta^*)=R(\delta^*,\theta^*)=\theta^*\Phi(-\theta^*/\sigma)$,
        where $\theta^*=\min\{\tau,\tau^*\sigma\}$, and $\tau^*\in \arg\max_{t\ge 0}t\Phi(-t)$, which is unique ($\tau^*\approx 0.752$).
\end{lemma}

\begin{proof}
Let $r(\delta,\pi)=\int R(\delta,\theta)d\pi(\theta)$ be the Bayes risk of $\delta\in{\cal D}$ with respect to prior (probability distribution) $\pi$ on $\Theta$.
By Theorem 17 in Chapter 5 of \citeappendix{Berger1985book}, if $(\delta^*,\pi^*)$ satisfies
\begin{align}
	&\delta^*\in \arg\min_{\delta\in {\cal D}} r(\delta,\pi^*), \text{ and }
	R(\delta^*,\theta)\le r(\delta^*,\pi^*) \text{ for all }\theta\in \Theta,\label{cond:mmr}
\end{align}
then $\delta^*$ is a minimax regret rule. 
Below I construct $(\delta^*,\pi^*)$ that satisfies (\ref{cond:mmr}).

I first restrict the search space of decision rules to an essentially complete class.
\footnote{A class ${\cal C}$ of decision rules is essentially complete if, for any decision rule $\delta\notin{\cal C}$, there is a decision rule $\delta'\in{\cal C}$ such that $R(\delta,\theta)\ge R(\delta',\theta)$ for all $\theta\in\Theta$.}
Since $Y$ has monotone likelihood ratio and the loss function satisfies $l(1,\theta)-l(0,\theta)\ge 0$ if $\theta<0$ and $l(1,\theta)-l(0,\theta)\le 0$ if $\theta>0$, it follows from Theorem 5 in Chapter 8 of \citeappendix{Berger1985book} (which is originally from \citeappendix{karlin1956MLR}) that the class of monotone rules $\delta(Y)=0\cdot \mathbf{1}\{Y<t\}+\lambda\cdot \mathbf{1}\{Y=t\}+1\cdot \mathbf{1}\{Y>t\}$,
where $t\in\mathbb{R}$ and $\lambda\in[0,1]$,
is essentially complete.
Furthermore, since $\mathbb{P}_{\theta}(Y=t)=0$, a smaller class of threshold rules $\delta(Y)=\mathbf{1}\left\{Y\ge t\right\}$, $t\in\mathbb{R}$, is also essentially complete.
Let $\delta_{t}$ denote the threshold rule with threshold $t$.

I show that $(\delta_0,\pi^*)$ satisfies (\ref{cond:mmr}), where $\pi^*$ assigns probability $1/2$ to each of $-\theta^*$ and $\theta^*$.
Here, $\theta^*=\min\{\tau,\tau^*\sigma\}$, and $\tau^*\in \arg\max_{t\ge 0}t\Phi(-t)$, which is unique and positive by Lemma \ref{lemma:unimodal2}.
First, I verify the second condition of (\ref{cond:mmr}).
Since $Y\sim {\cal N}(\theta,\sigma^2)$,
$$
R(\delta_t,\theta)=\theta\Phi((t-\theta)/\sigma)\mathbf{1}\{\theta\ge 0\}+(-\theta)(1-\Phi((t-\theta)/\sigma))\mathbf{1}\{\theta< 0\}.
$$
In particular,
$R(\delta_0,\theta)=\theta\Phi(-\theta/\sigma)\mathbf{1}\{\theta\ge 0\}+(-\theta)(1-\Phi(-\theta/\sigma))\mathbf{1}\{\theta< 0\}=|\theta|\Phi(-|\theta|/\sigma)$.
By Lemma \ref{lemma:unimodal2}, $R(\delta_0,\theta)$ is maximized at $\theta\in\{-\theta^*,\theta^*\}$ over $\Theta$. Therefore, the second condition of (\ref{cond:mmr}) holds:
$r(\delta_0,\pi^*)=(R(\delta_0,-\theta^*)+R(\delta_0,\theta^*))/2\ge R(\delta_0,\theta)$ for all $\theta\in \Theta$.

Next, I verify the first condition of (\ref{cond:mmr}).
Since the class of threshold rules is essentially complete, $\delta_{0}\in \arg\min_{\delta\in {\cal D}} r(\delta,\pi^*)$ if $0\in\arg\min_{t\in \mathbb{R}}r(\delta_t,\pi^*)$.
Observe
\begin{align*}
	r(\delta_t,\pi^*)
    &=\left[\theta^*(1-\Phi((t+\theta^*)/\sigma))+\theta^*\Phi((t-\theta^*)/\sigma)\right]/2,\\
 \frac{\partial r(\delta_t,\pi^*)}{\partial t}
	&=\sigma^{-1}\phi((t+\theta^*)/\sigma)\left[-\theta^*+\theta^*\frac{\phi((t-\theta^*)/\sigma)}{\phi((t+\theta^*)/\sigma)}\right]/2.
\end{align*}
Since $\frac{\phi((t-\theta^*)/\sigma)}{\phi((t+\theta^*)/\sigma)}$ is nondecreasing in $t$ by the monotone likelihood ratio property and $-\theta^*+\theta^*\frac{\phi((t-\theta^*)/\sigma)}{\phi((t+\theta^*)/\sigma)}=0$ for $t=0$, it follows that $\frac{\partial r(\delta_t,\pi^*)}{\partial t}\ge 0$ if $t>0$, $\frac{\partial r(\delta_t,\pi^*)}{\partial t}=0$ if $t=0$, and $\frac{\partial r(\delta_t,\pi^*)}{\partial t}\le 0$ if $t<0$.
Therefore, $0\in\arg\min_{t\in \mathbb{R}}r(\delta_t,\pi^*)$.
Thus, I conclude that $\delta_{0}$ is minimax regret and that ${\cal R}_{{\rm uni}}(\sigma;[-\tau,\tau])=R(\delta_0,-\theta^*)=R(\delta_0,\theta^*)=\theta^*\Phi(-\theta^*/\sigma)$.
\end{proof}

In the following lemma, I use the notation introduced in Appendix \ref{section:proof}.
Recall that $\bar{\cal M}\subset\mathbb{R}^n$ is the linear span of ${\cal M}$, $\left.\bar I\right|_{\bar{\cal M}}:\bar{\cal M}\rightarrow[-\infty,\infty]$ is the restriction of $\bar I:\mathbb{R}^n\rightarrow[-\infty,\infty]$ to $\bar{\cal M}$, and $\left.\partial\bar I\right|_{\bar{\cal M}}(\boldsymbol{0})$ is the superdifferential of $\bar I_{\bar{\cal M}}$ at $\boldsymbol{0}$.
Let $\boldsymbol{d}^*$ denote an arbitrary element of $\arg\min_{\boldsymbol d\in \left.\partial\bar I\right|_{\bar{\cal M}}(\boldsymbol{0})}\|\boldsymbol d\|$ (if it exists).
For $\boldsymbol{w}\in\mathbb{R}^n$, define $\bar I'(\boldsymbol{0};\boldsymbol{w})\coloneqq 
\lim_{\epsilon\downarrow 0}\frac{\bar I(\epsilon\boldsymbol{w})-\bar I(\boldsymbol{0})}{\epsilon}$ (if it exists).
Let ${\rm int}_{\bar{\cal M}}({\cal M})$ denote the interior of ${\cal M}$ in $\bar{\cal M}$:
$$
{\rm int}_{\bar{\cal M}}({\cal M})\coloneqq \left\{\boldsymbol{\mu}\in {\cal M}: \{\tilde{\boldsymbol \mu}\in \bar{\cal M}:\|\tilde{\boldsymbol \mu}-\boldsymbol{\mu}\|<\epsilon\}\subset{\cal M} \text{ for some }\epsilon>0\right\}.
$$
An extended real-valued function $f:\mathbb{E}\rightarrow[-\infty,\infty]$ on a vector space $\mathbb{E}$ is said to be convex if its epigraph, $\{(x,y)\in\mathbb{E}\times \mathbb{R}:f(x)\le y\}$, is a convex subset of $\mathbb{E}\times \mathbb{R}$.
$f$ is said to be concave if $-f$ is convex.
Define the effective domain of $f$ as ${\rm dom}(f)\coloneqq \{x\in\mathbb{E}:f(x)<\infty\}$.
$f$ is called proper if it does not attain $-\infty$ and ${\rm dom}(f)$ is nonempty.
By the definitions of $\bar I$ and $\left.\bar I\right|_{\bar{\cal M}}$, ${\rm dom}(-\bar I)={\cal M}$ and ${\rm dom}(-\left.\bar I\right|_{\bar{\cal M}})={\cal M}$.

    
\begin{lemma}\label{lemma:superdifferential}
    Under Assumption \ref{assumption:problem}, the following holds.
\begin{enumerate}[label=(\roman*)]
    \item $\bar I(\boldsymbol{\mu})<\infty$ for all $\boldsymbol{\mu}\in {\cal M}$.
    \item $\bar I(\cdot)$ is concave on $\mathbb{R}^n$ and
    $\left.\bar I\right|_{\bar{\cal M}}(\cdot)$ is concave on $\bar{\cal M}$.
    \item\label{lemma:superdifferential:int} $\boldsymbol{0}\in {\rm int}_{\bar{\cal M}}({\cal M})$.
    \item $\left.\partial\bar I\right|_{\bar{\cal M}}(\boldsymbol{0})$ is a nonempty, convex, closed, and bounded subset of $\bar{\cal M}$.
    \item\label{lemma:superdifferential:unique-dstar} $\arg\min_{\boldsymbol d\in \left.\partial\bar I\right|_{\bar{\cal M}}(\boldsymbol{0})}\|\boldsymbol d\|$ is singleton; that is, $\boldsymbol{d}^*$ exists and is unique.
    \item\label{lemma:superdifferential:dir} For every $\boldsymbol{w}\in\bar{\cal M}$, $\bar I'(\boldsymbol{0};\boldsymbol{w})$ exists and $\bar I'(\boldsymbol{0};\boldsymbol{w})=\min_{\boldsymbol{d}\in\left.\partial\bar I\right|_{\bar{\cal M}}(\boldsymbol{0})}\boldsymbol{d}'\boldsymbol{w}$.
    \item\label{lemma:superdifferential:dir-dstar} If $\|\boldsymbol{d^*}\|>0$, then $\bar I'(\boldsymbol{0};\boldsymbol{d^*}/\|\boldsymbol{d^*}\|)=\|\boldsymbol{d^*}\|$.
    \item\label{lemma:superdifferential:omega-dif} $\omega(\epsilon)<\infty$ for all $\epsilon\ge 0$.
    Furthermore, $\omega(\cdot)$ is right differentiable at $0$ with $\omega'(0)=\|\boldsymbol{d}^*\|$.
    \item $\omega(\cdot)$ is concave and continuous on $[0,\infty)$.
    \item\label{lemma:superdifferential:mu} 
    There exists a sequence $\{\boldsymbol{\mu}_\epsilon\}_{\epsilon\in (0,\bar\epsilon)}$ in ${\cal M}$ with $\bar\epsilon>0$ such that $\boldsymbol \mu_\epsilon\in \arg\max_{\boldsymbol{\mu}\in{\cal M}:\|\boldsymbol{\mu}\|\le\epsilon}\bar I(\boldsymbol{\mu})$ for all $\epsilon\in (0,\bar\epsilon)$.
    Furthermore, if $\|\boldsymbol{d^*}\|>0$, then $\lim_{\epsilon\downarrow 0}\frac{\boldsymbol \mu_\epsilon}{\epsilon}=\frac{\boldsymbol{d^*}}{\|\boldsymbol{d^*}\|}$ for any such sequence.
    \item If $\|\boldsymbol{d^*}\|>0$, then $\frac{\boldsymbol{d^*}}{\|\boldsymbol{d^*}\|}\in\arg\max_{\boldsymbol{w}\in \bar{\cal M}:\|\boldsymbol{w}\|\le 1} \bar I'(\boldsymbol{0};\boldsymbol{w})$.
\end{enumerate}
\end{lemma}
\begin{proof}
\begin{enumerate}[label=(\roman*), wide = 0pt]
\item Under Assumption \ref{assumption:problem}, $\bar I(\boldsymbol{0})<\infty$.
Suppose to the contrary that $\bar I(\boldsymbol{\mu})=\infty$ for some nonzero $\boldsymbol{\mu}\in {\cal M}$.
    By the centrosymmetry of ${\cal M}$, $-\boldsymbol \mu\in {\cal M}$,
    and pick any $\theta_0\in\Theta$ such that $\boldsymbol{m}(\theta_0)=-\boldsymbol \mu$.
    Since $\bar I(\boldsymbol{\mu})=\infty$ and $\bar I(\boldsymbol{0})<\infty$, we can take $\theta_1\in\Theta$ such that $\boldsymbol{m}(\theta_1)=\boldsymbol \mu$ and that $L((\theta_0+\theta_1)/2)=(L(\theta_0)+L(\theta_1))/2>\bar I(\boldsymbol{0})$.
    This contradicts the definition of $\bar I(\boldsymbol{0})$, since $(\theta_0+\theta_1)/2\in \Theta$ by the convexity of $\Theta$ and $\boldsymbol{m}((\theta_0+\theta_1)/2)=\boldsymbol{0}$.

\item Note first that $-\bar I$ is a proper function from $\mathbb{R}^n$ to $(-\infty,\infty]$ and $-\left.\bar I\right|_{\bar{\cal M}}$ is a proper function from $\bar{\cal M}$ to $(-\infty,\infty]$, since their effective domain, ${\cal M}$, is nonempty and they do not attain $-\infty$:
$\bar I(\boldsymbol{\mu})=-\infty$ for any $\boldsymbol{\mu}\in\mathbb{R}^n\setminus {\cal M}$, $\left.\bar I\right|_{\bar{\cal M}}(\boldsymbol{\mu})=-\infty$ for any $\boldsymbol{\mu}\in\bar{\cal M}\setminus {\cal M}$, and $\bar I(\boldsymbol{\mu})=\left.\bar I\right|_{\bar{\cal M}}(\boldsymbol{\mu})\in (-\infty,\infty)$ for all $\boldsymbol{\mu}\in {\cal M}$.
Moreover, ${\cal M}$ is convex.
Therefore, to show the concavity of $\bar I$ and $\left.\bar I\right|_{\bar{\cal M}}$, it suffices to show the concavity of $\bar I$ on ${\cal M}$ (and therefore the concavity of $\left.\bar I\right|_{\bar{\cal M}}$ on ${\cal M}$) in the usual sense of concavity of a real-valued function over a convex set \citepappendix[p. 21]{beck2017book}.
    Pick any $\boldsymbol{\mu}_0,\boldsymbol{\mu}_1\in{\cal M}$.
     Let $\{\theta_{0,k}\}_{k=1}^\infty$ and $\{\theta_{1,k}\}_{k=1}^\infty$ be sequences in $\Theta$ such that $\boldsymbol m(\theta_{0,k})=\boldsymbol{\mu}_0$ and $\boldsymbol m(\theta_{1,k})=\boldsymbol{\mu}_1$ for all $k$ and that $\lim_{k\rightarrow \infty} L(\theta_{0,k})=\bar I(\boldsymbol{\mu}_0)$ and $\lim_{k\rightarrow \infty} L(\theta_{1,k})=\bar I(\boldsymbol{\mu}_1)$.
	Then, for each $\lambda\in [0,1]$,
	$\lambda\theta_{0,k}+(1-\lambda)\theta_{1,k}\in\Theta$ by the convexity of $\Theta$, and
	$\boldsymbol m(\lambda\theta_{0,k}+(1-\lambda)\theta_{1,k})=\lambda\boldsymbol{\mu}_0+(1-\lambda)\boldsymbol{\mu}_1$ so that
	$\bar I(\lambda\boldsymbol{\mu}_0+(1-\lambda)\boldsymbol{\mu}_1)\ge L(\lambda\theta_{0,k}+(1-\lambda)\theta_{1,k})$
	by the definition of $\bar I(\cdot)$.
	Taking the limit of the right-hand side as $k\rightarrow\infty$ gives
	$\bar I(\lambda\boldsymbol{\mu}_0+(1-\lambda)\boldsymbol{\mu}_1)\ge \lambda \bar I(\boldsymbol{\mu}_0)+(1-\lambda)\bar I(\boldsymbol{\mu}_1)$.

\item Note first that since $\boldsymbol{0}\in {\cal M}$, the linear span of ${\cal M}$ is the affine hull of ${\cal M}$ in $\mathbb{R}^n$, defined as ${\rm aff}({\cal M})\coloneqq \left\{\sum_{k=1}^K\lambda_k\boldsymbol{\mu}_k:K\in\mathbb{N},\boldsymbol{\mu}_1,...,\boldsymbol{\mu}_K\in {\cal M},\lambda_1,...,\lambda_K\in\mathbb{R},\sum_{k=1}^K\lambda_k=1\right\}$.
It follows that ${\rm int}_{\bar{\cal M}}({\cal M})$ equals the relative interior of ${\cal M}$ in $\mathbb{R}^n$, defined as:
$$
{\rm ri}({\cal M})\coloneqq \left\{\boldsymbol{\mu}\in {\cal M}: \{\tilde{\boldsymbol \mu}\in {\rm aff}({\cal M}):\|\tilde{\boldsymbol \mu}-\boldsymbol{\mu}\|<\epsilon\}\subset{\cal M} \text{ for some }\epsilon>0\right\}.
$$
Therefore, it suffices to show that $\boldsymbol{0}\in {\rm ri}({\cal M})$.
Since ${\cal M}$ is a nonempty convex subset of $\mathbb{R}^n$, ${\rm ri}({\cal M})$ is nonempty and convex \citepappendix[Theorem 6.2]{rockafellar1970}.
Pick any $\boldsymbol{\mu}\in {\rm ri}({\cal M})$.
I show that $-\boldsymbol{\mu}\in {\rm ri}({\cal M})$, which implies that $\boldsymbol{0}=\boldsymbol{\mu}/2+(-\boldsymbol{\mu})/2\in {\rm ri}({\cal M})$ by the convexity of ${\rm ri}({\cal M})$.
Take $\epsilon>0$ such that $S\coloneqq\{\tilde{\boldsymbol \mu}\in {\rm aff}({\cal M}):\|\tilde{\boldsymbol \mu}-\boldsymbol{\mu}\|<\epsilon\}\subset{\cal M}$.
Pick any $\tilde{\boldsymbol \mu}\in {\rm aff}({\cal M})$ such that $\|\tilde{\boldsymbol \mu}-(-\boldsymbol{\mu})\|<\epsilon$.
Since ${\rm aff}({\cal M})=\bar{\cal M}$ is a vector space, $-\tilde{\boldsymbol \mu}\in {\rm aff}({\cal M})$.
Also, $\|-\tilde{\boldsymbol \mu}-\boldsymbol{\mu}\|=\|\tilde{\boldsymbol \mu}-(-\boldsymbol{\mu})\|<\epsilon$. Therefore, $-\tilde{\boldsymbol \mu}\in S\subset{\cal M}$, and by the centrosymmetry of ${\cal M}$, $\tilde{\boldsymbol \mu}\in {\cal M}$.
This implies that $\{\tilde{\boldsymbol \mu}\in {\rm aff}({\cal M}):\|\tilde{\boldsymbol \mu}-(-\boldsymbol{\mu})\|<\epsilon\}\subset{\cal M}$, and hence $-\boldsymbol{\mu}\in{\rm ri}({\cal M})$.

\item It is straightforward to see that $\left.\partial\bar I\right|_{\bar{\cal M}}(\boldsymbol{0})$ is the set of negative subgradients of $-\left.\bar I\right|_{\bar{\cal M}}$ at $\boldsymbol{0}\in {\rm int}_{\bar{\cal M}}({\cal M})$.
Since $-\left.\bar I\right|_{\bar{\cal M}}$ is a proper convex function from $\bar{\cal M}$ to $(-\infty,\infty]$ with the effective domain ${\cal M}$,
the statement holds by Theorems 3.9 and 3.14 of \citeappendix{beck2017book}.

\item Since $\left.\partial\bar I\right|_{\bar{\cal M}}(\boldsymbol{0})$ is a closed and bounded subset of $\bar{\cal M}$ and $\bar{\cal M}$ is a linear subspace of $\mathbb{R}^n$, $\left.\partial\bar I\right|_{\bar{\cal M}}(\boldsymbol{0})$ is compact by the Heine-Borel theorem.
Since $\|\cdot\|:\bar{\cal M}\rightarrow\mathbb{R}$ is continuous and $\left.\partial\bar I\right|_{\bar{\cal M}}(\boldsymbol{0})$ is nonempty and compact, it follows by the extreme value theorem that $\arg\min_{\boldsymbol d\in \left.\partial\bar I\right|_{\bar{\cal M}}(\boldsymbol{0})}\|\boldsymbol d\|$ is nonempty. Moreover, $\arg\min_{\boldsymbol d\in \partial \bar I(\boldsymbol{0})}\|\boldsymbol d\|$ is singleton, since $\|\cdot\|$ is strictly convex and $\left.\partial\bar I\right|_{\bar{\cal M}}(\boldsymbol{0})$ is convex.

\item
Pick any $\boldsymbol{w}\in\bar{\cal M}$.
Define the directional derivative of $-\left.\bar I\right|_{\bar{\cal M}}$ at $\boldsymbol{0}$ in the direction $\boldsymbol{w}$ as $-\left.\bar I\right|_{\bar{\cal M}}'(\boldsymbol{0};\boldsymbol{w})\coloneqq\lim_{\epsilon\downarrow 0}\frac{-\left.\bar I\right|_{\bar{\cal M}}(\epsilon\boldsymbol{w})-\left(-\left.\bar I\right|_{\bar{\cal M}}(\boldsymbol{0})\right)}{\epsilon}$.
Since $-\left.\bar I\right|_{\bar{\cal M}}$ is a proper convex function from $\bar{\cal M}$ to $(-\infty,\infty]$ with the effective domain ${\cal M}$ and $\boldsymbol{0}\in {\rm int}_{\bar{\cal M}}({\cal M})$, it follows from Theorem 3.21 of \citeappendix{beck2017book} that $-\left.\bar I\right|_{\bar{\cal M}}'(\boldsymbol{0};\boldsymbol{w})$ exists.
Also, Theorem 3.26 of \citeappendix{beck2017book} implies that
$-\left.\bar I\right|_{\bar{\cal M}}'(\boldsymbol{0};\boldsymbol{w})=\max_{\boldsymbol{d}\in -\left.\partial\bar I\right|_{\bar{\cal M}}(\boldsymbol{0})}\boldsymbol{d}'\boldsymbol{w}=-\min_{\boldsymbol{d}\in \left.\partial\bar I\right|_{\bar{\cal M}}(\boldsymbol{0})}\boldsymbol{d}'\boldsymbol{w}$,
where $-\left.\partial\bar I\right|_{\bar{\cal M}}(\boldsymbol{0})=\{-\boldsymbol{d}\in\bar{\cal M}: \boldsymbol{d}\in \left.\partial\bar I\right|_{\bar{\cal M}}(\boldsymbol{0})\}$, which is equal to the set of subgradients of $-\left.\bar I\right|_{\bar{\cal M}}$ at $\boldsymbol{0}$.
Since $-\left.\bar I\right|_{\bar{\cal M}}'(\boldsymbol{0};\boldsymbol{w})=-\lim_{\epsilon\downarrow 0}\frac{\left.\bar I\right|_{\bar{\cal M}}(\epsilon\boldsymbol{w})-\left.\bar I\right|_{\bar{\cal M}}(\boldsymbol{0})}{\epsilon}=-\lim_{\epsilon\downarrow 0}\frac{\bar I(\epsilon\boldsymbol{w})-\bar I(\boldsymbol{0})}{\epsilon}=-\bar I'(\boldsymbol{0};\boldsymbol{w})$, the statement follows.

\item 
By part \ref{lemma:superdifferential:dir}, $\bar I'(\boldsymbol{0};\boldsymbol{d}^*/\|\boldsymbol{d}^*\|)=\min_{\boldsymbol{d}\in \left.\partial\bar I\right|_{\bar{\cal M}}(\boldsymbol{0})}\boldsymbol{d}'\boldsymbol{d}^*/\|\boldsymbol{d}^*\|$.
Thus, to prove the statement, it suffices to show that $\min_{\boldsymbol{d}\in \left.\partial\bar I\right|_{\bar{\cal M}}(\boldsymbol{0})}\boldsymbol{d}'\boldsymbol{d}^*=\|\boldsymbol{d}^*\|^2$.
Suppose to the contrary that there exists $\tilde{\boldsymbol{d}}\in \left.\partial\bar I\right|_{\bar{\cal M}}(\boldsymbol{0})$ such that $(\tilde{\boldsymbol{d}})'\boldsymbol{d}^*<\|\boldsymbol{d}^*\|^2$.
For $\lambda\in \mathbb{R}$, let $\boldsymbol{d}_\lambda=\lambda \tilde{\boldsymbol{d}}+(1-\lambda)\boldsymbol{d}^*$ and $g(\lambda)=\|\boldsymbol{d}_\lambda\|^2=\lambda^2\|\tilde{\boldsymbol{d}}\|^2+(1-\lambda)^2\|\boldsymbol{d}^*\|^2+2\lambda(1-\lambda)(\tilde{\boldsymbol{d}})'\boldsymbol{d}^*$.
Since $g'(0)=-2\|\boldsymbol{d}^*\|^2+2(\tilde{\boldsymbol{d}})'\boldsymbol{d}^*<0$, there exists $\lambda\in (0,1)$ such that $g(\lambda)<g(0)$, so that $\|\boldsymbol{d}_\lambda\|<\|\boldsymbol{d}^*\|$.
However, by the convexity of $\left.\partial\bar I\right|_{\bar{\cal M}}(\boldsymbol{0})$, $\boldsymbol{d}_\lambda\in \left.\partial\bar I\right|_{\bar{\cal M}}(\boldsymbol{0})$, which contradicts the result from part \ref{lemma:superdifferential:unique-dstar} that $\boldsymbol{d}^*$ is a solution to $\min_{\boldsymbol d\in \left.\partial\bar I\right|_{\bar{\cal M}}(\boldsymbol{0})}\|\boldsymbol d\|$.

\item 
For any $\epsilon>0$ and $\boldsymbol{\mu}\in{\cal M}$ with $\|\boldsymbol{\mu}\|\le \epsilon$, $(\boldsymbol d^*)'\boldsymbol{\mu}\le \|\boldsymbol d^*\|\|\boldsymbol{\mu}\|\le\|\boldsymbol d^*\|\epsilon$ by the Cauchy-Schwarz inequality. Since $\boldsymbol{d}^*\in\left.\partial\bar I\right|_{\bar{\cal M}}(\boldsymbol{0})$,
$\bar I(\boldsymbol{\mu})\le \bar I(\boldsymbol{0})+(\boldsymbol{d}^*)'\boldsymbol{\mu}\le \bar I(\boldsymbol{0})+\|\boldsymbol{d}^*\|\epsilon=\omega(0)+\|\boldsymbol{d}^*\|\epsilon$.
Thus we obtain
\begin{align}
    \omega(\epsilon)=\sup_{\boldsymbol{\mu}\in{\cal M}:\|\boldsymbol{\mu}\|\le\epsilon}\bar I(\boldsymbol{\mu})\le \omega(0)+\|\boldsymbol{d}^*\|\epsilon<\infty.\label{eq:omega-ineq}
\end{align}
If $\|\boldsymbol{d}^*\|=0$, $\omega(\epsilon)=\omega(0)$ for any $\epsilon>0$, since $\omega(\cdot)$ is nondecreasing. Hence, $\omega'(0)=0$.

Consider the case in which $\|\boldsymbol{d}^*\|>0$.
Let $\tilde{\boldsymbol{w}}^*=\boldsymbol{d}^*/\|\boldsymbol{d}^*\|\in\bar {\cal M}$.
For any sufficiently small $\epsilon>0$, $\epsilon\tilde{\boldsymbol{w}}^*\in {\cal M}$ since $\boldsymbol{0}\in {\rm int}_{\bar{\cal M}}({\cal M})$, and we have
$$
\bar I(\epsilon \tilde{\boldsymbol{w}}^*)\le \sup_{\boldsymbol{\mu}\in{\cal M}:\|\boldsymbol{\mu}\|\le\epsilon\|\tilde{\boldsymbol{w}}^*\|}\bar I(\boldsymbol{\mu})=\omega(\epsilon\|\tilde{\boldsymbol{w}}^*\|)=\omega(\epsilon).
$$
We therefore obtain that for any sufficiently small $\epsilon>0$,
$$
\frac{\bar I(\epsilon \tilde{\boldsymbol{w}}^*)-\bar I(\boldsymbol{0})}{\epsilon}\le \frac{\omega(\epsilon)-\omega(0)}{\epsilon}\le \|\boldsymbol{d}^*\|,
$$
where the second inequality follows from \eqref{eq:omega-ineq}.
By part \ref{lemma:superdifferential:dir-dstar}, the left-hand side converges to $\|\boldsymbol{d}^*\|$ as $\epsilon\downarrow 0$.
Consequently, $\omega(\cdot)$ is right differentiable at $0$ with $\omega'(0)=\|\boldsymbol{d}^*\|$.

\item 
    To show the concavity, pick any $\epsilon_0,\epsilon_1\ge 0$.
     Let $\{\theta_{0,k}\}_{k=1}^\infty$ and $\{\theta_{1,k}\}_{k=1}^\infty$ be sequences in $\Theta$ such that $\|\boldsymbol m(\theta_{0,k})\|\le \epsilon_0$ and $\|\boldsymbol m(\theta_{1,k})\|\le \epsilon_1$ for all $k$ and that $\lim_{k\rightarrow \infty} L(\theta_{0,k})=\omega(\epsilon_0)$ and $\lim_{k\rightarrow \infty} L(\theta_{1,k})=\omega(\epsilon_1)$.
	Then, for each $\lambda\in [0,1]$,
	$\lambda\theta_{0,k}+(1-\lambda)\theta_{1,k}\in\Theta$ by the convexity of $\Theta$, and
	$\|\boldsymbol m(\lambda\theta_{0,k}+(1-\lambda)\theta_{1,k})\|\le \|\lambda \boldsymbol m(\theta_{0,k})\|+\|(1-\lambda)\boldsymbol m(\theta_{0,k})\|\le\lambda\epsilon_0+(1-\lambda)\epsilon_1$ so that
	$\omega(\lambda\epsilon_0+(1-\lambda)\epsilon_1)\ge L(\lambda\theta_{0,k}+(1-\lambda)\theta_{1,k})$
	by the definition of $\omega(\cdot)$.
	Taking the limit of the right-hand side as $k\rightarrow\infty$ gives
	$\omega(\lambda\epsilon_0+(1-\lambda)\epsilon_1)\ge \lambda \omega(\epsilon_0)+(1-\lambda)\omega(\epsilon_1)$.

The continuity of $\omega(\cdot)$ follows from part \ref{lemma:superdifferential:omega-dif} and concavity of $\omega(\cdot)$.

\item
Since $\boldsymbol{0}\in {\rm int}_{\bar{\cal M}}({\cal M})$, we can take $\bar\epsilon>0$ such that $\{\boldsymbol{\mu}\in\bar{\cal M}:\|\boldsymbol{\mu}\|<\bar\epsilon\}\subset{\cal M}$.
First, I show that $\arg\max_{\boldsymbol{\mu}\in{\cal M}:\|\boldsymbol{\mu}\|\le\epsilon}\bar I(\boldsymbol{\mu})$ is nonempty for any $\epsilon\in (0,\bar\epsilon)$.
Pick any $\epsilon\in (0,\bar\epsilon)$.
Let $B_\epsilon=\{\boldsymbol{\mu}\in\bar{\cal M}:\|\boldsymbol{\mu}\|\le \epsilon\}=\{\boldsymbol{\mu}\in{\cal M}:\|\boldsymbol{\mu}\|\le \epsilon\}$.
Since $-\left.\bar I\right|_{\bar{\cal M}}$ is a  convex function from $\bar{\cal M}$ to $(-\infty,\infty]$ with the effective domain ${\cal M}$ and $B_\epsilon\subset {\rm int}_{\bar{\cal M}}({\cal M})$, Theorem 2.21 of \citeappendix{beck2017book} implies that $-\left.\bar I\right|_{\bar{\cal M}}$ is continuous on $B_\epsilon$.
Since $B_\epsilon$ is nonempty and compact, it follows by the extreme value theorem that $\left.\bar I\right|_{\bar{\cal M}}$ attains a maximum over $B_\epsilon$.
Since $\left.\bar I\right|_{\bar{\cal M}}(\boldsymbol{\mu})=\bar I(\boldsymbol{\mu})$ for all $\boldsymbol{\mu}\in B_\epsilon$, $\arg\max_{\boldsymbol{\mu}\in{\cal M}:\|\boldsymbol{\mu}\|\le\epsilon}\bar I(\boldsymbol{\mu})$ is nonempty.

Now, suppose $\|\boldsymbol{d^*}\|>0$ and pick any sequence $\{\boldsymbol{\mu}_\epsilon\}_{\epsilon\in (0,\bar\epsilon)}$ such that $\boldsymbol \mu_\epsilon\in \arg\max_{\boldsymbol{\mu}\in{\cal M}:\|\boldsymbol{\mu}\|\le\epsilon}\bar I(\boldsymbol{\mu})$ for all $\epsilon\in (0,\bar\epsilon)$.
I show that $\frac{\boldsymbol{d^*}}{\|\boldsymbol{d^*}\|}=\lim_{\epsilon\downarrow 0}\frac{\boldsymbol \mu_\epsilon}{\epsilon}$.
Let $\boldsymbol{w}_\epsilon=\frac{\boldsymbol{\mu}_\epsilon}{\epsilon}$.
I first show that $\lim_{\epsilon\downarrow 0}(\boldsymbol{d}^*)'\boldsymbol{w}_\epsilon=\|\boldsymbol{d}^*\|$.
We have
$$
\bar I(\boldsymbol{\mu}_\epsilon)\le \bar I(\boldsymbol{0})+(\boldsymbol{d}^*)'\boldsymbol{\mu}_\epsilon\le \bar I(\boldsymbol{0})+\|\boldsymbol{d}^*\|\epsilon,~~~\epsilon\in (0,\bar\epsilon),
$$
where the first inequality holds since $\boldsymbol{d}^*\in\left.\partial\bar I\right|_{\bar{\cal M}}(\boldsymbol{0})$ and the second by the Cauchy-Schwarz inequality.
As a result,
$$
\frac{\omega(\epsilon)-\omega(0)}{\epsilon}=\frac{\bar I(\boldsymbol{\mu}_\epsilon)-\bar I(\boldsymbol{0})}{\epsilon}\le (\boldsymbol{d}^*)'\boldsymbol{w}_\epsilon\le \|\boldsymbol{d}^*\|,~~~\epsilon\in (0,\bar\epsilon).
$$
Since $\lim_{\epsilon\downarrow 0}\frac{\omega(\epsilon)-\omega(0)}{\epsilon}= \|\boldsymbol{d}^*\|$ by part \ref{lemma:superdifferential:omega-dif}, it follows that $\lim_{\epsilon\downarrow 0}(\boldsymbol{d}^*)'\boldsymbol{w}_\epsilon=\|\boldsymbol{d}^*\|$.

Next, let $\tilde{\boldsymbol{w}}^*=\frac{\boldsymbol{d}^*}{\|\boldsymbol{d}^*\|}$ and $B_1=\{\boldsymbol{w}\in\bar{\cal M}:\|\boldsymbol{w}\|\le 1\}$.
By the Cauchy-Schwarz inequality,
$(\boldsymbol{d}^*)'\boldsymbol{w}\le \|\boldsymbol{d}^*\|\|\boldsymbol{w}\|\le \|\boldsymbol{d}^*\|$ for all $\boldsymbol{w}\in B_1$, and both inequalities hold with equality simultaneously only if $\boldsymbol{w}=\tilde{\boldsymbol{w}}^*$.
Therefore, $(\boldsymbol{d}^*)'\boldsymbol{w}< \|\boldsymbol{d}^*\|$ for all $\boldsymbol{w}\in B_1$ such that $\boldsymbol{w}\neq\tilde{\boldsymbol{w}}^*$.

Now, for each $\delta>0$, let
$B_\delta(\tilde{\boldsymbol{w}}^*)=\{\boldsymbol{w}\in B_1:\|\boldsymbol{w}-\tilde{\boldsymbol{w}}^*\|< \delta\}$.
Then it holds that $\sup_{\boldsymbol{w}\in B_1\setminus B_\delta(\tilde{\boldsymbol{w}}^*)}(\boldsymbol{d}^*)'\boldsymbol{w}<\|\boldsymbol{d}^*\|$ for every $\delta>0$ for the following reason.
If $\delta>0$ is such that $B_1\setminus B_\delta(\tilde{\boldsymbol{w}}^*)$ is empty, then $\sup_{\boldsymbol{w}\in B_1\setminus B_\delta(\tilde{\boldsymbol{w}}^*)}(\boldsymbol{d}^*)'\boldsymbol{w}=-\infty<\|\boldsymbol{d}^*\|$.
If $B_1\setminus B_\delta(\tilde{\boldsymbol{w}}^*)$ is nonempty, then $\sup_{\boldsymbol{w}\in B_1\setminus B_\delta(\tilde{\boldsymbol{w}}^*)}(\boldsymbol{d}^*)'\boldsymbol{w}=(\boldsymbol{d}^*)'\tilde{\boldsymbol{w}}$ for some $\tilde{\boldsymbol{w}}\in B_1\setminus B_\delta(\tilde{\boldsymbol{w}}^*)$ by the extreme value theorem, since the mapping $\boldsymbol{w}\mapsto(\boldsymbol{d}^*)'\boldsymbol{w}$ is continuous and $B_1\setminus B_\delta(\tilde{\boldsymbol{w}}^*)$ is compact.
This implies that $\sup_{\boldsymbol{w}\in B_1\setminus B_\delta(\tilde{\boldsymbol{w}}^*)}(\boldsymbol{d}^*)'\boldsymbol{w}<\|\boldsymbol{d}^*\|$,
since $\tilde{\boldsymbol{w}}\in B_1\setminus B_\delta(\tilde{\boldsymbol{w}}^*)$ and $(\boldsymbol{d}^*)'\boldsymbol{w}< \|\boldsymbol{d}^*\|$ for all $\boldsymbol{w}\in B_1$ such that $\boldsymbol{w}\neq\tilde{\boldsymbol{w}}^*$. 

Finally, I prove that $\lim_{\epsilon\downarrow 0}\boldsymbol{w}_\epsilon=\tilde{\boldsymbol{w}}^*$.
Pick any $\delta>0$.
Since $\lim_{\epsilon\downarrow 0}(\boldsymbol{d}^*)'\boldsymbol{w}_\epsilon=\|\boldsymbol{d}^*\|$ and $\sup_{\boldsymbol{w}\in B_1\setminus B_\delta(\tilde{\boldsymbol{w}}^*)}(\boldsymbol{d}^*)'\boldsymbol{w}<\|\boldsymbol{d}^*\|$,
there exists $\tilde\epsilon\in (0,\bar\epsilon)$ such that
$\sup_{\boldsymbol{w}\in B_1\setminus B_\delta(\tilde{\boldsymbol{w}}^*)}(\boldsymbol{d}^*)'\boldsymbol{w}<(\boldsymbol{d}^*)'\boldsymbol{w}_\epsilon$ for all $\epsilon\in (0,\tilde\epsilon)$.
This implies that for all $\epsilon\in (0,\tilde\epsilon)$, $\boldsymbol{w}_\epsilon\notin (B_1\setminus B_\delta(\tilde{\boldsymbol{w}}^*))$.
Since
$\boldsymbol{w}_\epsilon=\frac{\boldsymbol{\mu}_\epsilon}{\epsilon}\in B_1$, 
$\boldsymbol{w}_\epsilon\in B_\delta(\tilde{\boldsymbol{w}}^*)$ for all $\epsilon\in (0,\tilde\epsilon)$.
Since $\delta>0$ is arbitrary, $\lim_{\epsilon\downarrow 0}\boldsymbol{w}_\epsilon=\tilde{\boldsymbol{w}}^*$.

\item By part \ref{lemma:superdifferential:dir},
$$
\sup_{\boldsymbol{w}\in\bar{\cal M}:\|\boldsymbol{w}\|\le 1}\bar I'(\boldsymbol{0};\boldsymbol{w})=\sup_{\boldsymbol{w}\in\bar{\cal M}:\|\boldsymbol{w}\|\le 1}\min_{\boldsymbol{d}\in\left.\partial\bar I\right|_{\bar{\cal M}}(\boldsymbol{0})}\boldsymbol{d}'\boldsymbol{w}.
$$
Since both $\{\boldsymbol{w}\in\bar{\cal M}:\|\boldsymbol{w}\|\le 1\}$ and $\left.\partial\bar I\right|_{\bar{\cal M}}(\boldsymbol{0})$ are compact and convex and $\boldsymbol{d}'\boldsymbol{w}$ is a bilinear function of $(\boldsymbol{d},\boldsymbol{w})$, it follows from Sion's minimax theorem that
$$
\sup_{\boldsymbol{w}\in\bar{\cal M}:\|\boldsymbol{w}\|\le 1}\min_{\boldsymbol{d}\in\left.\partial\bar I\right|_{\bar{\cal M}}(\boldsymbol{0})}\boldsymbol{d}'\boldsymbol{w}=\min_{\boldsymbol{d}\in\left.\partial\bar I\right|_{\bar{\cal M}}(\boldsymbol{0})}\sup_{\boldsymbol{w}\in\bar{\cal M}:\|\boldsymbol{w}\|\le 1}\boldsymbol{d}'\boldsymbol{w}.
$$
As $\|\boldsymbol{d}^*\|>0$ by assumption and $\boldsymbol{d}^*\in\arg\min_{\boldsymbol{d}\in\left.\partial\bar I\right|_{\bar{\cal M}}(\boldsymbol{0})}\|\boldsymbol{d}\|$, $\|\boldsymbol{d}\|>0$ for any $\boldsymbol{d}\in\left.\partial\bar I\right|_{\bar{\cal M}}(\boldsymbol{0})$.
It is straightforward to see that $\sup_{\boldsymbol{w}\in\bar{\cal M}:\|\boldsymbol{w}\|\le 1}\boldsymbol{d}'\boldsymbol{w}=\|\boldsymbol{d}\|$ for any $\boldsymbol{d}\in\partial \bar I(\boldsymbol{0})$: by the Cauchy-Schwarz inequality, $\boldsymbol{d}'\boldsymbol{w}\le\|\boldsymbol{d}\|\|\boldsymbol{w}\|\le\|\boldsymbol{d}\|$ for any $\boldsymbol{w}\in\bar {\cal M}$ such that $\|\boldsymbol{w}\|\le 1$, but the inequalities hold with equality when $\boldsymbol{w}=\boldsymbol{d}/\|\boldsymbol{d}\|$.
Therefore,
$$
\min_{\boldsymbol{d}\in\left.\partial\bar I\right|_{\bar{\cal M}}(\boldsymbol{0})}\sup_{\boldsymbol{w}\in\bar{\cal M}:\|\boldsymbol{w}\|\le 1}\boldsymbol{d}'\boldsymbol{w}=\min_{\boldsymbol{d}\in\left.\partial\bar I\right|_{\bar{\cal M}}(\boldsymbol{0})}\|\boldsymbol{d}\|=\|\boldsymbol{d}^*\|,
$$
so $\sup_{\boldsymbol{w}\in\bar{\cal M}:\|\boldsymbol{w}\|\le 1}\bar I'(\boldsymbol{0};\boldsymbol{w})=\|\boldsymbol{d^*}\|$.
On the other hand, by part \ref{lemma:superdifferential:dir-dstar}, $\bar I'(\boldsymbol{0};\boldsymbol{d^*}/\|\boldsymbol{d^*}\|)=\|\boldsymbol{d^*}\|$.
Consequently, $\frac{\boldsymbol{d^*}}{\|\boldsymbol{d^*}\|}\in\arg_{\boldsymbol{w}\in\bar{\cal M}:\|\boldsymbol{w}\|\le 1}\bar I'(\boldsymbol{0};\boldsymbol{w})$.
\qedhere
\end{enumerate}
\end{proof}

The two lemmas below immediately follow from Lemma D.1 in Supplemental Appendix D of \citeappendix{armstrong2018optimal} and from Lemma A.1 of \citeappendix{armstrong2018optimal} (which is originally from Lemma 4 of \citeappendix{donoho1994}), respectively, in the case in which ${\cal F}={\cal G}$ and ${\cal F}$ is centrosymmetric in their notation, and hence the proof is omitted.

\begin{lemma}\label{lemma:dif_omega}
	Suppose that $L$ and $\boldsymbol{m}$ are linear; $\Theta$ is convex and centrosymmetric; that $\theta_\epsilon$ attains the modulus of continuity at $\epsilon>0$ with $\|\boldsymbol m(\theta_\epsilon)\|=\epsilon$; and that there exists $\iota\in\Theta$ such that $L(\iota)=1$ and $\theta_\epsilon+c\iota\in\Theta$ for all $c$ in a neighborhood of zero.
	Then, $\omega(\cdot)$ is differentiable at $\epsilon$ with $\omega'(\epsilon)=\frac{\epsilon}{\boldsymbol m(\iota)'\boldsymbol m(\theta_\epsilon)}$.
\end{lemma}


\begin{lemma}\label{lemma:max_bias}
    Suppose that $L$ and $\boldsymbol{m}$ are linear; that $\Theta$ is convex and centrosymmetric; that $\theta_\epsilon$ attains the modulus of continuity at $\epsilon>0$ with $\|\boldsymbol m(\theta_\epsilon)\|=\epsilon$; and that $\omega(\cdot)$ is differentiable at $\epsilon$.
    Then, the estimator for $L(\theta)$ given by $\hat L_\epsilon(\boldsymbol{Y})=\omega'(\epsilon)\frac{\boldsymbol{m}(\theta_\epsilon)'}{\|\boldsymbol{m}(\theta_\epsilon)\|}\boldsymbol{Y}$ achieves the maximum bias over $\Theta$ at $-\theta_\epsilon$ and the minimum bias over $\Theta$ at $\theta_\epsilon$. That is,
    $$
    \sup_{\theta\in\Theta}\mathbb{E}_\theta[\hat L_\epsilon(\boldsymbol{Y})-L(\theta)]=\mathbb{E}_{-\theta_\epsilon}[\hat L_\epsilon(\boldsymbol{Y})-L(-\theta_\epsilon)] \text{ and } \inf_{\theta\in\Theta}\mathbb{E}_\theta[\hat L_\epsilon(\boldsymbol{Y})-L(\theta)]=\mathbb{E}_{\theta_\epsilon}[\hat L_\epsilon(\boldsymbol{Y})-L(\theta_\epsilon)].
    $$
\end{lemma}

\subsection{Proof of Lemma \ref{lemma:one-dim0}}\label{proof:lemma:one-dim}

Note first that if $L(\bar\theta)=0$, then $L(\theta)=0$ for all $\theta\in[-\bar\theta,\bar\theta]$ by the linearity of $L$, and therefore any decision rule is minimax regret with maximum regret equal to zero.
The result for the case in which $L(\bar\theta)> 0$ and $\boldsymbol{m}(\bar\theta)= \boldsymbol{0}$ is proven after Lemma \ref{lemma:one-dim0} in Section \ref{section:one-dim}.

Consider the case in which $L(\bar\theta)> 0$ and $\boldsymbol{m}(\bar\theta)\neq \boldsymbol{0}$.
First, I show that $T(\boldsymbol{Y})=\frac{\boldsymbol{m}(\bar \theta)'\boldsymbol{Y}}{\|\boldsymbol{m}(\bar \theta)\|^2}$ is a sufficient statistic for $\lambda\in [-1,1]$, where $\lambda$ represents the location of $\theta$ within $[-\bar\theta,\bar\theta]=\{\theta\in \mathbb{V}:\theta=\lambda\bar\theta,\lambda\in [-1,1]\}$.
Let
$$h(\boldsymbol{y})=\frac{1}{\sqrt{(2\pi)^n}\sigma^n}\exp\left(-\frac{1}{2\sigma^2}\|\boldsymbol y\|^2\right),~~~g(t,\lambda)=\exp\left(-\frac{1}{2\sigma^2}\left(-2\lambda t+\lambda^2\right)\|\boldsymbol{m}(\bar \theta)\|^2\right).
$$
For $\lambda\in [-1, 1]$,
$\boldsymbol{Y}\sim {\cal N}(\lambda \boldsymbol{m}(\bar \theta),\sigma^2 \boldsymbol I_n)$ under $\lambda\bar \theta$.
It follows that the probability density of $\boldsymbol{Y}$ under $\lambda$ is given by
\begin{align*}
&f(\boldsymbol{y};\lambda)=\frac{1}{\sqrt{(2\pi)^n}\sigma^n}\exp\left(-\frac{1}{2\sigma^2}\|\boldsymbol y-\lambda \boldsymbol{m}(\bar \theta)\|^2\right)\\
&=\frac{1}{\sqrt{(2\pi)^n}\sigma^n}\exp\left(-\frac{1}{2\sigma^2}(\|\boldsymbol y\|^2-2\lambda \boldsymbol{m}(\bar \theta)'\boldsymbol{y}+\lambda^2\|\boldsymbol{m}(\bar \theta)\|^2)\right)=h(\boldsymbol{y})g(T(\boldsymbol{y}),\lambda).
\end{align*}
By the factorization theorem, $T(\boldsymbol{Y})$ is a sufficient statistic for $\lambda$.

By Theorem 1 in Chapter 1 of \citeappendix{Berger1985book}, the class of decision rules that only depend on $T(\boldsymbol{Y})$ is essentially complete.
With this restricted class of rules, the minimax regret problem for $[-\bar \theta,\bar \theta]$ is equivalent to the univariate problem in which we observe a univariate signal $T\sim {\cal N}\left(\lambda,\frac{\sigma^2}{\|\boldsymbol m(\bar\theta)\|^2}\right)$ and the welfare contrast is $\lambda L(\bar\theta)$ for $\lambda\in [-1,1]$.
This univariate problem has the same set of minimax regret rules as the univariate problem in which $T\sim {\cal N}\left(\lambda,\frac{\sigma^2}{\|\boldsymbol m(\bar\theta)\|^2}\right)$ and the welfare contrast is $\lambda$ for $\lambda\in [-1,1]$, since the constant multiplier $L(\bar\theta)$ in the welfare contrast only scales up the regret of a decision rule in the latter problem by $L(\bar\theta)$ for all $\lambda\in [-1,1]$.
The minimax risk for the former problem is that for the latter problem multiplied by $L(\bar\theta)$.
Therefore, Lemma \ref{lemma:bounded_normal} implies that $\delta(T)=\mathbf{1}\left\{T\ge 0\right\}$ is minimax regret for the former univariate problem, and its minimax risk is given by
\begin{align}
L(\bar\theta){\cal R}_{\rm uni}\left(\frac{\sigma}{\|\boldsymbol m(\bar\theta)\|};[-1,1]\right)=\begin{cases}
L(\bar\theta)\Phi\left(-\frac{\|\boldsymbol m(\bar\theta)\|}{\sigma}\right) ~~ &\text{ if } 1\le \tau^*\frac{\sigma}{\|\boldsymbol m(\bar\theta)\|},\\
\tau^*\sigma \frac{L(\bar\theta)}{\|\boldsymbol m(\bar\theta)\|}\Phi\left(-\tau^*\right) &\text{ if } 1> \tau^*\frac{\sigma}{\|\boldsymbol m(\bar\theta)\|}.
\label{eq:1d-mmrisk}
\end{cases}
\end{align}
Consequently, $\delta^*(\boldsymbol Y)=\mathbf{1}\{T(\boldsymbol{Y})\ge 0\}=\mathbf{1}\left\{\boldsymbol{m}(\bar \theta)'\boldsymbol{Y}\ge 0\right\}$ is minimax regret for the original one-dimensional subproblem $[-\bar \theta,\bar \theta]$.
The minimax risk for this problem is given by \eqref{eq:1d-mmrisk}.
\qed

\subsection{Proof of Lemma \ref{lemma:hardest-1d}}\label{proof:lemma:hardest-1d}


\paragraph{Proof of Statement \ref{lemma:hardest-1d:risk}.}
Let $g(\epsilon)=\omega(\epsilon)\Phi(-\epsilon/\sigma)$.
By Lemma \ref{lemma:mod}, $\omega(\cdot)$ is continuous, and so is $g(\cdot)$. From the extreme value theorem, $g(\cdot)$ attains a maximum over $[0,\tau^*\sigma]$.
Let $\epsilon^*\in\arg\max_{\epsilon\in[0,\tau^*\sigma]}g(\epsilon)$ and ${\cal E}=\{\|\boldsymbol{m}(\theta)\|:\theta\in\Theta\}$.
By Lemma \ref{lemma:one-dim0},
\begin{align*}
&
\sup_{\bar\theta\in\Theta:L(\bar\theta)\ge 0}{\cal R}([-\bar \theta,\bar \theta])=\sup_{\epsilon\in {\cal E}}\sup_{\bar\theta\in\Theta:\|\boldsymbol m(\bar\theta)\|=\epsilon,L(\bar\theta)\ge 0}{\cal R}([-\bar \theta,\bar \theta]) \\
&=\sup\Bigg\{\sup_{\epsilon\in {\cal E}\cap [0, \tau^*\sigma]}\sup_{\bar\theta\in\Theta:\|\boldsymbol m(\bar\theta)\|=\epsilon,L(\bar\theta)\ge 0}L(\bar\theta)\Phi\left(-\frac{\|\boldsymbol m(\bar\theta)\|}{\sigma}\right),\\
&~~~~~~~~~~~~\sup_{\epsilon\in {\cal E}\cap (\tau^*\sigma,\infty)}\sup_{\bar\theta\in\Theta:\|\boldsymbol m(\bar\theta)\|=\epsilon,L(\bar\theta)\ge 0} \dfrac{\tau^*\sigma L(\bar\theta)}{\|\boldsymbol m(\bar\theta)\|}\Phi\left(-\tau^*\right)\Bigg\}\\
&\le\sup\left\{\sup_{\epsilon\in {\cal E}\cap [0, \tau^*\sigma]}\omega(\epsilon)\Phi(-\epsilon/\sigma),\sup_{\epsilon\in {\cal E}\cap (\tau^*\sigma,\infty)}\frac{\tau^*\sigma\omega(\epsilon)}{\epsilon}\Phi(-\tau^*)\right\}\\
&\le\sup\left\{\sup_{\epsilon\in [0, \tau^*\sigma]}\omega(\epsilon)\Phi(-\epsilon/\sigma),\sup_{\epsilon>\tau^*\sigma}\frac{\tau^*\sigma\omega(\epsilon)}{\epsilon}\Phi(-\tau^*)\right\},
\end{align*}
where the second last inequality holds by the definition of $\omega(\epsilon)$.
By the concavity of $\omega(\epsilon)$, $\frac{\omega(\epsilon)}{\epsilon}$ is continuous and nonincreasing on $(0,\infty)$, and hence $\sup_{\epsilon>\tau^*\sigma}\frac{\tau^*\sigma\omega(\epsilon)}{\epsilon}\Phi(-\tau^*)=\omega(\tau^*\sigma)\Phi(-\tau^*)$.
Therefore,
$\sup_{\bar\theta\in\Theta:L(\bar\theta)\ge 0}{\cal R}([-\bar \theta,\bar \theta])\le\sup_{\epsilon\in [0, \tau^*\sigma]}\omega(\epsilon)\Phi(-\epsilon/\sigma)=\omega(\epsilon^*)\Phi(-\epsilon^*/\sigma)$.
To show the other direction of the inequality,
let $\{\theta_{\epsilon^*,k}\}_{k=1}^\infty$ be a sequence in $\Theta$ such that $L(\theta_{\epsilon^*,k})\ge 0$ and $\|\boldsymbol m(\theta_{\epsilon^*,k})\|\le \epsilon^*$ for all $k$ and that $\lim_{k\rightarrow \infty} L(\theta_{\epsilon^*,k})=\omega(\epsilon^*)$.
By definition and Lemma \ref{lemma:one-dim0}, for all $k$,
\begin{align*}
    &\sup_{\bar\theta\in\Theta:L(\bar\theta)\ge 0}{\cal R}([-\bar \theta,\bar \theta])\ge {\cal R}([-\theta_{\epsilon^*,k},\theta_{\epsilon^*,k}])=L(\theta_{\epsilon^*,k})\Phi\left(-\|\boldsymbol m(\theta_{\epsilon^*,k})\|/\sigma\right)\ge L(\theta_{\epsilon^*,k})\Phi\left(-\epsilon^*/\sigma\right).
\end{align*}
Taking the limit as $k\rightarrow\infty$ gives $\sup_{\bar\theta\in\Theta:L(\bar\theta)\ge 0}{\cal R}([-\bar \theta,\bar \theta])\ge \omega(\epsilon^*)\Phi(-\epsilon^*/\sigma)$.
\qed

\paragraph{Proof of Statement \ref{lemma:hardest-1d:epsilon}.}
Let $g(\epsilon)=\omega(\epsilon)\Phi(-\epsilon/\sigma)$.
As shown in the proof of Statement \ref{lemma:hardest-1d:risk},
$g(\cdot)$ is continuous and attains a maximum over $[0,\tau^*\sigma]$.
I show that $g(\cdot)$ is strictly quasi-concave on $[0,\infty)$, which implies that $g(\cdot)$ has a unique maximum over $[0,\tau^*\sigma]$.
Pick any $\epsilon_0,\epsilon_1$ such that $0\le \epsilon_0<\epsilon_1$.
I show that $g((1-\lambda)\epsilon_0+\lambda\epsilon_1)>\min\{g(\epsilon_0),g(\epsilon_1)\}$ for all $\lambda\in(0,1)$.
Consider two cases: (i) $\omega(\epsilon_0)>0$; and (ii) $\omega(\epsilon_0)=0$.
For case (i), let $\psi(a,b)=a\Phi(-b)$ for $a>0$ and $b\in\mathbb{R}$.
By the concavity of $\omega(\cdot)$, $g((1-\lambda)\epsilon_0+\lambda\epsilon_1)\ge ((1-\lambda)\omega(\epsilon_0)+\lambda\omega(\epsilon_1))\Phi(-(1-\lambda)\epsilon_0/\sigma-\lambda\epsilon_1/\sigma)=\psi((1-\lambda)\omega(\epsilon_0)+\lambda\omega(\epsilon_1),(1-\lambda)\epsilon_0/\sigma+\lambda\epsilon_1/\sigma)$.
Furthermore, Lemma \ref{lemma:quasi-concave} implies that
$\psi((1-\lambda)\omega(\epsilon_0)+\lambda\omega(\epsilon_1),(1-\lambda)\epsilon_0/\sigma+\lambda\epsilon_1/\sigma)>\min\{\psi(\omega(\epsilon_0),\epsilon_0/\sigma),\psi(\omega(\epsilon_1),\epsilon_1/\sigma)\}=\min\{g(\epsilon_0),g(\epsilon_1)\}$ for all $\lambda\in(0,1)$. For case (ii), we must have $\epsilon_0=0$ and $\omega(\epsilon)>0$ for all $\epsilon>0$: if $\omega(\epsilon')=0$ for some $\epsilon'>0$, then $\omega(\epsilon)=0$ for all $\epsilon\ge 0$ since $\omega(\cdot)$ is nonnegative, nondecreasing, and concave, but this contradicts the assumption that $L(\theta)>0$ for some $\theta\in\Theta$ in Assumption \ref{assumption:problem}.
As a result, $g((1-\lambda)\epsilon_0+\lambda\epsilon_1)>0=g(\epsilon_0)$ for all $\lambda\in (0,1)$.

Next, I show that $\epsilon^*>0$ if and only if $\sigma\omega'(0)>2\phi(0)\omega(0)$.
The ``if'' part immediately follows, since the right derivative of $g(\cdot)$ at $0$ is positive if $\sigma\omega'(0)>2\phi(0)\omega(0)$.
To show the ``only if'' part, suppose $\sigma\omega'(0)\le2\phi(0)\omega(0)$.
If $\omega'(0)=0$, then $\omega(\cdot)$ is constant, since $\omega(\cdot)$ is nondecreasing and concave. It follows that $g(\cdot)$ attains its unique maximum over $[0,\tau^*\sigma]$ at $0$.
If $\omega'(0)>0$, then define $f(\epsilon)=(\omega(0)+\omega'(0)\epsilon)\Phi(-\epsilon/\sigma)$ for $\epsilon\in\mathbb{R}$.
Since $\omega(\cdot)$ is concave,
$f(\epsilon)=(\omega(0)+\omega'(0)\epsilon)\Phi(-\epsilon/\sigma)\ge \omega(\epsilon)\Phi(-\epsilon/\sigma)=g(\epsilon)$
for all $\epsilon\ge 0$.
The inequality holds with equality when $\epsilon=0$, so to show that $0\in\arg\max_{\epsilon\in [0,\tau^*\sigma]}g(\epsilon)$, it suffices to show that $0\in\arg\max_{\epsilon\in [0,\tau^*\sigma]}f(\epsilon)$.
Since $\sigma\le2\phi(0)\frac{\omega(0)}{\omega'(0)}$, it follows from Lemma \ref{lemma:unimodal2} that $f(\epsilon)$ is decreasing on $[c,\infty)$ for some $c\le 0$.
Therefore, $0\in\arg\max_{\epsilon\in [0,\tau^*\sigma]}f(\epsilon)$.
\qed

\paragraph{Proof of Statement \ref{lemma:hardest-1d:sol}.}
Pick any $\theta_{\epsilon^*}\in\Theta$ that attains the modulus of continuity at $\epsilon^*$.
By the definition of $\bar I(\cdot)$ in Section \ref{section:mmr},
it is straightforward to see that 
$\boldsymbol{m}(\theta_{\epsilon^*})\in \arg\max_{\boldsymbol{\mu}\in{\cal M}:\|\boldsymbol{\mu}\|\le \epsilon^*}\bar I(\boldsymbol{\mu})$.
Below, I first show that $\|\boldsymbol{\mu}_{\epsilon^*}\|=\epsilon^*$ for any $\boldsymbol{\mu}_{\epsilon^*}\in\arg\max_{\boldsymbol{\mu}\in{\cal M}:\|\boldsymbol{\mu}\|\le \epsilon^*}\bar I(\boldsymbol{\mu})$, which implies that $\|\boldsymbol{m}(\theta_{\epsilon^*})\|=\epsilon^*$ and that
${\cal R}([-\theta_{\epsilon^*},\theta_{\epsilon^*}])=L(\theta_{\epsilon^*})\Phi\left(-\|\boldsymbol m(\theta_{\epsilon^*})\|/\sigma\right)=\omega(\epsilon^*)\Phi\left(-\epsilon^*/\sigma\right)=\sup_{\bar\theta\in\Theta:L(\bar\theta)\ge 0}{\cal R}([-\bar \theta,\bar \theta])$,
where the first and last equalities hold by Lemma \ref{lemma:one-dim0} and Statement \ref{lemma:hardest-1d:risk}, respectively.

Pick any $\boldsymbol{\mu}_{\epsilon^*}\in\arg\max_{\boldsymbol{\mu}\in{\cal M}:\|\boldsymbol{\mu}\|\le \epsilon^*}\bar I(\boldsymbol{\mu})$.
Then there exists a sequence $\{\theta_k\}_{k=1}^\infty$ in $\Theta$ such that $\boldsymbol{m}(\theta_k)=\boldsymbol{\mu}_{\epsilon^*}$ for all $k$ and that $\lim_{k\rightarrow\infty}L(\theta_k)=\bar I(\boldsymbol{\mu}_{\epsilon^*})=\omega(\epsilon^*)$.
By definition and Lemma \ref{lemma:one-dim0},
$$
\sup_{\bar\theta\in\Theta:L(\bar\theta)\ge 0}{\cal R}([-\bar \theta,\bar \theta])\ge {\cal R}([-\theta_k,\theta_k])=L(\theta_k)\Phi\left(-\|\boldsymbol{m}(\theta_k)\|/\sigma\right)=L(\theta_k)\Phi\left(-\|\boldsymbol{\mu}_{\epsilon^*}\|/\sigma\right).
$$
Taking the limit as $k\rightarrow\infty$ gives
$\sup_{\bar\theta\in\Theta:L(\bar\theta)\ge 0}{\cal R}([-\bar \theta,\bar \theta])\ge \omega(\epsilon^*)\Phi\left(-\|\boldsymbol{\mu}_{\epsilon^*}\|/\sigma\right)\ge \omega(\epsilon^*)\Phi\left(-\epsilon^*/\sigma\right)$.
Since $\sup_{\bar\theta\in\Theta:L(\bar\theta)\ge 0}{\cal R}([-\bar \theta,\bar \theta])= \omega(\epsilon^*)\Phi(-\epsilon^*/\sigma)$ by Statement \ref{lemma:hardest-1d:risk}, the inequalities hold with equality, so that $\omega(\epsilon^*)\Phi\left(-\|\boldsymbol{\mu}_{\epsilon^*}\|/\sigma\right)= \omega(\epsilon^*)\Phi\left(-\epsilon^*/\sigma\right)$.
Note that $\omega(\epsilon^*)>0$, since if
$\omega(\epsilon^*)=0$, then
$\omega(\epsilon^*)\Phi\left(-\epsilon^*/\sigma\right)=0\le \omega(\epsilon)\Phi\left(-\epsilon/\sigma\right)$ for all $\epsilon\ge 0$, which contradicts the fact that $\epsilon^*$ is a unique maximizer of $\omega(\epsilon)\Phi\left(-\epsilon/\sigma\right)$ over $[0,\tau^*\sigma]$.
It follows that $\Phi\left(-\|\boldsymbol{\mu}_{\epsilon^*}\|/\sigma\right)= \Phi\left(-\epsilon^*/\sigma\right)$ and hence
$\|\boldsymbol{\mu}_{\epsilon^*}\|=\epsilon^*$.

	Next, I show that $\arg\max_{\boldsymbol{\mu}\in{\cal M}:\|\boldsymbol{\mu}\|\le \epsilon^*}\bar I(\boldsymbol{\mu})$ is singleton, which implies that $\boldsymbol{m}(\theta_{\epsilon^*})$ does not depend on the choice of $\theta_{\epsilon^*}$ among potentially multiple $\theta$'s that attain the modulus of continuity at $\epsilon^*$.
    Pick any $\boldsymbol{\mu}_{\epsilon^*}, \tilde{\boldsymbol{\mu}}_{\epsilon^*}\in\arg\max_{\boldsymbol{\mu}\in{\cal M}:\|\boldsymbol{\mu}\|\le \epsilon^*}\bar I(\boldsymbol{\mu})$, and
	suppose that $\boldsymbol{\mu}_{\epsilon^*}\neq \tilde{\boldsymbol{\mu}}_{\epsilon^*}$. Let $\boldsymbol{\mu}_\lambda=\lambda \boldsymbol{\mu}_{\epsilon^*}+(1-\lambda)\tilde{\boldsymbol{\mu}}_{\epsilon^*}\in {\cal M}$ for some $\lambda\in(0,1)$.
    Then $\boldsymbol{\mu}_\lambda$ is an interior point of $\{\boldsymbol{\mu}\in\mathbb{R}^n:\|\boldsymbol{\mu}\|\le \epsilon^*\}$, since $\boldsymbol{\mu}_{\epsilon^*}, \tilde{\boldsymbol{\mu}}_{\epsilon^*}\in \{\boldsymbol{\mu}\in\mathbb{R}^n:\|\boldsymbol{\mu}\|\le \epsilon^*\}$ and $\{\boldsymbol{\mu}\in\mathbb{R}^n:\|\boldsymbol{\mu}\|\le \epsilon^*\}$ is strictly convex. It follows that $\|\boldsymbol{\mu}_\lambda\|<\epsilon^*$.
	Furthermore, by the concavity of $\bar I(\cdot)$ shown in Lemma \ref{lemma:superdifferential}, $\bar I(\boldsymbol{\mu}_\lambda)\ge \lambda \bar I(\boldsymbol{\mu}_{\epsilon^*})+(1-\lambda)\bar I(\tilde{\boldsymbol{\mu}}_{\epsilon^*})=\max_{\boldsymbol{\mu}\in{\cal M}:\|\boldsymbol{\mu}\|\le \epsilon^*}\bar I(\boldsymbol{\mu})$.
	Thus, $\boldsymbol{\mu}_\lambda\in\arg\max_{\boldsymbol{\mu}\in{\cal M}:\|\boldsymbol{\mu}\|\le \epsilon^*}\bar I(\boldsymbol{\mu})$, but then it must be the case that $\|\boldsymbol{\mu}_\lambda\|=\epsilon^*$, which is a contradiction.
\qed

\subsection{Proof of Lemma \ref{theorem:minimax_nonrandom}}\label{proof:theorem:minimax_nonrandom}

Note first that since $\bar I(\boldsymbol{\mu}_{\epsilon^*})=\omega(\epsilon^*)\ge 0$, $\bar I(\boldsymbol{\mu}_{\epsilon^*})\Phi\left(-\frac{(\boldsymbol{\mu}_{\epsilon^*})'(\boldsymbol{\mu}_{\epsilon^*})}{\sigma\|\boldsymbol{\mu}_{\epsilon^*}\|}\right)\ge \bar I(\boldsymbol{\mu})\Phi\left(-\frac{(\boldsymbol{\mu}_{\epsilon^*})'\boldsymbol{\mu}}{\sigma\|\boldsymbol{\mu}_{\epsilon^*}\|}\right)$ for any $\boldsymbol{\mu}\in {\cal M}$ such that $\bar I(\boldsymbol{\mu})\le 0$.
It then suffices to show that $\bar I(\boldsymbol{\mu}_{\epsilon^*})\Phi\left(-\frac{(\boldsymbol{\mu}_{\epsilon^*})'(\boldsymbol{\mu}_{\epsilon^*})}{\sigma\|\boldsymbol{\mu}_{\epsilon^*}\|}\right)\ge \bar I(\boldsymbol{\mu})\Phi\left(-\frac{(\boldsymbol{\mu}_{\epsilon^*})'\boldsymbol{\mu}}{\sigma\|\boldsymbol{\mu}_{\epsilon^*}\|}\right)$ for any $\boldsymbol{\mu}\in \tilde{\cal{M}}_1$, where
$\tilde{\cal{M}}_1=\{\boldsymbol{\mu}\in {\cal M}:\bar I(\boldsymbol{\mu})> 0\}$.
Let
$\tilde{\cal{M}}_2=\{\boldsymbol{\mu}\in {\cal M}:\bar I(\boldsymbol{\mu})> 0,(\boldsymbol{\mu}_{\epsilon^*})'\boldsymbol{\mu}>0\}$.
$\tilde{\cal M}_1$ and $\tilde{\cal M}_2$ are nonempty, since $\boldsymbol{\mu}_{\epsilon^*}\in \tilde{\cal M}_2$: in particular, $\bar I(\boldsymbol{\mu}_{\epsilon^*})=\omega(\epsilon^*)>0$, since if
$\omega(\epsilon^*)=0$, then
$\omega(\epsilon^*)\Phi\left(-\epsilon^*/\sigma\right)=0\le \omega(\epsilon)\Phi\left(-\epsilon/\sigma\right)$ for all $\epsilon\ge 0$, which contradicts the fact that $\epsilon^*$ is a unique maximizer of $\omega(\epsilon)\Phi\left(-\epsilon/\sigma\right)$ over $[0,\tau^*\sigma]$.
Also, $\tilde{\cal M}_1$ and $\tilde{\cal M}_2$ are convex, since the functions $\bar I(\boldsymbol{\mu})$ and $\boldsymbol{\mu}\mapsto (\boldsymbol{\mu}_{\epsilon^*})'\boldsymbol{\mu}$ are concave on a convex set ${\cal M}$; see Lemma \ref{lemma:superdifferential} for the concavity of $\bar I(\cdot)$.
Define functions $f:\tilde{\cal M}_1\rightarrow \mathbb{R}$ and $g:\tilde{\cal M}_1\times \tilde{\cal M}_2\rightarrow \mathbb{R}$ as
\begin{align*}
    f(\boldsymbol{\mu})&=\bar I(\boldsymbol{\mu})\Phi\left(-\frac{\|\boldsymbol{\mu}\|}{\sigma}\right),~~
    g(\boldsymbol{\mu},\tilde{\boldsymbol{\mu}})=\bar I(\boldsymbol{\mu})\Phi\left(-\frac{\tilde{\boldsymbol{\mu}}'\boldsymbol{\mu}}{\sigma\|\tilde{\boldsymbol{\mu}}\|}\right).
\end{align*}
As $\boldsymbol{0}\notin \tilde{\cal M}_2$, $g(\boldsymbol{\mu},\tilde{\boldsymbol{\mu}})$ is well defined for every $(\boldsymbol{\mu},\tilde{\boldsymbol{\mu}})\in\tilde{\cal M}_1\times \tilde{\cal M}_2$.

Below, I show that
$g(\boldsymbol{\mu}_{\epsilon^*},\boldsymbol{\mu}_{\epsilon^*})\ge g(\boldsymbol{\mu},\boldsymbol{\mu}_{\epsilon^*})$ for all $\boldsymbol{\mu}\in \tilde{\cal M}_1$.
The proof consists of four steps.

\begin{step}\label{step:g-min}
    $g(\boldsymbol{\mu},\tilde{\boldsymbol{\mu}})\ge g(\boldsymbol{\mu},\boldsymbol{\mu})$ for all $(\boldsymbol{\mu},\tilde{\boldsymbol{\mu}})\in\tilde{\cal M}_2\times \tilde{\cal M}_2$.
\end{step}
\begin{proof}
    Pick any $\boldsymbol{\mu}\in\tilde{\cal M}_2$. I show that $\boldsymbol{\mu}\in\arg\min_{\tilde{\boldsymbol{\mu}}\in\tilde{\cal M}_2}g(\boldsymbol{\mu},\tilde{\boldsymbol{\mu}})$.
    Since $\bar I(\boldsymbol{\mu})> 0$, $\arg\min_{\tilde{\boldsymbol{\mu}}\in\tilde{\cal M}_2}g(\boldsymbol{\mu},\tilde{\boldsymbol{\mu}})=\argmax_{\tilde{\boldsymbol{\mu}}\in\tilde{\cal M}_2}\left(\frac{\tilde{\boldsymbol{\mu}}}{\|\tilde{\boldsymbol{\mu}}\|}\right)'\boldsymbol{\mu}$, so it suffices to show that $\boldsymbol \mu \in \argmax_{\tilde{\boldsymbol{\mu}}\in\tilde{\cal M}_2}\left(\frac{\tilde{\boldsymbol{\mu}}}{\|\tilde{\boldsymbol{\mu}}\|}\right)'\boldsymbol{\mu}$.
    It is straightforward to see that $\boldsymbol \mu \in \argmax_{\tilde{\boldsymbol{\mu}}\in\mathbb{R}^n:\|\tilde{\boldsymbol{\mu}}\|>0}\left(\frac{\tilde{\boldsymbol{\mu}}}{\|\tilde{\boldsymbol{\mu}}\|}\right)'\boldsymbol{\mu}$: by the Cauchy-Schwarz inequality, $\left(\frac{\tilde{\boldsymbol{\mu}}}{\|\tilde{\boldsymbol{\mu}}\|}\right)'\boldsymbol{\mu}\le \frac{\|\tilde{\boldsymbol{\mu}}\|}{\|\tilde{\boldsymbol{\mu}}\|}\|\boldsymbol{\mu}\|=\|\boldsymbol{\mu}\|$ for any nonzero $\tilde{\boldsymbol{\mu}}\in\mathbb{R}^n$, but the inequality holds with equality when $\tilde{\boldsymbol \mu}=\boldsymbol \mu$.
    Since $\boldsymbol{\mu}\in\tilde{\cal M}_2\subset \{\tilde{\boldsymbol{\mu}}\in\mathbb{R}^n:\|\tilde{\boldsymbol{\mu}}\|>0\}$, it follows that
    $\boldsymbol \mu \in \argmax_{\tilde{\boldsymbol{\mu}}\in\tilde{\cal M}_2}\left(\frac{\tilde{\boldsymbol{\mu}}}{\|\tilde{\boldsymbol{\mu}}\|}\right)'\boldsymbol{\mu}$.
\end{proof}

\begin{step}\label{step:f-max}
$f(\boldsymbol{\mu}_{\epsilon^*})\ge f(\boldsymbol{\mu})$ for all $\boldsymbol{\mu}\in\tilde{\cal M}_1$.
\end{step}
\begin{proof}
    Pick any $\boldsymbol{\mu}\in \tilde{\cal M}_1$.
    Since $f(\boldsymbol{\mu}_{\epsilon^*})=\omega(\epsilon^*)\Phi(-\epsilon^*/\sigma)$, I show that $f(\boldsymbol{\mu})\le\omega(\epsilon^*)\Phi(-\epsilon^*/\sigma)$.
    First, if $\|\boldsymbol{\mu}\|\le \tau^*\sigma$, then $f(\boldsymbol{\mu})\le \omega(\|\boldsymbol{\mu}\|)\Phi(-\|\boldsymbol{\mu}\|/\sigma)\le\omega(\epsilon^*)\Phi(-\epsilon^*/\sigma)$, where the first inequality holds by the definitions of $\bar I(\cdot)$ and $\omega(\cdot)$ and the second inequality by the fact that $\epsilon^*$ is a solution to $\max_{\epsilon\in [0,\tau^*\sigma]}\omega(\epsilon)\Phi(-\epsilon/\sigma)$.

    Next, suppose that $\|\boldsymbol{\mu}\|>\tau^*\sigma$.
    Since $\boldsymbol{\mu}\in \tilde{\cal M}_1$, there exists $\bar\theta$ in $\Theta$ such that $\boldsymbol{m}(\bar\theta)=\boldsymbol{\mu}$ and that $L(\bar\theta)>0$.
    Let $\delta(\boldsymbol{Y})=\mathbf{1}\{\boldsymbol{\mu}'\boldsymbol{Y}\ge 0\}$, whose regret under $\bar\theta$ is given by
    $R(\delta,\bar\theta)=L(\bar\theta)\Phi(-\|\boldsymbol{\mu}\|/\sigma)$.
    By Lemma \ref{lemma:one-dim0}, $\delta$ is minimax regret for the one-dimensional subproblem $[-\bar\theta,\bar\theta]$, so that
    $L(\bar\theta)\Phi(-\|\boldsymbol{\mu}\|/\sigma)=R(\delta,\bar\theta)\le \sup_{\theta\in [-\bar\theta,\bar\theta]}R(\delta,\theta)=\tau^*\sigma L(\bar\theta)\Phi(-\tau^*)/\|\boldsymbol{\mu}\|$.
    This inequality implies that $\Phi(-\|\boldsymbol{\mu}\|/\sigma)\le \tau^*\sigma\Phi(-\tau^*)/\|\boldsymbol{\mu}\|$, since $L(\bar\theta)>0$.
    Therefore,
    \begin{align*}
        &f(\boldsymbol{\mu})=\bar I(\boldsymbol{\mu})\Phi(-\|\boldsymbol{\mu}\|/\sigma)\le\tau^*\sigma\frac{\bar I(\boldsymbol{\mu})}{\|\boldsymbol{\mu}\|}\Phi(-\tau^*)\le \tau^*\sigma\frac{\omega(\|\boldsymbol{\mu}\|)}{\|\boldsymbol{\mu}\|}\Phi(-\tau^*)\\
        &\le\omega(\tau^*\sigma)\Phi(-\tau^*)\le \omega(\epsilon^*)\Phi(-\epsilon^*/\sigma),
    \end{align*}
    where the second last inequality follows since $\frac{\omega(\epsilon)}{\epsilon}$ is nonincreasing on $(0,\infty)$ by the concavity of $\omega(\epsilon)$, and the last inequality by the definition of $\epsilon^*$.
\end{proof}

\begin{step}\label{step:g-quasi}
    $g(\cdot,\tilde{\boldsymbol{\mu}})$ is strictly quasi-concave on $\tilde{\cal M}_1$ for every $\tilde{\boldsymbol{\mu}}\in \tilde{\cal M}_2$.
\end{step}
\begin{proof}
Pick any $\tilde{\boldsymbol{\mu}}\in \tilde{\cal M}_2$, $\boldsymbol{\mu}_0,\boldsymbol{\mu}_1\in \tilde{\cal M}_1$, and $\lambda\in (0,1)$.
Let $\boldsymbol{\mu}_\lambda=(1-\lambda)\boldsymbol{\mu}_0+\lambda \boldsymbol{\mu}_1$.
By the concavity of $\bar I(\cdot)$, we have
$g(\boldsymbol{\mu}_\lambda,\tilde{\boldsymbol{\mu}})\ge \left((1-\lambda)\bar I(\boldsymbol{\mu}_0)+\lambda \bar I(\boldsymbol{\mu}_1)\right)\Phi\left(-\left((1-\lambda)\frac{\tilde{\boldsymbol{\mu}}'\boldsymbol{\mu}_0}{\sigma\|\tilde{\boldsymbol{\mu}}\|}+\lambda\frac{\tilde{\boldsymbol{\mu}}'\boldsymbol{\mu}_1}{\sigma\|\tilde{\boldsymbol{\mu}}\|}\right)\right)$.
By Lemma \ref{lemma:quasi-concave}, the function $(a,b)\mapsto a\Phi(-b)$ is strictly quasi-concave on $(0,\infty)\times \mathbb{R}$.
It follows that
$g(\boldsymbol{\mu}_\lambda,\tilde{\boldsymbol{\mu}})> \min\left\{\bar I(\boldsymbol{\mu}_0)\Phi\left(-\frac{\tilde{\boldsymbol{\mu}}'\boldsymbol{\mu}_0}{\sigma\|\tilde{\boldsymbol{\mu}}\|}\right),\bar I(\boldsymbol{\mu}_1)\Phi\left(-\frac{\tilde{\boldsymbol{\mu}}'\boldsymbol{\mu}_1}{\sigma\|\tilde{\boldsymbol{\mu}}\|}\right)\right\}=\min\{g(\boldsymbol{\mu}_0,\tilde{\boldsymbol{\mu}}),g(\boldsymbol{\mu}_1,\tilde{\boldsymbol{\mu}})\}$.
\end{proof}

\begin{step}
    $g(\boldsymbol{\mu}_{\epsilon^*},\boldsymbol{\mu}_{\epsilon^*})\ge g(\boldsymbol{\mu},\boldsymbol{\mu}_{\epsilon^*})$ for all $\boldsymbol{\mu}\in \tilde{\cal M}_1$.
\end{step}
\begin{proof}
    Suppose to the contrary that $g(\boldsymbol{\mu}^*,\boldsymbol{\mu}_{\epsilon^*})>g(\boldsymbol{\mu}_{\epsilon^*},\boldsymbol{\mu}_{\epsilon^*})$ for some $\boldsymbol{\mu}^*\in \tilde{\cal M}_1$.
    Let $\boldsymbol{\mu}_\lambda=(1-\lambda)\boldsymbol{\mu}_{\epsilon^*}+\lambda\boldsymbol{\mu}^*$ for $\lambda\in (0,1)$.
    By the convexity of $\tilde{\cal M}_1$ and the continuity of the function $\boldsymbol{\mu}\mapsto (\boldsymbol{\mu}_{\epsilon^*})'\boldsymbol{\mu}$ at $\boldsymbol{\mu}=\boldsymbol{\mu}_{\epsilon^*}$, $\boldsymbol{\mu}_\lambda\in \tilde {\cal M}_2$ for any $\lambda\in (0,1)$ sufficiently close to $0$.
    Also, by the continuity of the function $\tilde{\boldsymbol{\mu}}\mapsto g(\boldsymbol{\mu}^*,\tilde{\boldsymbol{\mu}})$ at $\tilde{\boldsymbol{\mu}}=\boldsymbol{\mu}_{\epsilon^*}$ and since $g(\boldsymbol{\mu}^*,\boldsymbol{\mu}_{\epsilon^*})>g(\boldsymbol{\mu}_{\epsilon^*},\boldsymbol{\mu}_{\epsilon^*})$,
    there exists $\lambda\in (0,1)$ such that $\boldsymbol{\mu}_\lambda\in \tilde {\cal M}_2$ and $g(\boldsymbol{\mu}^*,\boldsymbol{\mu}_\lambda)> g(\boldsymbol{\mu}_{\epsilon^*},\boldsymbol{\mu}_{\epsilon^*})$.
    Pick such a $\lambda\in (0,1)$.
    By Step \ref{step:g-min}, $g(\boldsymbol{\mu}_{\epsilon^*},\boldsymbol{\mu}_\lambda)\ge g(\boldsymbol{\mu}_{\epsilon^*},\boldsymbol{\mu}_{\epsilon^*})$, and therefore $\min\{g(\boldsymbol{\mu}^*,\boldsymbol{\mu}_\lambda),g(\boldsymbol{\mu}_{\epsilon^*},\boldsymbol{\mu}_\lambda)\}\ge g(\boldsymbol{\mu}_{\epsilon^*},\boldsymbol{\mu}_{\epsilon^*})$.
    Furthermore, by Step \ref{step:f-max}, we have $g(\boldsymbol{\mu}_{\epsilon^*},\boldsymbol{\mu}_{\epsilon^*})=f(\boldsymbol{\mu}_{\epsilon^*})\ge f(\boldsymbol{\mu}_\lambda) =g(\boldsymbol{\mu}_\lambda,\boldsymbol{\mu}_\lambda)$.
    It follows that $\min\{g(\boldsymbol{\mu}^*,\boldsymbol{\mu}_\lambda),g(\boldsymbol{\mu}_{\epsilon^*},\boldsymbol{\mu}_\lambda)\}\ge g(\boldsymbol{\mu}_\lambda,\boldsymbol{\mu}_\lambda)$, which is a contradiction to the strict quasi-concavity of $g(\cdot,\boldsymbol{\mu}_\lambda)$ shown by Step \ref{step:g-quasi}.
\end{proof}

\subsection{Proof of Theorem \ref{theorem:minimax-bias}}\label{proof:low}

    Pick any $\epsilon_0\ge 0$.
    Then, $\epsilon_0\in\arg\max_{\epsilon\ge 0}(\omega(\epsilon)-\epsilon\sqrt{V_{\epsilon_0}}/\sigma)$, where $V_{\epsilon_0}=(\sigma\omega'(\epsilon_0))^2$, since the function $\epsilon\mapsto\omega(\epsilon)-\epsilon\sqrt{V_{\epsilon_0}}/\sigma$ is concave on $[0,\infty)$ and its derivative at $\epsilon_0$, $\omega'(\epsilon_0)-\sqrt{V_{\epsilon_0}}/\sigma$, is zero.
    The statements in Theorem \ref{theorem:minimax-bias} then follow from Theorem \ref{theorem:low} with $V=V_{\epsilon_0}$.
    \qed

\begin{theorem}[Optimal Bias-Variance Tradeoff] \label{theorem:low}
    Suppose that Assumption \ref{assumption:problem} holds and that $\omega'(0)>0$. Then, the following holds.
    \begin{enumerate}[label=(\roman*)]
        \item\label{theorem:low-bound} Suppose that $\omega(\cdot)$ is differentiable on $(0,\infty)$.
        Let $V\ge 0$ and suppose that there exists $\epsilon_V\in\arg\max_{\epsilon\ge 0}(\omega(\epsilon)-\epsilon\sqrt{V}/\sigma)$ and that $\theta_{\epsilon_V}\in\Theta$ attains the modulus of continuity at $\epsilon_V$ with $\|\boldsymbol m(\theta_{\epsilon_V})\|=\epsilon_V$.
        Then, $\hat L_{\epsilon_V}$ has constant variance of $V$ or less:
        \begin{align*}
            \Var(\hat L_{\epsilon_V},\theta)=\begin{cases}
            (\sigma\omega'(0))^2 \le V & ~~\text{if } \epsilon_V=0,\\
            (\sigma\omega'(\epsilon_V))^2=V & ~~\text{if } \epsilon_V>0.
        \end{cases}
        \end{align*}
        Furthermore, $\hat L_{\epsilon_V}$ minimizes the maximum squared bias among all estimators with the maximum variance less than or equal to $V$:
        $$\sup_{\theta\in\Theta}{\rm Bias}(\hat L_{\epsilon_V},\theta)^2=\inf_{\tilde L\in {\cal C}(V)}\sup_{\theta\in\Theta}{\rm Bias}(\tilde L,\theta)^2=(\omega(\epsilon_V)-\epsilon_V\omega'(\epsilon_V))^2.$$
        \item\label{theorem:low-unbound} $\hat L_{0}$ minimizes the maximum squared bias among all estimators:
        $$\sup_{\theta\in\Theta}{\rm Bias}(\hat L_{0},\theta)^2=\inf_{\tilde L\in {\cal C}}\sup_{\theta\in\Theta}{\rm Bias}(\tilde L,\theta)^2=\omega(0)^2.$$
    \end{enumerate}
\end{theorem}

\paragraph{Proof of Statement \ref{theorem:low-bound}.}
Clearly, $\hat L_{\epsilon_V}$ has constant variance of $(\sigma\omega'(\epsilon_V))^2$.
Consider the case in which $\epsilon_V=0$.
Since $0\in\arg\max_{\epsilon\ge 0}(\omega(\epsilon)-\epsilon\sqrt{V}/\sigma)$, the right derivative of $\omega(\epsilon)-\epsilon\sqrt{V}/\sigma$ at $0$ must be nonpositive, which implies that $(\sigma\omega'(0))^2\le V$.
By Statement \ref{theorem:low-unbound}, which is proven below, $\sup_{\theta\in\Theta}{\rm Bias}(\hat L_{0},\theta)^2=\inf_{\tilde L\in {\cal C}(V)}\sup_{\theta\in\Theta}{\rm Bias}(\tilde L,\theta)^2=\omega(0)^2$.

Next, consider the case in which $\epsilon_V>0$.
Since $\epsilon_V\in\arg\max_{\epsilon\ge 0}(\omega(\epsilon_V)-\epsilon\sqrt{V}/\sigma)$, the derivative of $\omega(\epsilon)-\epsilon\sqrt{V}/\sigma$ at $0$ must be zero, which implies that $(\sigma\omega'(\epsilon_V))^2=V$.
By Lemma \ref{lemma:max_bias}, $\hat L_{\epsilon_V}$ achieves the maximum squared bias over $\Theta$ at $-\theta_{\epsilon_V}$ and $\theta_{\epsilon_V}$, so we obtain
\begin{align*}
\sup_{\theta\in\Theta}{\rm Bias}(\hat L_{\epsilon_V},\theta)^2={\rm Bias}(\hat L_{\epsilon_V},-\theta_{\epsilon_V})^2={\rm Bias}(\hat L_{\epsilon_V},\theta_{\epsilon_V})^2=\left(\omega(\epsilon_V)-\omega'(\epsilon_V)\epsilon_V\right)^2.
\end{align*}
Hence, it suffices to show that $\inf_{\tilde L\in {\cal C}(V)}\sup_{\theta\in\Theta}{\rm Bias}(\tilde L,\theta)^2\ge \left(\omega(\epsilon_V)-\omega'(\epsilon_V)\epsilon_V\right)^2$.
Since the one-dimensional subfamily $[-\theta_{\epsilon_V},\theta_{\epsilon_V}]$ is a subset of $\Theta$,
$$
\inf_{\tilde L\in {\cal C}(V)}\sup_{\theta\in\Theta}{\rm Bias}(\tilde L,\theta)^2\ge \inf_{\tilde L\in {\cal C}(V)}\sup_{\theta\in[-\theta_{\epsilon_V},\theta_{\epsilon_V}]}{\rm Bias}(\tilde L,\theta)^2.
$$
I solve the one-dimensional subproblem on the right-hand side and show that the right-hand side is equal to $\left(\omega(\epsilon_V)-\omega'(\epsilon_V)\epsilon_V\right)^2$.
Let $T(\boldsymbol{Y})=\frac{L(\theta_{\epsilon_V})}{\|\boldsymbol{m}(\theta_{\epsilon_V})\|^2}\boldsymbol{m}(\theta_{\epsilon_V})'\boldsymbol{Y}$.
Note that $L(\theta_{\epsilon_V})=\omega(\epsilon_V)>0$: if $\omega(\epsilon_V)=0$, then $\omega(\epsilon)=0$ for all $\epsilon\ge 0$ since $\omega(\cdot)$ is nonnegative, nondecreasing, and concave, but this contradicts the assumption that $L(\theta)>0$ for some $\theta\in\Theta$ in Assumption \ref{assumption:problem}.
Following the argument used in the proof of Lemma \ref{lemma:one-dim0} in Appendix \ref{proof:lemma:one-dim}, $T(\boldsymbol{Y})$ is shown to be a sufficient statistic for $\theta$ within the subfamily $[-\theta_{\epsilon_V},\theta_{\epsilon_V}]$.
Let ${\cal C}_T$ denote the class of estimators that depend on $\boldsymbol{Y}$ only through $T(\boldsymbol{Y})$.
By a Rao-Blackwell sufficiency argument, one can restrict attention to ${\cal C}_T$. Specifically, for any estimator $\tilde L\in {\cal C}$, let $\tilde L_T(\boldsymbol{Y})=\mathbb{E}[\tilde L(\boldsymbol{Y})|T(\boldsymbol{Y})]$, where the conditional expectation does not depend on $\theta$ since $T(\boldsymbol{Y})$ is sufficient.
Then, it is clear that ${\rm Bias}(\tilde L_T,\theta)={\rm Bias}(\tilde L,\theta)$ and
$\Var(\tilde L_T,\theta)\le \Var(\tilde L,\theta)$ for all $\theta\in [-\theta_{\epsilon_V},\theta_{\epsilon_V}]$.
Therefore,
$$\inf_{\tilde L\in {\cal C}(V)}\sup_{\theta\in[-\theta_{\epsilon_V},\theta_{\epsilon_V}]}{\rm Bias}(\tilde L,\theta)^2=\inf_{\tilde L\in {\cal C}(V)\cap {\cal C}_T}\sup_{\theta\in[-\theta_{\epsilon_V},\theta_{\epsilon_V}]}{\rm Bias}(\tilde L,\theta)^2.
$$
Now, observe that $T(\boldsymbol{Y})\sim {\cal N}\left(L(\theta),\frac{\omega(\epsilon_V)^2\sigma^2}{\epsilon_V^2}\right)$ under $\theta\in [-\theta_{\epsilon_V},\theta_{\epsilon_V}]$.
The minimax problem on the right-hand side is equivalent to the problem of minimizing the maximum squared bias of estimating the mean $L(\theta)\in [-L(\theta_{\epsilon_V}),L(\theta_{\epsilon_V})]$ of the univariate normal sample $T\sim {\cal N}\left(L(\theta),\frac{\omega(\epsilon_V)^2\sigma^2}{\epsilon_V^2}\right)$ subject to a bound on maximum variance.
It follows from Theorem 1 of \citeappendix{low1995tradeoff} that
\begin{align*}
&\inf_{\tilde L\in {\cal C}(V)\cap {\cal C}_T}\sup_{\theta\in[-\theta_{\epsilon_V},\theta_{\epsilon_V}]}{\rm Bias}(\tilde L,\theta)^2=\left(1-\min\left\{\frac{\sqrt{V}}{\omega(\epsilon_V)\sigma/\epsilon_V},1\right\}\right)^2L(\theta_{\epsilon_V})^2\\
&=\left(1-\min\left\{\frac{\omega'(\epsilon_V)\epsilon_V}{\omega(\epsilon_V)},1\right\}\right)^2\omega(\epsilon_V)^2=\left(\omega(\epsilon_V)-\omega'(\epsilon_V)\epsilon_V\right)^2,
\end{align*}
where the last equality holds since $\omega(\epsilon_V)\ge \omega'(\epsilon_V)\epsilon_V$ by the concavity of $\omega(\cdot)$.
\qed

\paragraph{Proof of Statement \ref{theorem:low-unbound}.}
I use the notation introduced in Appendix \ref{section:proof}.
By Lemma \ref{lemma:superdifferential}, there exists a unique solution $\boldsymbol d^*$ to $\min_{\boldsymbol d\in \left.\partial\bar I\right|_{\bar{\cal M}}(\boldsymbol{0})}\|\boldsymbol d\|$, $\boldsymbol{w}^*= \frac{\boldsymbol d^*}{\|\boldsymbol d^*\|}$, and $\omega'(0)=\|\boldsymbol{d}^*\|$.
For any $\theta\in\Theta$,
\begin{align*}
    &L(\theta)\le \bar I(\boldsymbol{m}(\theta))\le \bar I(\boldsymbol{0})+(\boldsymbol{d}^*)'\boldsymbol{m}(\theta)=\omega(0)+\omega'(0)(\boldsymbol{w}^*)'\boldsymbol{m}(\theta),
\end{align*}
where the first inequality holds by the definition of $\bar I(\cdot)$, the second inequality holds since $\boldsymbol{d}^*\in \left.\partial\bar I\right|_{\bar{\cal M}}(\boldsymbol{0})$, and the equality holds since $\omega(0)=\bar I(\boldsymbol{0})$ by definition.
Therefore,
\begin{align*}
    &\inf_{\theta\in\Theta}{\rm Bias}(\hat L_0,\theta)=\inf_{\theta\in\Theta}\left(\omega'(0)(\boldsymbol{w}^*)'\boldsymbol{m}(\theta)-L(\theta)\right)\ge -\omega(0).
\end{align*}
Since ${\rm Bias}(\hat L_0,\theta)$ is linear in $\theta$ and $\Theta$ is centrosymmetric, it follows that $\sup_{\theta\in\Theta}{\rm Bias}(\hat L_0,\theta)\le\omega(0)$.
Hence, $\sup_{\theta\in\Theta}{\rm Bias}(\hat L_0,\theta)^2\le\omega(0)^2$, which implies $\inf_{\tilde L\in {\cal C}}\sup_{\theta\in\Theta}{\rm Bias}(\tilde L,\theta)^2\le \omega(0)^2$.

In the following, I show that $\inf_{\tilde L\in {\cal C}}\sup_{\theta\in\Theta}{\rm Bias}(\tilde L,\theta)^2\ge \omega(0)^2$, which proves Statement \ref{theorem:low-unbound}.
Let $\{\theta_k\}_{k=1}^\infty$ be a sequence in $\Theta$ such that $L(\theta_k)\ge 0$ and $\boldsymbol m(\theta_k)=\boldsymbol{0}$ for all $k$ and that $\lim_{k\rightarrow \infty} L(\theta_k)=\omega(0)$.
For any estimator $\tilde L\in {\cal C}$,
${\rm Bias}(\tilde L,\theta_k)=\mathbb{E}[\tilde L(\boldsymbol{Y})]-L(\theta_k)$ and ${\rm Bias}(\tilde L,-\theta_k)=\mathbb{E}[\tilde L(\boldsymbol{Y})]+L(\theta_k)$, where $\boldsymbol{Y}\sim {\cal N}(\boldsymbol{0},\sigma^2\boldsymbol{I}_n)$.
Therefore, for any $\tilde L\in {\cal C}$,
$$
\sup_{\theta\in \{-\theta_k,\theta_k\}}{\rm Bias}(\tilde L,\theta)^2={\rm Bias}(\tilde L,-\theta_k)^2\mathbf{1}\{\mathbb{E}[\tilde L(\boldsymbol{Y})]\ge 0\}+{\rm Bias}(\tilde L,\theta_k)^2\mathbf{1}\{\mathbb{E}[\tilde L(\boldsymbol{Y})]< 0\}\ge L(\theta_k)^2.
$$
Since $\{-\theta_k,\theta_k\}\subset\Theta$, it follows that $$
\inf_{\tilde L\in {\cal C}}\sup_{\theta\in\Theta}{\rm Bias}(\tilde L,\theta)^2\ge \inf_{\tilde L\in {\cal C}}\sup_{\theta\in\{-\theta_k,\theta_k\}}{\rm Bias}(\tilde L,\theta)^2\ge L(\theta_k)^2.
$$
Taking the limit as $k\rightarrow\infty$ gives $\inf_{\tilde L\in {\cal C}}\sup_{\theta\in\Theta}{\rm Bias}(\tilde L,\theta)^2\ge \omega(0)^2$.
\qed

\subsection{Proof of Theorem \ref{theorem:donoho}}\label{proof:donoho}

For notational brevity, I use $\tilde\epsilon$ to denote $\epsilon_{\rm MSE}$.

\paragraph{Proof of Statement \ref{theorem:donoho-original}.}
The statement follows from the arguments in \cite{donoho1994}, but I provide the proof for completeness.
I first show that $\tilde\epsilon>0$.
Note that $\omega(\epsilon)>0$ for all $\epsilon>0$: if $\omega(\epsilon)=0$ for some $\epsilon>0$, then $\omega(\epsilon')=0$ for all $\epsilon'\ge 0$ since $\omega(\cdot)$ is nonnegative, nondecreasing, and concave, but this contradicts the assumption that $L(\theta)>0$ for some $\theta\in\Theta$ in Assumption \ref{assumption:problem}.
Furthermore, $\lim_{\epsilon\downarrow 0}\omega'(\epsilon)=\omega'(0)$, which follows from the fact that $\omega'(\cdot)$ is nonincreasing by the concavity of $\omega(\cdot)$ and the intermediate value property of derivatives (see, e.g., Theorem 5.12 of \citeappendix{rudin1976}).
Let $g(\epsilon)=\frac{\sigma^2}{\epsilon^2+\sigma^2}\omega(\epsilon)^2$.
By differentiating $g(\epsilon)$, we have
$g'(0)=2\omega(0)\omega'(0)\ge 0$ and $g'(\epsilon)=2\frac{\sigma^2\omega(\epsilon)^2}{\epsilon^2+\sigma^2}\left[\frac{\omega'(\epsilon)}{\omega(\epsilon)}-\frac{\epsilon}{\epsilon^2+\sigma^2}\right]$ for all $\epsilon>0$.
Since $\lim_{\epsilon\downarrow 0}\omega(\epsilon)=\omega(0)\ge 0$ and $\lim_{\epsilon\downarrow 0}\omega'(\epsilon)=\omega'(0)>0$, $\frac{\omega'(\epsilon)}{\omega(\epsilon)}-\frac{\epsilon}{\epsilon^2+\sigma^2}$ either converges to $\frac{\omega'(0)}{\omega(0)}>0$ or diverges to $\infty$.
Therefore, $g'(\epsilon)>0$ for all sufficiently small $\epsilon>0$, so that $g(\epsilon)>g(0)$ for all sufficiently small $\epsilon>0$.
Since $\tilde\epsilon\in \arg\max_{\epsilon\ge 0}g(\epsilon)$, $\tilde\epsilon>0$.

Moreover, it must be the case that $g'(\tilde\epsilon)=0$, which implies that $\frac{\omega'(\tilde \epsilon)}{\omega(\tilde\epsilon)}-\frac{\tilde\epsilon}{\tilde\epsilon^2+\sigma^2}=0$, since $\omega(\tilde\epsilon)>0$.
Hence, $\hat L_{\tilde\epsilon}(\boldsymbol{Y})=\omega'(\tilde\epsilon)\frac{\boldsymbol{m}(\theta_{\tilde\epsilon})'}{\|\boldsymbol{m}(\theta_{\tilde\epsilon})\|}\boldsymbol{Y}=\frac{\omega(\tilde\epsilon)}{\tilde\epsilon^2+\sigma^2}\boldsymbol{m}(\theta_{\tilde\epsilon})'\boldsymbol{Y}$.
To show that $\hat L_{\tilde\epsilon}$ is minimax affine MSE, note first that by Lemma \ref{lemma:max_bias}, $\hat L_{\tilde\epsilon}$ achieves the maximum squared bias at $-\theta_{\tilde\epsilon}$ and $\theta_{\tilde\epsilon}$. Since $\hat L_{\tilde\epsilon}$ has constant variance over $\Theta$, $\hat L_{\tilde\epsilon}$ achieves maximum MSE over $\Theta$ at $-\theta_{\tilde\epsilon}$ and $\theta_{\tilde\epsilon}$.
Hence,
$\sup_{\theta\in\Theta}\mathbb{E}_\theta[(\hat L_{\tilde\epsilon}(\boldsymbol{Y})-L(\theta))^2]=\sup_{\theta\in[-\theta_{\tilde\epsilon},\theta_{\tilde\epsilon}]}\mathbb{E}_\theta[(\hat L_{\tilde\epsilon}(\boldsymbol{Y})-L(\theta))^2]$, which implies $\inf_{\tilde L\in{\cal C}_{\rm affine}}\sup_{\theta\in\Theta}\mathbb{E}_\theta[(\tilde L(\boldsymbol{Y})-L(\theta))^2]\le \sup_{\theta\in[-\theta_{\tilde\epsilon},\theta_{\tilde\epsilon}]}\mathbb{E}_\theta[(\hat L_{\tilde\epsilon}(\boldsymbol{Y})-L(\theta))^2]$.
Furthermore, Lemma 1 of \cite{donoho1994} and its proof show that $\hat L_{\tilde\epsilon}(\boldsymbol{Y})=\frac{\omega(\tilde\epsilon)}{\tilde\epsilon^2+\sigma^2}\boldsymbol{m}(\theta_{\tilde\epsilon})'\boldsymbol{Y}$ is minimax affine MSE for the one-dimensional subfamily $[-\theta_{\tilde\epsilon},\theta_{\tilde\epsilon}]$. Specifically, we have
$\sup_{\theta\in[-\theta_{\tilde\epsilon},\theta_{\tilde\epsilon}]}\mathbb{E}_\theta[(\hat L_{\tilde\epsilon}(\boldsymbol{Y})-L(\theta))^2]=\inf_{\tilde L\in{\cal C}_{\rm affine}}\sup_{\theta\in[-\theta_{\tilde\epsilon},\theta_{\tilde\epsilon}]}\mathbb{E}_\theta[(\tilde L(\boldsymbol{Y})-L(\theta))^2]$,
which implies $\sup_{\theta\in[-\theta_{\tilde\epsilon},\theta_{\tilde\epsilon}]}\mathbb{E}_\theta[(\hat L_{\tilde\epsilon}(\boldsymbol{Y})-L(\theta))^2]\le \inf_{\tilde L\in{\cal C}_{\rm affine}}\sup_{\theta\in\Theta}\mathbb{E}_\theta[(\tilde L(\boldsymbol{Y})-L(\theta))^2]$.
Consequently,
$\inf_{\tilde L\in{\cal C}_{\rm affine}}\sup_{\theta\in\Theta}\mathbb{E}_\theta[(\tilde L(\boldsymbol{Y})-L(\theta))^2]=\sup_{\theta\in[-\theta_{\tilde\epsilon},\theta_{\tilde\epsilon}]}\mathbb{E}_\theta[(\hat L_{\tilde\epsilon}(\boldsymbol{Y})-L(\theta))^2]=\sup_{\theta\in\Theta}\mathbb{E}_\theta[(\hat L_{\tilde\epsilon}(\boldsymbol{Y})-L(\theta))^2]$.
\qed

\paragraph{Proof of Statement \ref{theorem:donoho-comparison}}
The statement immediately holds when $\epsilon^*=0$, so let $\epsilon^*>0$.
As shown in the proof of Statement \ref{theorem:donoho-original}, we have $\frac{\tilde\epsilon}{\tilde\epsilon^2+\sigma^2}=\frac{\omega'(\tilde \epsilon)}{\omega(\tilde\epsilon)}$.
Therefore, to show that $\epsilon^*<\tilde\epsilon$, it suffices to show that $\frac{\epsilon}{\epsilon^2+\sigma^2}\neq \frac{\omega'(\epsilon)}{\omega(\epsilon)}$ for all $\epsilon \in (0,\epsilon^{*}]$.
To prove this, I first show that $\frac{\epsilon}{\epsilon^2+\sigma^2}<\frac{\phi(\epsilon/\sigma)}{\sigma(1-\Phi(\epsilon/\sigma))}$ or, equivalently,
$\frac{\epsilon^2+\sigma^2}{\epsilon}>\sigma\frac{1-\Phi(\epsilon/\sigma)}{\phi(\epsilon/\sigma)}$ for all $\epsilon\in (0,\epsilon^*]$.
Let $g(\epsilon)=\frac{\epsilon^2+\sigma^2}{\epsilon}$. We have $g'(\epsilon)=1-\frac{\sigma^2}{\epsilon^2}$, which is strictly increasing in $\epsilon>0$. Therefore, $g(\epsilon)$ is minimized over $(0,\infty)$ at $\epsilon=\sigma$, at which $g'(\epsilon)=0$.
For any $\epsilon>0$,
$g(\epsilon)\ge g(\sigma)=2\sigma>\sigma\frac{1-\Phi(0)}{\phi(0)}>\sigma\frac{1-\Phi(\epsilon/\sigma)}{\phi(\epsilon/\sigma)}$, where the second last inequality holds since $\frac{1-\Phi(0)}{\phi(0)}\approx 1.253$, and the last holds since $\frac{1-\Phi(\epsilon/\sigma)}{\phi(\epsilon/\sigma)}$ is strictly decreasing in $\epsilon$ by the fact that the Mills ratio $\frac{1-\Phi(x)}{\phi(x)}$ of a standard normal random variable is strictly decreasing.

Next, I show that $\frac{\phi(\epsilon/\sigma)}{\sigma(1-\Phi(\epsilon/\sigma))}\le\frac{\omega'(\epsilon)}{\omega(\epsilon)}$ for all $\epsilon\in (0,\epsilon^*]$.
Let $g(\epsilon)=\omega(\epsilon)\Phi(-\epsilon/\sigma)$.
For $\epsilon> 0$,
	$g'(\epsilon)=\omega'(\epsilon)\Phi\left(-\epsilon/\sigma\right)-\frac{\omega(\epsilon)}{\sigma}\phi\left(-\epsilon/\sigma\right)=\left[\sigma\omega'(\epsilon)\frac{1-\Phi\left(\epsilon/\sigma\right)}{\phi\left(\epsilon/\sigma\right)}-\omega(\epsilon)\right]\frac{\phi\left(\epsilon/\sigma\right)}{\sigma}$.
    Let $h(\epsilon)=\sigma\omega'(\epsilon)\frac{1-\Phi\left(\epsilon/\sigma\right)}{\phi\left(\epsilon/\sigma\right)}-\omega(\epsilon)$.
    Since $\epsilon^*>0$ maximizes $g(\epsilon)$ over $\epsilon\in [0,\epsilon^*]$, we must have $g'(\epsilon^*)\ge 0$ or, equivalently, $h(\epsilon^*)\ge 0$.
    Note that $h(\cdot)$ is nonincreasing, since $\omega'(\cdot)$ is nonincreasing by the concavity of $\omega(\cdot)$, $\frac{1-\Phi(\epsilon/\sigma)}{\phi(\epsilon/\sigma)}$ is strictly decreasing in $\epsilon$, and $\omega(\cdot)$ is nondecreasing.
    Therefore, $h(\epsilon)\ge 0$ for all $\epsilon\in (0,\epsilon^*]$, which implies that $\frac{\phi(\epsilon/\sigma)}{\sigma(1-\Phi(\epsilon/\sigma))}\le\frac{\omega'(\epsilon)}{\omega(\epsilon)}$ for all $\epsilon\in (0,\epsilon^*]$.

    I have shown that $\frac{\epsilon}{\epsilon^2+\sigma^2}<\frac{\phi(\epsilon/\sigma)}{\sigma(1-\Phi(\epsilon/\sigma))}$ and $\frac{\phi(\epsilon/\sigma)}{\sigma(1-\Phi(\epsilon/\sigma))}\le\frac{\omega'(\epsilon)}{\omega(\epsilon)}$ for all $\epsilon\in (0,\epsilon^*]$.
    Therefore, $\frac{\epsilon}{\epsilon^2+\sigma^2}\neq \frac{\omega'(\epsilon)}{\omega(\epsilon)}$ for all $\epsilon \in (0,\epsilon^{*}]$.
    \qed

\section{Additional Details for Section \ref{section:app}}
\renewcommand{\theequation}{C.\arabic{equation}}\label{appendix:app-details}
\setcounter{equation}{0}

\subsection{Derivation of $\omega(0)$, $\omega'(0)$, and $\boldsymbol{w}^*$}\label{appendix:derivation-wstar}

The following proposition can be used to derive $\omega(0)$, $\omega'(0)$, and $\boldsymbol{w}^*$.
To derive $\boldsymbol{w}^*$, let $\boldsymbol{\mu}_\epsilon=(f_\epsilon(x_1,d_1)/\sigma(x_1,d_1),...,f_\epsilon(x_n,d_n)/\sigma(x_n,d_n))'$ for $\epsilon\in [0,\bar\epsilon]$, where $(f_\epsilon(x_i,d_i))_{i=1,...,n}$ and $\bar\epsilon$ are given by the proposition.
Then,
\begin{align*}
\boldsymbol{w^*}=\lim_{\epsilon\downarrow 0}\frac{\boldsymbol \mu_\epsilon}{\epsilon}=\begin{cases}
0~~&\text{ if } x_i<c_1 \text{ or } x_i>x_{+,{\rm min}},\\
-\sigma(x_i,0)/(\tilde n\bar\sigma)~~&\text{ if } c_1\le x_i<c_0,\\
\sigma(x_{+,{\rm min}},1)/\bar\sigma~~&\text{ if } x_i=x_{+,{\rm min}}.
\end{cases}
\end{align*}
\begin{proposition}[Solution to Modulus Problem for Eligibility Cutoff Choice]\label{proposition:lip_mod}
	Suppose that $d_i=\mathbf{1}\{x_i\ge c_0\}$ for all $i=1,...,n$ and that $x_i\neq x_j$ for any $i\neq j$, $i,j=1,...,n$.
	Then, there exists $\bar\epsilon>0$ such that for any $\epsilon\in [0,\bar\epsilon]$, one solution to (\ref{eq:mod_lip_finite}) is given by
	\begin{align*}
		f_\epsilon(x_i,0)&=\begin{cases}
			0~~&\text{ if } x_i< c_1 \text{ or } x_i\ge c_0,\\
			-\sigma^2(x_i,0)\epsilon/(\tilde n\bar\sigma)~~&\text{ if } c_1\le x_i<c_0,
		\end{cases}\\
		f_\epsilon(x_i,1)&=\begin{cases}
			0~~&\text{ if } x_i> x_{+,{\rm min}},\\
			C(x_{+,{\rm min}} - x_i)+\sigma^2(x_{+,{\rm min}},1)\epsilon/\bar\sigma~~&\text{ if } x_i\le x_{+,{\rm min}},
		\end{cases}
	\end{align*}
	and the modulus of continuity is given by
	$\omega(\epsilon)=C\frac{1}{\tilde n}\sum_{i:c_1\le x_i<c_0}[x_{+,{\rm min}} - x_i]+\bar\sigma\epsilon$.
\end{proposition}
\begin{proof}
    Let $\bar\epsilon=(C/\max_i \sigma(x_i,d_i))\min_{i,j:i\neq j}|x_i-x_j|>0$.
Fix $\epsilon\in [0,\bar\epsilon]$ and consider the following problem:
\begin{align}
	&\max_{(f(x_i,0),f(x_i,1))_{i=1,...,n}\in\mathbb{R}^{2n}} \frac{1}{\tilde n}\sum_{i:c_1\le x_i<c_0}[f(x_i,1)-f(x_i,0)]\label{eq:mod_lip_finite_relax}\\
	s.t. & ~~\sum_{i:c_1\le x_i<c_0}\frac{f(x_i,0)^2}{\sigma^2(x_i,0)}+\frac{f(x_{+,{\rm min}},1)^2}{\sigma^2(x_{+,{\rm min}},1)}\le \epsilon^2,\nonumber\\
    &~~f(x_i,1)-f(x_{+,\rm min},1)\le C|x_i-x_{+,\rm min}|, ~~i\in\{j: c_1\le x_j< c_0\}\nonumber.
\end{align}
I first provide a solution to (\ref{eq:mod_lip_finite_relax}) and then show that the solution also solves the problem (\ref{eq:mod_lip_finite}).
Observe that, given a value of $f(x_{+,{\rm min}},1)$, the objective is maximized only when $f(x_i,1)=C(x_{+,{\rm min}}-x_i)+f(x_{+,{\rm min}},1)$ for any $i$ with $c_1\le x_i <c_0$ under the constraints of (\ref{eq:mod_lip_finite_relax}). Also, none of the objective and constraints depends on $f(x_i,0)$ with $x_i<c_1$ or $x_i\ge c_0$ or on $f(x_i,1)$ with $x_i<c_1$ or $x_i>x_{+,{\rm min}}$.
Thus, a solution to (\ref{eq:mod_lip_finite_relax}) is given by $(f_\epsilon(x_i,0),f_\epsilon(x_i,1))_{i=1,...,n}$, where $f_\epsilon(x_i,0)=0$ for any $i$ with $x_i<c_1$ or $x_i\ge c_0$, $f_\epsilon(x_i,1)=C(x_{+,{\rm min}}-x_i)+f(x_{+,{\rm min}},1)$ for any $i$ with $x_i <c_0$, $f_\epsilon(x_i,1)=0$ for any $i$ with $x_i>x_{+,{\rm min}}$, and $((f_\epsilon(x_i,0))_{i:c_1\le x_i<c_0},f_\epsilon(x_{+,{\rm min}},1))$ solves
\begin{align*}
	&\max_{\boldsymbol f \in\mathbb{R}^{\tilde n+1}} \frac{C}{\tilde n}\sum_{i:c_1\le x_i<c_0}(x_{+,{\rm min}}-x_i)+f(x_{+,{\rm min}},1)-\frac{1}{\tilde n}\sum_{i:c_1\le x_i<c_0}f(x_i,0) \\
	s.t. &~~\sum_{i:c_1\le x_i<c_0}\frac{f(x_i,0)^2}{\sigma^2(x_i,0)}+\frac{f(x_{+,{\rm min}},1)^2}{\sigma^2(x_{+,{\rm min}},1)}\le \epsilon^2,
\end{align*}
where $\boldsymbol f$ denotes $((f(x_i,0))_{i:c_1\le x_i<c_0},f(x_{+,{\rm min}},1))$. 
This is a convex optimization problem that maximizes a weighted sum of $\tilde n+1$ unknowns under the constraint on the upper bound on a weighted Euclidean norm of the unknowns.
Simple calculations show that the solution is given by $f_\epsilon(x_i,0)=-\sigma^2(x_i,0)\epsilon/(\tilde n\bar\sigma)$ for any $i$ with $c_1\le x_i<c_0$ and $f_\epsilon(x_{+,{\rm min}},1)=\sigma^2(x_{+,{\rm min}},1)\epsilon/\bar\sigma$.
Note that  $(f_\epsilon(x_i,0),f_\epsilon(x_i,1))_{i=1,...,n}$ equals the expression in the statement of Proposition \ref{proposition:lip_mod}.

The constraints of (\ref{eq:mod_lip_finite_relax}) are weaker than those of (\ref{eq:mod_lip_finite}).
To show that $(f_\epsilon(x_i,0),f_\epsilon(x_i,1))_{i=1,...,n}$ also solves (\ref{eq:mod_lip_finite}), it suffices to show that $(f_\epsilon(x_i,0),f_\epsilon(x_i,1))_{i=1,...,n}$ satisfies the constraints of (\ref{eq:mod_lip_finite}).
Clearly, the norm constraint is satisfied.
To see that the Lipschitz constraint is satisfied for $d=0$, 
observe that for any $i,j$ with $i\neq j$, 
\begin{align*}
	&|f_\epsilon(x_i,0)-f_\epsilon(x_j,0)|^2=f_\epsilon(x_i,0)^2+f_\epsilon(x_j,0)^2-2f_\epsilon(x_i,0)f_\epsilon(x_j,0)\\
    &\le f_\epsilon(x_i,0)^2+f_\epsilon(x_j,0)^2\le (\max_{k} \sigma(x_k,d_k))^2\bar \epsilon^2=C^2\min_{k,l:k\neq l}|x_k-x_l|^2\le C^2|x_i-x_j|^2,
\end{align*}
where the first inequality holds since $f_\epsilon(x_i,0)\le 0$ for all $i$, the second inequality follows from the norm constraint of (\ref{eq:mod_lip_finite_relax}) and the fact that $f_\epsilon(x_i,0)=0$ for all $i$ with $x_i<c_1$ or $x_i\ge c_0$, and the equality in the second line from the definition of $\bar\epsilon$.
For $d=1$, by construction,
$|f_\epsilon(x_i,1)-f_\epsilon(x_j,1)|=C|x_i-x_j|$ for any $i,j$ with $x_i,x_j\le x_{+,{\rm min}}$, and $|f_\epsilon(x_i,1)-f_\epsilon(x_j,1)|=0$ for any $i,j$ with $x_i,x_j >x_{+,{\rm min}}$.
Note that by the norm constraint of (\ref{eq:mod_lip_finite_relax}),
$f_\epsilon(x_{+,{\rm min}},1)^2\le (\max_{k} \sigma(x_k,d_k))^2\bar \epsilon^2=C^2\min_{k,l:k\neq l}|x_k-x_l|^2$.
It follows that, for any $i,j$ with $x_i\le x_{+,{\rm min}}$ and $x_j>x_{+,{\rm min}}$,
$|f_\epsilon(x_i,1)-f_\epsilon(x_j,1)|= C(x_{+,{\rm min}} - x_i)+f_\epsilon(x_{+,{\rm min}},1)\le C(x_{+,{\rm min}}-x_i)+C\min_{k,l:k\neq l}|x_k-x_l|\le C(x_{+,{\rm min}}-x_i)+C(x_j-x_{+,{\rm min}})= C(x_j-x_i)$.

Finally, the expression for $\omega(\epsilon)$ follows from a simple calculation.
\end{proof}

\subsection{Procedure for Computing $\epsilon^*\in\arg\max_{\epsilon\in[0,\tau^*]}\omega(\epsilon)\Phi(-\epsilon)$}\label{appendix:computation}

Here, I provide a procedure for computing $\epsilon^*\in\arg\max_{\epsilon\in[0,\tau^*]}\omega(\epsilon)\Phi(-\epsilon)$ for the case in which $2\phi(0)\omega(0)<\omega'(0)$.
The procedure is based on the first-order condition.
By differentiating $\omega(\epsilon)\Phi(-\epsilon)$, we have
$\omega'(\epsilon)\Phi\left(-\epsilon\right)-\omega(\epsilon)\phi\left(-\epsilon\right)=\left[\frac{1-\Phi\left(\epsilon\right)}{\phi\left(\epsilon\right)}-\frac{\omega(\epsilon)}{\omega'(\epsilon)}\right]\omega'(\epsilon)\phi\left(\epsilon\right)$.
$\frac{1-\Phi(\epsilon)}{\phi(\epsilon)}$ is the Mills ratio of a standard normal random variable, which is strictly decreasing in $\epsilon$.
Since $\omega(\epsilon)$ is nonnegative, nondecreasing, and concave, $\frac{\omega(\epsilon)}{\omega'(\epsilon)}$ is nondecreasing in $\epsilon$.
Therefore, $\frac{1-\Phi\left(\epsilon\right)}{\phi\left(\epsilon\right)}-\frac{\omega(\epsilon)}{\omega'(\epsilon)}$ is strictly decreasing in $\epsilon$.
This suggests the following procedure to compute $\epsilon^*$.
\begin{enumerate}
	\item If $\frac{1-\Phi\left(\tau^*\right)}{\phi\left(\tau^*\right)}-\frac{\omega(\tau^*)}{\omega'(\tau^*)}\ge 0$, $\epsilon^*=\tau^*$.
	\item If not, use the bisection method to find $\epsilon^*\in[0,\tau^*]$ that solves $\frac{1-\Phi\left(\epsilon\right)}{\phi\left(\epsilon\right)}-\frac{\omega(\epsilon)}{\omega'(\epsilon)}=0$.
\end{enumerate}
For each $\epsilon>0$, once we solve the problem (\ref{eq:mod_lip_finite}) to compute $\omega(\epsilon)$ and $(f_\epsilon(x_i,0),f_\epsilon(x_i,1))_{i=1,...,n}$, we can compute $\omega'(\epsilon)$ using the closed-form expression
$\omega'(\epsilon)=
\frac{\epsilon}{\sum_{i=1}^n d_i f_\epsilon(x_i,d_i)/\sigma^2(x_i,d_i)}$,
which is obtained by an application of Lemma \ref{lemma:dif_omega} with $\iota(x,d)=d$ for all $x\in\mathbb{R}$.

\section{Asymptotic Optimality When the Error Distribution is Unknown}\label{appendix:asymptotics}

\renewcommand{\theequation}{D.\arabic{equation}}
\setcounter{equation}{0}

In this section, I propose a feasible version of the minimax regret decision rule when the error distribution is unknown and establish its asymptotic optimality, closely following the asymptotic framework considered by \citeappendix{armstrong2018optimal}.

Suppose that the sample $\boldsymbol Y_n=(Y_1,...,Y_n)'$ is generated by the following model:
$$
\boldsymbol Y_n =\boldsymbol m_n(\theta)+\boldsymbol U_n,~~~\theta\in \Theta, ~\boldsymbol U_n\sim Q\in {\cal Q}_n,
$$
where $\boldsymbol U_n=(U_1,...,U_n)'$ and ${\cal Q}_n$ is the set of possible distributions of $\boldsymbol U_n$.
As in the main text, $\Theta$ is a known subset of a vector space $\mathbb{V}$, $\boldsymbol m_n: \mathbb{V}\rightarrow \mathbb{R}^n$ is a known linear function, and consider a binary policy choice problem in which $L(\theta)$ represents the welfare contrast, where $L: \mathbb{V}\rightarrow \mathbb{R}$ is a known linear function.
For notational simplicity, I hold $\Theta$ and $L$ fixed over $n$ and assume that the variance of $\boldsymbol{U}_n$ is $\sigma^2\boldsymbol{I}_n$ under any $Q\in{\cal Q}_n$ for some fixed, unknown $\sigma>0$.
It is straightforward to extend the result to the case in which $\Theta$ and $L$ may change as $n$ increases and the variance of $\boldsymbol{U}_n$ may vary across $Q\in{\cal Q}_n$.

Regret of decision rule $\delta_n:\mathbb{R}^n\rightarrow [0,1]$ now depends on $\theta$ and $Q$:
	$$
	R_n(\delta_n,\theta,Q)\coloneqq\begin{cases}
	L(\theta)(1-\mathbb{E}_{\theta,Q}[\delta_n(\boldsymbol{Y}_n)])~~&\text{ if } L(\theta)\ge 0,\\
	-L(\theta)\mathbb{E}_{\theta,Q}[\delta_n(\boldsymbol{Y}_n)]~~&\text{ if } L(\theta)< 0,
	\end{cases}
	$$
	where $\mathbb{E}_{\theta,Q}$ denotes the expectation under $(\theta,Q)$.
	Let ${\cal R}_n(\Theta,{\cal Q}_n)$ denote the minimax risk:
	$$
	{\cal R}_n(\Theta,{\cal Q}_n)\coloneqq\inf_{\delta_n\in{\cal D}_n}\sup_{\theta\in\Theta, Q\in{\cal Q}_n}R_n(\delta_n,\theta,Q),
	$$
    where ${\cal D}_n$ is the set of all decision rules.
I say that a sequence of decision rules $\delta_n$ is {\it asymptotically minimax regret} if
	$$
	\lim_{n\rightarrow \infty}\left(\sup_{\theta\in\Theta, Q\in{\cal Q}_n}R_n(\delta_n,\theta,Q)-{\cal R}_n(\Theta,{\cal Q}_n)\right)=0.
	$$

I consider a feasible version of the decision rule in Theorem \ref{theorem:main}.
Let ${\cal M}_n\coloneqq\{\boldsymbol{m}_n(\theta):\theta\in\Theta\}$, $I_n(\boldsymbol{\mu})\coloneqq\{L(\theta):\boldsymbol{m}_n(\theta)=\boldsymbol{\mu},\theta\in\Theta\}$ for $\boldsymbol{\mu}\in\mathbb{R}^n$, and $\omega_n(\epsilon)\coloneqq \sup\{L(\theta): \|\boldsymbol{m}_n(\theta)\|\le \epsilon,\theta\in\Theta\}$ for $\epsilon\ge 0$.
Also, let $\epsilon^*_n\in\arg\max_{\epsilon\in [0,\tau^*\sigma]}\omega_n(\epsilon)\Phi(-\epsilon/\sigma)$
and $\boldsymbol w^*_n=\lim_{\epsilon\downarrow 0}\frac{\boldsymbol{\mu}_{n,\epsilon}}{\epsilon}$, where $\boldsymbol{\mu}_{n,\epsilon}\in \arg\max_{\boldsymbol{\mu}\in{\cal M}_n:\|\boldsymbol{\mu}\|\le \epsilon}\sup I_n(\boldsymbol{\mu})$.
Suppose we have some estimator $\hat\sigma_n\coloneqq\hat\sigma_n(\boldsymbol Y_n)$ of $\sigma$, which is a function of the sample $\boldsymbol Y_n$.
Let $\hat \epsilon_n\in\arg\max_{\epsilon\in [0,\tau^*\hat\sigma_n]}\omega_n(\epsilon)\Phi(-\epsilon/\hat\sigma_n)$ be an estimator for $\epsilon_n^*$.
Assuming that there exists $\theta_{n,\hat \epsilon_n}$ such that $L(\theta_{n,\hat \epsilon_n})=\omega(\hat \epsilon_n)$ and $\|\boldsymbol{m}_n(\theta_{n,\hat \epsilon_n})\|\le \hat \epsilon_n$,
consider the minimax regret decision rule under normal errors in Theorem \ref{theorem:main}, where $\sigma$ is replaced with $\hat\sigma_n$:
 	\begin{align*}
	\hat\delta_n(\boldsymbol Y_n)=\begin{cases}\mathbf{1}\left\{\boldsymbol{m}_n(\theta_{n,\hat \epsilon_n})'\boldsymbol{Y}_n\ge 0\right\} & ~~\text{if } \omega_n'(0)>0, 2\phi(0)\frac{\omega_n(0)}{\omega_n'(0)} < \hat\sigma_n,\\
	\mathbf{1}\left\{(\boldsymbol{w}_n^*)'\boldsymbol{Y}_n\ge 0\right\} & ~~\text{if } \omega_n'(0)>0, 2\phi(0)\frac{\omega_n(0)}{\omega_n'(0)} = \hat\sigma_n,\\
	\Phi\left(\dfrac{(\boldsymbol{w}_n^*)'\boldsymbol{Y}_n}{((2\phi(0)\omega_n(0)/\omega_n'(0))^2-\hat\sigma_n^2)^{1/2}}\right)& ~~\text{if } \omega_n'(0)>0, 2\phi(0)\frac{\omega_n(0)}{\omega_n'(0)} > \hat\sigma_n,\\
        1/2 & ~~\text{if } \omega_n'(0)=0.
	\end{cases}
	\end{align*}


I show that $\hat\delta_n$ is asymptotically minimax regret under suitable conditions.
For a sequence of random variables $W_n$, I use the notation $W_n\underset{\Theta,{\cal Q}_n}{\stackrel{p}{\rightarrow}}w$, where $w$ is a constant, to denote the convergence in probability uniformly over $\theta\in \Theta$ and $Q\in{\cal Q}_n$---that is, for every $c>0$,
$\lim_{n\rightarrow\infty}\sup_{\theta\in\Theta,Q\in{\cal Q}_n}\mathbb{P}_{\theta,Q}(|W_n-w|>c)=0$.
I use the notation $W_n\underset{\Theta,{\cal Q}_n}{\stackrel{d}{\rightarrow}} {\cal L}$, where ${\cal L}$ is a probability law, to denote the convergence in distribution uniformly over $\theta\in \Theta$ and $Q\in{\cal Q}_n$---that is,
$\lim_{n\rightarrow \infty}\sup_{\theta\in\Theta,Q\in{\cal Q}_n}\sup_{w\in\mathbb{R}}|\mathbb{P}_{\theta,Q}(W_n\le w)-F_{\cal L}(w)|=0$,
where $F_{\cal L}$ is the cumulative distribution function of the law ${\cal L}$, which is assumed to be continuous.

\begin{assumption}[Conditions for Asymptotic Optimality]\label{assumption:asymptotics}
~
\begin{enumerate}[label=(\roman*)]
        \item $\omega_n'(0)>0$ for any sufficiently large $n$, and $s^*\coloneqq\lim_{n\rightarrow\infty}2\phi(0)\frac{\omega_n(0)}{\omega'_n(0)}$ exists.
        \item  For any sufficiently large $n$, there exists $\theta_{n,\epsilon_n^*}$ that attains the modulus of continuity at $\epsilon_n^*$.
        \item $\hat\sigma_n\underset{\Theta,{\cal Q}_n}{\stackrel{p}{\rightarrow}}\sigma$.
        \item For any sufficiently large $n$, the following holds with probability one under any $(\theta,Q)\in\Theta\times {\cal Q}_n$: there exists $\theta_{n,\hat \epsilon_n}$ that attains the modulus of continuity at $\hat \epsilon_n$.
		\item\label{assumption:uniform}
        Either of the following holds.
        \begin{enumerate}[label=(\alph*)]
            \item\label{assumption:asymptotics-det} $\sigma > s^*$ and $\frac{\boldsymbol m_n(\theta_{n,\epsilon^*_n})'\boldsymbol U_n+(\boldsymbol m_n(\theta_{n,\hat \epsilon_n})-\boldsymbol m_n(\theta_{n,\epsilon^*_n}))\boldsymbol Y_n}{\sigma\epsilon^*_n}\underset{\Theta,{\cal Q}_n}{\stackrel{d}{\rightarrow}} {\cal N}(0,1)$.
            \item\label{assumption:asymptotics-ran} $\sigma < s^*$ and $\frac{(\boldsymbol{w}^*_n)'\boldsymbol{U}_n +((s^*_n)^2-\hat\sigma^2_n)^{1/2}\xi}{s^*_n} \underset{\Theta,{\cal Q}_n}{\stackrel{d}{\rightarrow}} {\cal N}(0,1)$, where $s^*_n\coloneqq 2\phi(0)\frac{\omega_n(0)}{\omega'_n(0)}$ and $\xi|\boldsymbol Y_n\sim {\cal N}(0,1)$ under any $(\theta,Q)\in\Theta\times {\cal Q}_n$ for all $n$.
        \end{enumerate}
		\item \label{assumption:asymptotics-normal} For any sufficiently large $n$, there exists $Q\in {\cal Q}_n$ such that $\boldsymbol U_n\sim {\cal N}(\boldsymbol 0, \sigma^2 \boldsymbol I_n)$ under $Q$.
        \item $\sup_{\theta\in\Theta}|L(\theta)|<\infty$.
	\end{enumerate}
\end{assumption}

The key conditions are Assumption \ref{assumption:asymptotics}\ref{assumption:uniform} and \ref{assumption:asymptotics-normal}.
The former assumes the uniform asymptotic normality of the weighted sum of the errors $\boldsymbol U_n$ (plus an error due to the estimation of $\sigma$ when $\hat\delta_n$ is nonrandomized and a mean-zero noise $((s^*_n)^2-\hat\sigma^2_n)^{1/2}\xi$ when $\hat\delta_n$ is randomized).
This assumption is used to show that the maximum regret of $\hat\delta_n$ over $\theta\in\Theta$ and $Q\in{\cal Q}_n$ converges to the minimax risk under normal errors in Theorem \ref{theorem:main}.
The uniform asymptotic normality can be established, for example, under conditions on the weight vector $\frac{\boldsymbol m_n(\theta_{n,\epsilon^*_n})}{\sigma \epsilon^*_n}$ or $\frac{\boldsymbol{w}^*_n}{s^*_n}$ \`a la the Lindeberg condition for the central limit theorem.
In the context of uniform asymptotic valid inference on RD parameters, such conditions are directly imposed \citepappendix{Imbens2019RDD} or verified under low-level conditions \citepappendix{armstrong2018optimal}.
Assumption \ref{assumption:asymptotics}\ref{assumption:asymptotics-normal} states that the class of possible distributions contains a normal distribution.
Under this assumption, the minimax risk under normal errors is a lower bound on the minimax risk over the class of possible distributions.
Since the maximum regret of $\hat\delta_n$ is shown to converge to the lower bound, $\hat\delta_n$ is asymptotically minimax regret.
Similar assumptions are often used in the literature on minimax estimation and inference in nonparametric regression models \citepappendix{fan1993local,armstrong2018optimal}.


\begin{theorem}[Asymptotic Optimality]\label{theorem:asymptotics}
	Suppose that $(L,\boldsymbol{m}_n,\Theta)$ satisfies Assumption \ref{assumption:problem} for any sufficiently large $n$ and that Assumption \ref{assumption:asymptotics} holds.
	Then, $\hat\delta_n$ is asymptotically minimax regret.
\end{theorem}

\begin{proof}
    Let $Q_{\cal N}$ denote ${\cal N}(\boldsymbol{0},\sigma^2\boldsymbol{I}_n)$ and ${\cal R}_n(\Theta,{\cal N})=\inf_{\delta_n\in{\cal D}_n}\sup_{\theta\in\Theta}R_n(\delta_n,\theta,Q_{\cal N})$ denote the minimax risk under normal errors.
    Since $Q_{\cal N}\in {\cal Q}_n$,
    \begin{align}
    {\cal R}_n(\Theta,{\cal N})\le \inf_{\delta_n\in{\cal D}_n}\sup_{\theta\in\Theta, Q\in{\cal Q}_n}R_n(\delta_n,\theta,Q)={\cal R}_n(\Theta,{\cal Q}_n).\label{eq:asy-pf2}
    \end{align}
    Below, I show that
    \begin{align}
			\sup_{\theta\in\Theta,Q\in{\cal Q}_n}R_n(\hat\delta_n,\theta,Q)\le {\cal R}_n(\Theta,{\cal N})+ c_n \label{eq:asy-pf1}
	\end{align}
	for some positive constant $c_n\rightarrow 0$.
    Once it is done, combining inequalities (\ref{eq:asy-pf2}) and (\ref{eq:asy-pf1}) yields
    $\sup_{\theta\in\Theta,Q\in{\cal Q}_n}R_n(\hat\delta_n,\theta,Q)-{\cal R}_n(\Theta,{\cal Q}_n)\le c_n$, which implies that $\lim_{n\rightarrow \infty}\left(\sup_{\theta\in\Theta, Q\in{\cal Q}_n}R_n(\hat\delta_n,\theta,Q)-{\cal R}_n(\Theta,{\cal Q}_n)\right)=0$.

    In what follows, I show that (\ref{eq:asy-pf1}) holds for the case in which $\sigma>s^*=\lim_{n\rightarrow\infty}s^*_n$.
    The proof for the case in which $\sigma<s^*$ is omitted, since it is similar, noting that we can write $\hat\delta_n(\boldsymbol Y_n)=\mathbb{P}_{\theta,Q}\left((\boldsymbol{w}^*_n)'\boldsymbol{Y}_n+((s^*_n)^2-\hat\sigma^2_n)^{1/2}\xi\ge 0|\boldsymbol Y_n\right)$ when $\hat\sigma_n< s^*_n$.
    Suppose $\sigma>s^*$.
    Let $\tilde\delta_n(\boldsymbol Y_n)=\mathbf{1}\left\{\boldsymbol{m}_n(\theta_{n,\hat \epsilon_n})'\boldsymbol{Y}_n\ge 0\right\}$.
    By construction,
    $\sup_{\theta\in\Theta,Q\in{\cal Q}_n}\mathbb{P}_{\theta,Q}\left(\hat\delta_n(\boldsymbol{Y}_n)\neq \tilde\delta_n(\boldsymbol{Y}_n)\right)\le\sup_{\theta\in\Theta,Q\in{\cal Q}_n}\mathbb{P}_{\theta,Q}\left(\hat\sigma_n\le s^*_n\right)$.
    The right-hand side converges to zero, since $\hat\sigma_n\underset{\Theta,{\cal Q}_n}{\stackrel{p}{\rightarrow}}\sigma>s^*=\lim_{n\rightarrow\infty}s^*_n$.
    Consider the following inequality:
    \begin{align*}
        &\sup_{\theta\in\Theta,Q\in{\cal Q}_n}L(\theta)(1-\mathbb{E}_{\theta,Q}[\hat\delta_n(\boldsymbol{Y}_n)])\\
        &\le\sup_{\theta\in\Theta,Q\in{\cal Q}_n}L(\theta)(1-\mathbb{E}_{\theta,Q}[\tilde\delta_n(\boldsymbol Y_n)])+\left(\sup_{\theta\in\Theta}|L(\theta)|\right)\sup_{\theta\in\Theta,Q\in{\cal Q}_n}\mathbb{E}_{\theta,Q}\left[|\hat\delta_n(\boldsymbol{Y}_n)-\tilde\delta_n(\boldsymbol Y_n)|\right].
    \end{align*}
    The second term on the right-hand side is bounded above by $\left(\sup_{\theta\in\Theta}|L(\theta)|\right)\sup_{\theta\in\Theta,Q\in{\cal Q}_n}\mathbb{E}_{\theta,Q}\left[\mathbf{1}\left\{\hat\delta_n(\boldsymbol Y_n)\neq\tilde\delta_n(\boldsymbol Y_n)\right\}\right]$,
    which converges to zero as $n\rightarrow\infty$, since $\sup_{\theta\in\Theta}|L(\theta)|<\infty$ and $\sup_{\theta\in\Theta,Q\in{\cal Q}_n}\mathbb{P}_{\theta,Q}\left(\hat\delta_n(\boldsymbol{Y}_n)\neq \tilde\delta_n(\boldsymbol{Y}_n)\right)\rightarrow 0$.
    It follows that for some positive constant $c_n\rightarrow 0$,
    \begin{align}
        \sup_{\theta\in\Theta,Q\in{\cal Q}_n}L(\theta)(1-\mathbb{E}_{\theta,Q}[\hat\delta_n(\boldsymbol{Y}_n)])\le\sup_{\theta\in\Theta,Q\in{\cal Q}_n}L(\theta)(1-\mathbb{E}_{\theta,Q}[\tilde\delta_n(\boldsymbol Y_n)])+c_n. \label{eq:asy-pf30}
    \end{align}

    Next, we bound $\sup_{\theta\in\Theta,Q\in{\cal Q}_n}L(\theta)(1-\mathbb{E}_{\theta,Q}[\tilde\delta_n(\boldsymbol Y_n)])$.
    Since $\sigma>s^*=\lim_{n\rightarrow\infty}s^*_n$, for any sufficiently large $n$, $\sigma>s^*_n$, and hence $\epsilon^*_n>0$ by Lemma \ref{lemma:hardest-1d}.
    Let $Z_n=\frac{\boldsymbol m_n(\theta_{n,\epsilon^*_n})'\boldsymbol U_n+(\boldsymbol m_n(\theta_{n,\hat \epsilon_n})-\boldsymbol m_n(\theta_{n,\epsilon^*_n}))\boldsymbol Y_n}{\sigma\epsilon^*_n}$. We have
	 \begin{align*}
	 &\sup_{\theta\in\Theta,Q\in{\cal Q}_n}L(\theta)(1-\mathbb{E}_{\theta,Q}[\tilde\delta_n(\boldsymbol{Y}_n)])=\sup_{\theta\in\Theta,Q\in{\cal Q}_n}L(\theta)\left(1-\mathbb{P}_{\theta,Q}(\boldsymbol{m}_n(\theta_{n,\hat \epsilon_n})'\boldsymbol{Y}_n\ge 0)\right)\\
    &=\sup_{\theta\in\Theta,Q\in{\cal Q}_n}L(\theta)\left(1-\mathbb{P}_{\theta,Q}\left((\boldsymbol m_n(\theta_{n,\hat \epsilon_n})-\boldsymbol m_n(\theta_{n,\epsilon^*_n}))\boldsymbol Y_n+\boldsymbol m_n(\theta_{n,\epsilon^*_n})'\boldsymbol{Y}_n\ge 0\right)\right)\\
    &=\sup_{\theta\in\Theta,Q\in{\cal Q}_n}L(\theta)\left(1-\mathbb{P}_{\theta,Q}\left(-Z_n\le \frac{\boldsymbol m_n(\theta_{n,\epsilon^*_n})'\boldsymbol m_n(\theta)}{\sigma\epsilon^*_n}\right)\right)\\
    &\le\sup_{\theta\in\Theta}L(\theta)\left(1-\Phi\left(\frac{\boldsymbol m_n(\theta_{n,\epsilon^*_n})'\boldsymbol m_n(\theta)}{\sigma\epsilon^*_n}\right)\right)\\
	 &~~~~+\left(\sup_{\theta\in\Theta}|L(\theta)|\right)\sup_{\theta\in\Theta,Q\in{\cal Q}_n}\left|\mathbb{P}_{\theta,Q}\left(-Z_n\le \frac{\boldsymbol m_n(\theta_{n,\epsilon^*_n})'\boldsymbol m_n(\theta)}{\sigma\epsilon^*_n}\right)-\Phi\left(\frac{\boldsymbol m_n(\theta_{n,\epsilon^*_n})'\boldsymbol m_n(\theta)}{\sigma\epsilon^*_n}\right)\right|.
	 \end{align*}
     The second term in the last expression is bounded above by
     $\left(\sup_{\theta\in\Theta}|L(\theta)|\right)\sup_{\theta\in\Theta,Q\in{\cal Q}_n}\sup_{z\in\mathbb{R}}\left|\mathbb{P}_{\theta,Q}\left(-Z_n\le z\right)-\Phi\left(z\right)\right|$,
     which converges to zero as $n\rightarrow \infty$, since $\sup_{\theta\in\Theta}|L(\theta)|<\infty$ and $Z_n\underset{\Theta,{\cal Q}_n}{\stackrel{d}{\rightarrow}} {\cal N}(0,1)$.
     Note that, when $\sigma>s^*_n$, the first term $\sup_{\theta\in\Theta}L(\theta)\left(1-\Phi\left(\frac{\boldsymbol m_n(\theta_{n,\epsilon^*_n})'\boldsymbol m_n(\theta)}{\sigma\epsilon^*_n}\right)\right)$ is the maximum regret of decision rule $\delta^*_n(\boldsymbol Y_n)=\mathbf{1}\left\{\boldsymbol{m}_n(\theta_{n,\epsilon^*_n})'\boldsymbol{Y}_n\ge 0\right\}$ under normal errors. Since $\delta^*_n$ is minimax regret under normal errors,
     $\sup_{\theta\in\Theta}L(\theta)\left(1-\Phi\left(\frac{\boldsymbol m_n(\theta_{n,\epsilon^*_n})'\boldsymbol m_n(\theta)}{\sigma\epsilon^*_n}\right)\right)={\cal R}_n(\Theta,{\cal N})$.
     Together with (\ref{eq:asy-pf30}), the above argument implies that for some positive  $c_n\rightarrow 0$, $\sup_{\theta\in\Theta,Q\in{\cal Q}_n}L(\theta)(1-\mathbb{E}_{\theta,Q}[\hat\delta_n(\boldsymbol{Y}_n)])\le {\cal R}_n(\Theta,{\cal N})+ c_n$.
    The same argument applies to $\sup_{\theta\in\Theta,Q\in{\cal Q}_n}-L(\theta)\mathbb{E}_{\theta,Q}[\hat\delta_n(\boldsymbol{Y}_n)]$, so that inequality (\ref{eq:asy-pf1}) holds as desired.
\end{proof}

\section{Empirical Application: Additional Tables and Figures}\label{appendix:empirical_result}

\setcounter{table}{0}
\renewcommand{\thetable}{E.\arabic{table}}

\begin{table}[!ht]\centering
	\newcolumntype{C}{>{\centering\arraybackslash}X}
	
	\caption{Child Educational Outcomes and Characteristics}
	{\footnotesize
		\begin{tabularx}{\textwidth}{lCCC}
			
			\hline\hline
			& All    & Eligible & Ineligible                    \\
			&        & villages & villages                       \\
			& (1) & (2) & (3) \\
			\hline\addlinespace[1ex]
			\multicolumn{4}{c}{Panel A. Educational outcomes (child-level means)}                   \\\addlinespace[1ex]
			Enrollment                   & 0.366  & 0.494    & 0.259                          \\
			Normalized total test scores & 0.000  & 0.248    & $-0.209$                         \\
			Highest grade child has achieved              & 0.876  & 1.132    & 0.636                          \\
			&&&\\
			\multicolumn{4}{c}{Panel B. Child and household characteristics (child-level means)}    \\\addlinespace[1ex]
			Child's age                  & 8.121  & 8.174    & 8.071                          \\
			Child is female              & 0.503  & 0.476    & 0.525                          \\
			Head's age                   & 47.653 & 47.387   & 47.904                         \\
			Head years of schooling      & 0.156  & 0.198    & 0.117                          \\
			Number of members            & 10.812 & 10.815   & 10.808                         \\
			Number of children           & 5.971  & 6.098    & 5.850                          \\
			Muslim                       & 0.587  & 0.576    & 0.597                          \\
			Basic roofing                & 0.516  & 0.534    & 0.500                          \\
			Number of motorbikes         & 0.299  & 0.319    & 0.279                          \\
			Number of phones         & 0.185  & 0.199    & 0.172                          \\
			&        &          &                                \\
			Total number of children     & 23,282  & 10,645    & 12,637                          \\
			Total number of villages     & 287    & 136      & 151\\
			\bottomrule                          
		\end{tabularx}
	}
	\caption*{
		\scriptsize {\it Notes}: This table reports child-level averages of educational outcomes and characteristics by program eligibility in the year 2008, namely 2.5 years after the start of the BRIGHT program.
		Panel A reports the educational outcomes' means.
		Panel B reports the means of child and household characteristics.
		Column (1) shows the means for children in all villages. Columns (2) and (3) show the means for children in villages selected for BRIGHT school and in unselected villages, respectively.
	}
	\label{table:summary}
\end{table}

\setcounter{figure}{0}
\renewcommand{\thefigure}{E.\arabic{figure}}

	\begin{figure}[!h]
		\centering
		\caption{Weight to Each Village Attached by Plug-in Rules} 
		\begin{subfigure}[b]{0.49\textwidth}   
			\centering 
			\includegraphics[width=\textwidth]{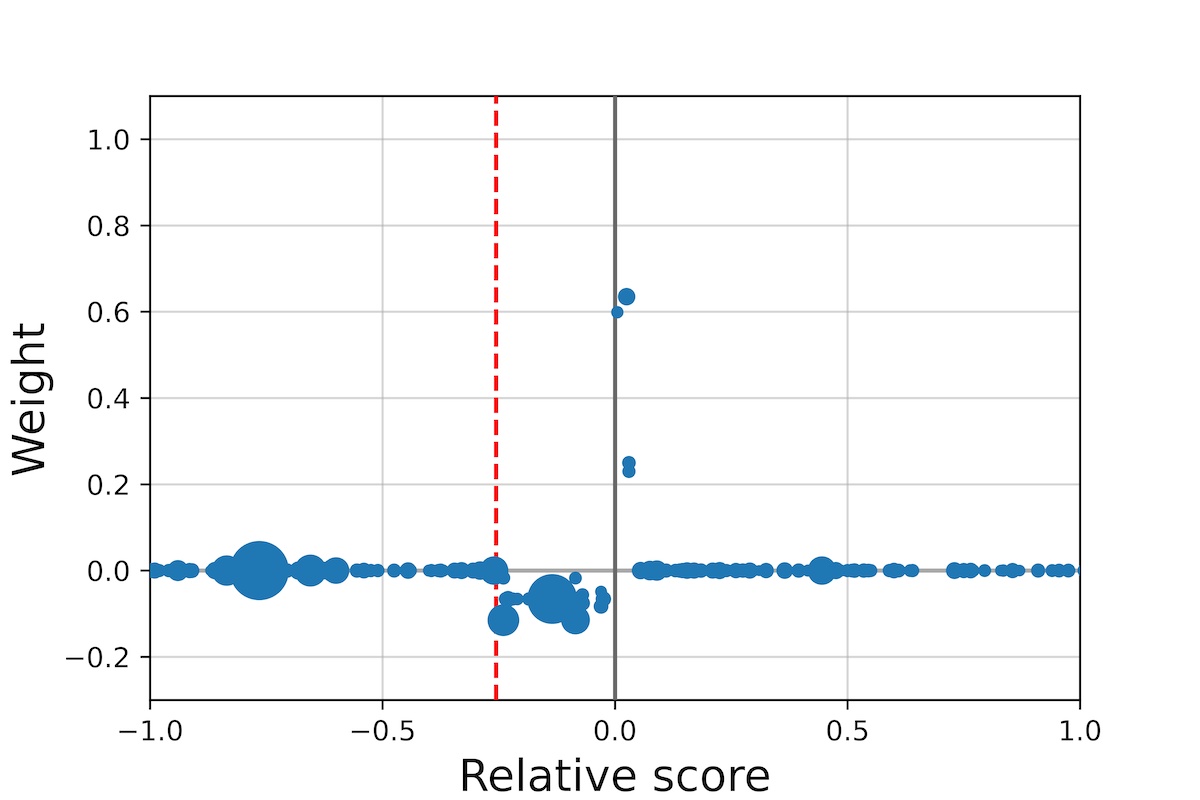}
			\vspace{-0.8em}
			\caption[]%
			{\footnotesize Minimax Affine MSE, $C=0.5$\\
            ~~~}
		\end{subfigure}
		\begin{subfigure}[b]{0.49\textwidth}   
			\centering 
			\includegraphics[width=\textwidth]{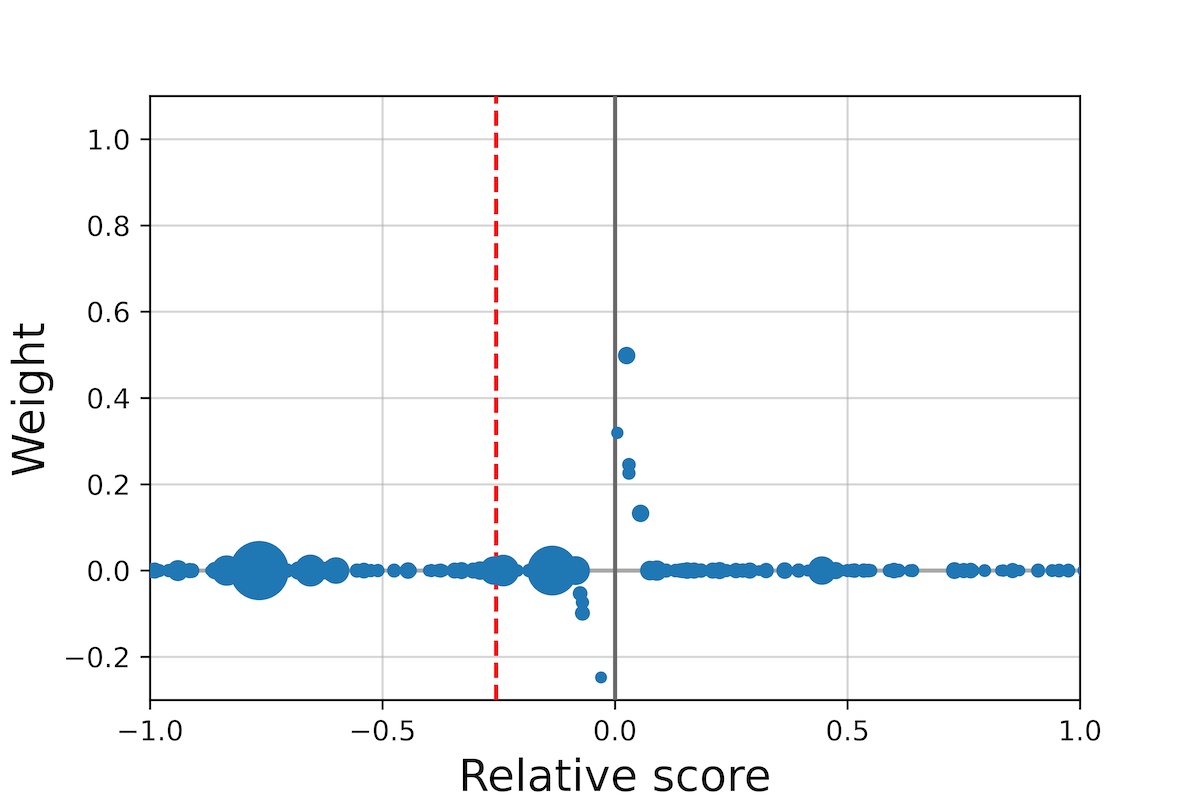}
			\vspace{-0.8em}
			\caption[]%
			{\footnotesize Minimax Affine MSE Under Constant Effect,\\ $~~~~~~~~~~~~~~~~~~~~~~~~~~~~~$$C=0.5$}
		\end{subfigure}
		\begin{subfigure}[b]{0.49\textwidth}   
			\centering 
			\includegraphics[width=\textwidth]{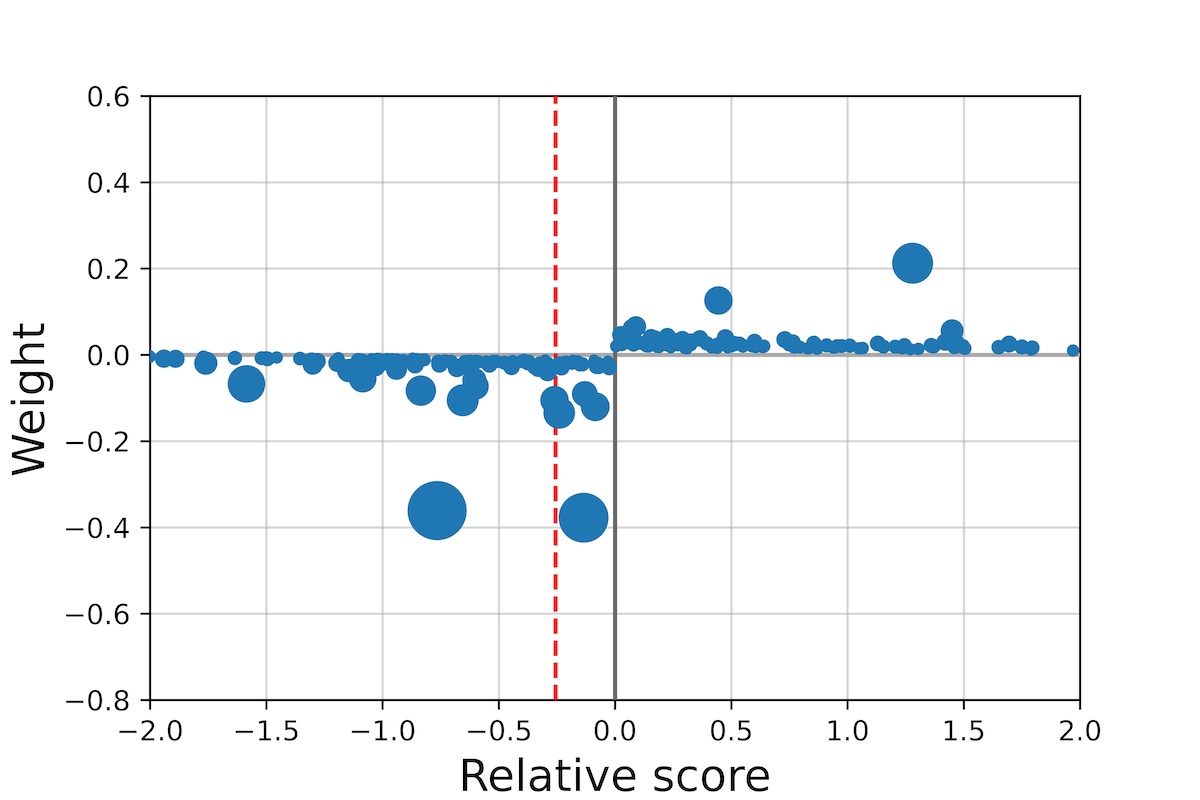}
			\caption[]%
			{\footnotesize Polynomial of Degree $2$} 
		\end{subfigure}
		\begin{subfigure}[b]{0.49\textwidth}   
			\centering 
			\includegraphics[width=\textwidth]{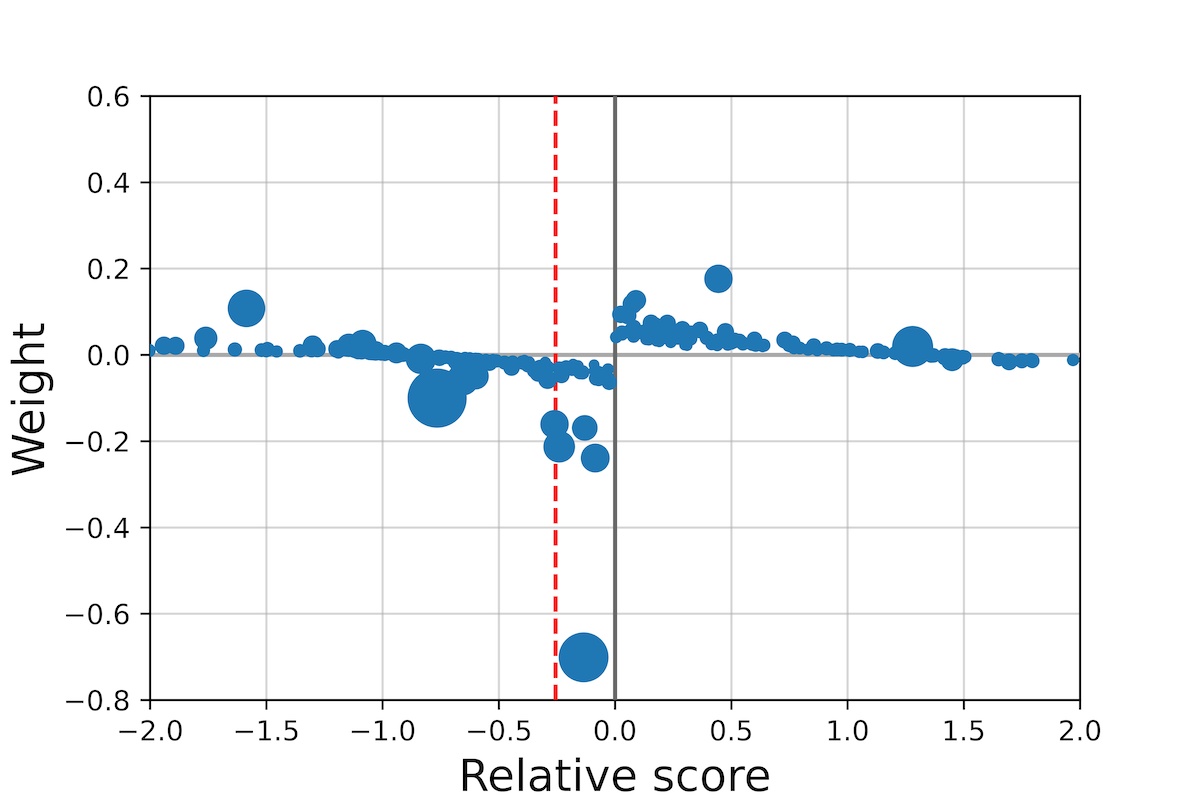}
			\caption[]%
			{\footnotesize Polynomial of Degree $4$} 
		\end{subfigure}
		
		\caption*{\scriptsize{\it Notes}: This figure shows the weight $w_i$ attached to each village by the plug-in decision rules of the form $\delta(\boldsymbol Y)=\mathbf{1}\{\sum_{i=1}^nw_iY_i\ge 0\}$.
			The weights are normalized so that $\sum_{i=1}^nw_i^2=1$.
			The horizontal axis indicates the relative score of each village.
			Each circle corresponds to each village.
			The size of circles is proportional to the inverse of the standard error of the enrollment rate $Y_i$.
			The vertical dashed line corresponds to the new cutoff $-0.256$.
			Panels (a) and (b) show results for the plug-in rules based on the minimax affine MSE estimators with or without the assumption of constant conditional treatment effects when the Lipschitz constant $C$ is 0.5.
			Panels (c) and (d) show results for the plug-in rules based on the polynomial regression estimators of degrees 2 and 4, respectively.}
		\label{fig:weight_plug}
	\end{figure}

\clearpage

	\begin{figure}[!t]
		\centering
		\caption{Optimal Decisions for Alternative New Policies}
		\begin{subfigure}[b]{0.49\textwidth}
			\centering 
			\includegraphics[width=\textwidth]{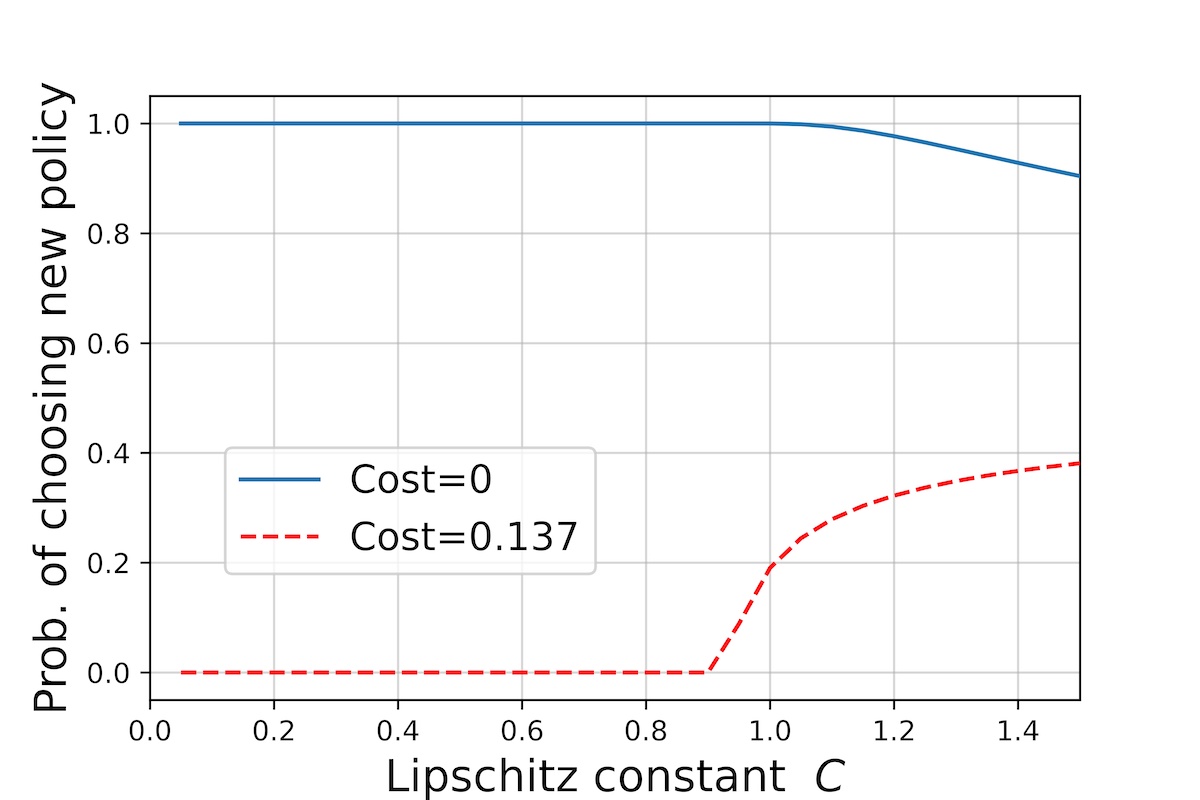}
            \vspace{-1em}
			\caption[]%
			{\footnotesize Constructing Schools in the Top 10\% of\\
            $~~~~~~$ Previously Ineligible Villages}
		\end{subfigure}
		\begin{subfigure}[b]{0.49\textwidth}   
			\centering 
			\includegraphics[width=\textwidth]{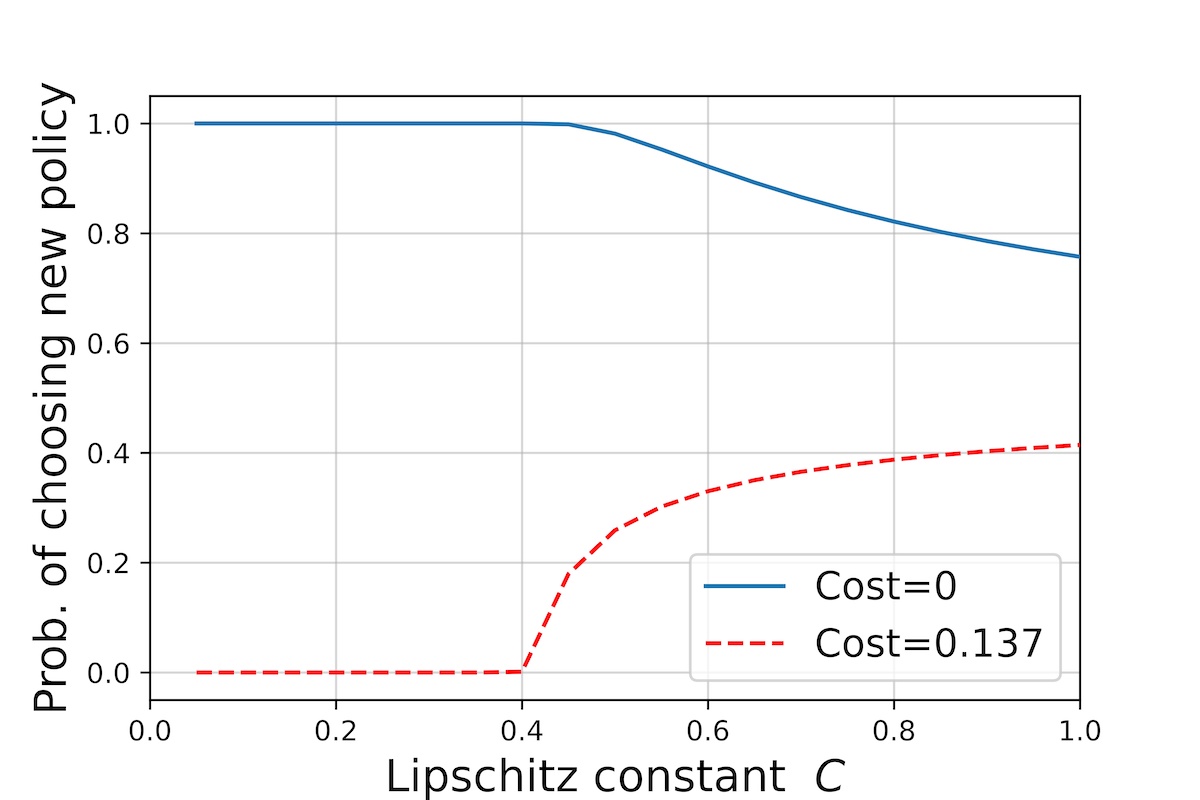}
            \vspace{-1em}
			\caption[]%
			{\footnotesize Constructing Schools in the Top 30\% of\\
            $~~~~~~$ Previously Ineligible Villages}
		\end{subfigure}
		\caption*{\scriptsize{\it Notes}: This figure shows the probability of choosing the new policy computed by the minimax regret rule.
			The new policy is to construct BRIGHT schools in previously ineligible villages whose relative scores are in the top 10\% (Panel (a)) or in the top 30\% (Panel (b)).
			The solid line shows results for the scenario in which we ignore the policy cost.
			The dashed line shows results for the scenario in which the policy cost measured in the unit of the enrollment rate is 0.137.
			I report results for the range $[0.05,0.1,...,1.45,1.5]$ of the Lipschitz constant $C$ in Panel (a) and for the range $[0.05,0.1,...,0.95,1]$ in Panel (b).
		}
		\label{fig:other_cutoffs}
	\end{figure}


\singlespacing
\bibliographystyleappendix{myecca}
\bibliographyappendix{reference}
\onehalfspacing

\end{document}